\documentclass[12pt]{article}
\input{myPreliminary}

\newcommand{\blind}{1}
\newcommand{\revise}{1}
\graphicspath{{./art/}}

\begin{document}

\if1\blind
{
	\title{\bf A General Framework For Constructing Locally Self-Normalized Multiple-Change-Point Tests}
	\author[1]{Cheuk Hin Cheng \thanks{Email: 
				\href{mailto:andychengcheukhin@link.cuhk.edu.hk}
				{\nolinkurl{andychengcheukhin@link.cuhk.edu.hk}}}}
	\author[2]{Kin Wai Chan \thanks{Email: 
				\href{mailto:kinwaichan@sta.cuhk.edu.hk}
				{\nolinkurl{kinwaichan@sta.cuhk.edu.hk}}}}
	\affil[1]{Department of Statistics, The Chinese University of Hong Kong.}		
		
  \maketitle
} \fi

\if0\blind
{
  \bigskip
  \bigskip
  \bigskip
  \begin{center}
    {\LARGE\bf A General Framework For Constructing Locally Self-Normalized Multiple-Change-Point Tests}
\end{center}
  \medskip
} \fi

\bigskip
\begin{abstract}
We propose a general framework to construct self-normalized multiple-change-point tests with time series data. 
The only building block is a user-specified one-change-point detecting statistic, 
which covers a wide class of popular 
methods, 
including cumulative sum process, outlier-robust rank statistics and order statistics. 
Neither robust and consistent estimation of nuisance parameters, selection of bandwidth parameters, nor pre-specification of the number of change points is required. The finite-sample performance shows that our proposal is size-accurate, 
robust against misspecification of the alternative hypothesis, and more powerful than existing methods. Case studies of NASDAQ option volume and Shanghai-Hong Kong Stock Connect turnover are provided.
\end{abstract}

\vfill

\newpage
\spacingset{1.5} 
\section{Introduction}
Testing structural stability is of paramount importance in statistical inference. There is considerable amount of literature and practical interest of testing existence of change points (CPs). 
\cite{csorgHo198820,csorgo1997limit} provided comprehensively both classical parametric and non-parametric approaches for the at most one change (AMOC) problem. For detection of multiple CPs, 
one method is to assume a fixed or bounded number of CPs ($m$). 
The test statistics are constructed by dividing the data into segments 
and maximizing a sum of AMOC-type statistics computed within each segment; see \cite{antoch2013testing} and \cite{shao2010testing}. 
Since the true number of CPs ($M$) is usually unknown in practice, 
several information criteria \citep{yao1988estimating,lee1995CPcriteria} have been proposed 
for estimating $M$. 
We remark that such approaches are vulnerable to misspecification of $M$
and the computational cost is often scaled exponentially with $m$. 
Another approach relies on sequential estimation via binary segmentation
\citep{vostrikova1981detecting,bai1998estimating}. 
This class of methods sequentially performs AMOC-type tests and produces CP estimates if the no-CP null hypothesis is rejected. 
The same procedure is repeated on each of the two subsamples split at the estimated CP
until no segment can be split further. 
Since the structural difference between data before and after a particular CP may be weakened 
in the presence of other CPs, 
the AMOC-type tests may inevitably lose power under the multiple-CP situation. 

To improve power, localization can be used. 
\cite{bauer1980extension} 
introduced the moving sum (MOSUM) approach.  
\cite{chu1995mosum} and \cite{eichinger2018mosum} later 
applied MOSUM to CPs detection. 
This method selects a local window size and performs inference on each local subsample. To avoid selection of the window size, \cite{fryzlewicz2014wild} proposed wild binary segmentation, which randomly draws subsamples without fixing the window size. Both algorithms attempt to find a subsample which contains only one CP to boost power without specifying $m$. However, the above methods require consistent and robust estimation of nuisance parameters. For example, the long-run variance (LRV) $\sigma^2$ of the partial sum usually appears in a test statistic
in the from of $T_n(\sigma^2) = L_n/\sigma^2$, 
where $L_n$ satisfies $L_n \inD \sigma^2\mathbb{L}$ under the null for some pivotal distribution $\mathbb{L}$, and ``$\inD$'' denotes convergence in distribution. Unfortunately, estimating $\sigma^2$ can be non-trivial and may require tuning a bandwidth parameter \citep{andrews1991heteroskedasticity}. Worse still, \cite{kiefer2005new} found that different selected bandwidths can noticeably influence the tests. 

To avoid the challenging estimation of $\sigma^2$, \cite{lobato2001testing} first proposed to use a self-normalizer (SN) $V_n$ that satisfies $V_n \inD \sigma^2\mathbb{V}$ for some pivotal $\mathbb{V}$ so that $V_n$ cancels out the nuisance $\sigma^2$ in $L_n$ to form a self-normalized test statistic 
$T_n' = L_n/V_n \inD \mathbb{L}/\mathbb{V}$ under the null.
\cite{shao2010testing} later modified the SN for handling the AMOC problem. In the multiple-CP problem, \cite{zhang2018unsupervised} proposed a self-normalized test which intuitively scans for the first and last CPs. However, it may sacrifice power when larger changes take place in the middle. Moreover, applying their proposed test to CP location estimation is non-trivial.

In this paper, we propose a general framework combining self-normalization and localization to construct a powerful multiple-CP test, which is applicable to a broad class of time series models including both short-range and long-range dependent data. Our method also allows utilization of outlier-robust change detecting statistics to improve size accuracy and power. The proposed method is driven by principles which generalize one-CP test statistic to a multiple-CP test statistic. 
Our proposed test is shown to be locally powerful asymptotically. Recursive formula is derived for fast computation.

The remaining part of the paper is structured as follows. We first review the self-normalization idea proposed by \cite{shao2010testing} under the AMOC problem in Section \ref{sec: SN intro}. Our proposal is demonstrated by using the CUSUM process as motivation in Section \ref{sec: LSN motivation + CUSUM}. Limiting distribution, consistency and local power analysis are also provided. In Section \ref{sec: general framework}, detailed principles and the most general form of our proposed test are presented with examples of applying rank and order statistics. The extensions to test for structural changes of other parameters, long range dependent and multivariate time series are also provided. 
In Section \ref{sec: Implementation}, 
the differences between different self-normalized
approaches will be comprehensively discussed. 
Some implantation issues are also discussed. 
Simulation experiments are presented in Section \ref{sec: simulation}, which shows promising size accuracy and substantial improvement of power over existing methods. The paper ends with real-data analysis of the NASDAQ option volume in Section \ref{sec: nasdaq call vol} and the Shanghai-Hong Kong Stock Connect turnover in Section \ref{sec: SHHK stock connect}. Proofs of main results, additional simulation results, critical values, recursive formulas and algorithms can be found in the Appendix.

\section{Introduction of self-normalization}
\label{sec: SN intro}
Consider the signal-plus-noise model: $X_i = \mu_i + Z_i$, where $\mu_i=\E(X_i)$ for $i=1, \ldots,n$ and $\{ Z_i : i\in\mathbb{Z}\}$ is a stationary noise sequence. The AMOC
testing problem is formulated as follows:
\begin{align} 
H_0&: \mu_1 = \cdots = \mu_n, \label{eqt:H0}\\
H_1&: \exists\; 0<k_1<n, \quad \mu_1 = \cdots =\mu_{k_1} \neq \mu_{k_1+1}= \cdots = \mu_n.  \label{eqt:H1_oneCP}
\end{align}
Let $\xi_{k_1,k_2} = \sum_{i = k_1}^{k_2} X_i$ if $1 \leq k_1 \leq k_2 \leq n$; $\xi_{k_1,k_2} = 0$ otherwise. Define the CUSUM process
\begin{align}
\label{EQ Global CUSUM}
	\C_n(\lfloor nt \rfloor) = n^{-1/2}\sum_{i = 1}^{\lfloor nt \rfloor}\left( X_i - n^{-1}\xi_{1,n}\right), 
	\qquad t \in [0,1], 
\end{align}
where $\lfloor nt \rfloor$ is the largest integer part of $nt$. The limiting distribution of functional of \eqref{EQ Global CUSUM} relies on the following assumption.

\begin{assumption}
\label{as:invariance principle}
As $n\rightarrow\infty$, 
$\{ n^{-1/2}\sum_{i = 1}^{\lfloor nt \rfloor} Z_i: t\in [0,1]   \} \Rightarrow   \{\sigma \BM(t): t\in [0,1]   \},
$ where $\sigma^2 = \lim_{n \rightarrow \infty} \Var(\xi_{1,n})/n \in(0,\infty)$ is the long-run variance (LRV), $\BM(\cdot)$ is the standard Brownian motion and ``$\Rightarrow$'' denotes convergence in distribution in the Skorokhod space \citep{billingsley1999convergence}.
\end{assumption}
Assumption \ref{as:invariance principle} is known as a functional central limit theorem (FCLT) or Donsker's invariance principle. Under standard regularity conditions (RCs), Assumption \ref{as:invariance principle} is satisfied. For example, \cite{herrndorf1984functional} proved the FCLT for the dependent data under some mixing conditions; \cite{wu2007strong} 
later showed the strong convergence for stationary processes using physical and predictive dependence measures.
Under the null hypothesis $H_0$, by 
continuous mapping theorem, we have 
$\{\C_n(\lfloor nt \rfloor) : t\in[0,1]\} \Rightarrow \{\sigma \left\{\BM(t) - t\BM(1)\right\}: t\in[0,1]\} $.

Classically, the celebrated Kolmogorov--Smirnov (KS) test statistic, defined as $\KS_n(\sigma) = \sup_{k=1, \ldots, n} \left|\C_n(k)/\sigma \right|$, can be used for the AMOC problem. Since $\sigma$ is typically unknown, we need to estimate it by a consistent $\widehat{\sigma}$
no matter under $H_0$ or $H_1$; see \cite{chan2020mean} and \cite{chan2022_dlrv} for some possible estimators. Hence, $\KS_n(\widehat{\sigma}) \inD \mathbb{KS} := \sup_{t\in(0,1)}| \BM(t) - t\BM(1)|$, which is known as the Kolmogorov distribution. \cite{shao2010testing} proposed to bypass the estimation of $\sigma$ by normalizing $\C_n(k)$ by a non-degenerate standardizing random process called a self normalizer:
\begin{equation*}
\V_n(k) = n^{-2}\left\{\sum_{i = 1}^k \left(\xi_{1,i} - \frac{i}{k} \xi_{1,k} \right)^2  + \sum_{i = k+1}^n \left(\xi_{i,n} - \frac{n-i+1}{n-k} \xi_{k+1,n} \right)^2 \right \}, \quad k = 1,\ldots, n - 1.
\end{equation*}
The resulting self-normalized test statistic for the AMOC problem is $\shao_n^{(1)}$, where 
\begin{equation}
\label{EQ: shao 1 ts}
\begin{split}
\shao_n^{(1)} =  \sup_{k = 1,\ldots, n-1} \shao_n^{(1)}(k)
\qquad \text{and} \qquad
\shao_n^{(1)}(k) = \C_n^2(k)/\V_n(k).
\end{split}
\end{equation}
Under Assumption \ref{as:invariance principle} and $H_0$, the limiting distribution of $\shao_n^{(1)}$ is non-degenerate and pivotal. 
The nuisance parameter $\sigma^2$ is asymptotically cancelled out in 
the numerator $\C_n^2(k)$ and the denominator $\V_n(k)$ on the right hand side of (\ref{EQ: shao 1 ts}).
Moreover, since there is no CP in the intervals $[1, k_1]$ and $[k_1+1, n]$ under $H_1$, the SN at the true CP, $\V_n(k_1)$, is invariant to the change and therefore their proposed test does not suffer from the well-known non-monotonic power problem; see, e.g., \citet{vogelsang1999}. 
However, this appealing feature no longer exists in the multiple-CP setting. 
\cite{zhang2018unsupervised} improved by proposing a self-normalized multiple-CP test without specifying $m$. The test utilizes forward and backward scans to select two time intervals, 
$[1,k_1]$ and $[k_2,n]$ for $1 < k_1, k_2 < n$,
and aims at detecting the largest change. 
Applying the SN in \cite{shao2010testing} to each of these two subsamples, 
their test was proved to be consistent. 
However, the method tends to consider only the first and the last CPs. 
Thus, it affects power as will be shown in our simulation studies in Section \ref{sec: simulation}. 
In the next section, we propose a framework for 
extending a one-CP test to a valid multiple-CP self-normalized test.
It is achieved by utilizing the localization idea.  

\section{Locally self-normalized statistic}
\label{sec: LSN motivation + CUSUM}
We consider the multiple-CP testing problem, i.e., to test the $H_0$ in (\ref{eqt:H0}) 
against 
\begin{equation}
\label{EQ: multiple CP hypothesis test}
H_{\geq 1}: \mu_1 = \cdots = \mu_{k_1} \neq  \mu_{k_1+1} = \cdots = \mu_{k_2} \neq \cdots \neq \mu_{k_M + 1} = \cdots = \mu_n,
\end{equation}
for an unknown number of CPs $M \geq 1$, and unknown times 
$1<k_1<\cdots< k_{M}<n$.
Also denote $\pi_j = k_j/n$ and 
$\Delta_j = \mu_{k_j+1} - \mu_{k_j}$ be the $j$th relative CP time and the corresponding mean change, respectively, for $j=1, \ldots, M$.
Since $M$ is usually unknown and possibly larger than one in practice, the multiple-change alternative (\ref{EQ: multiple CP hypothesis test}) is more sensible than 
the one-change alternative (\ref{eqt:H1_oneCP}). 
To improve the power of \eqref{EQ Global CUSUM}, we generalize the CUSUM process to a localized CUSUM statistic. Define it and its SN as
{\small
\begin{align}
	\L^{(\C)}_n(k\mid s,e) 
		&:= (e-s+1)^{-1/2}\sum_{i = s}^k\left\{ X_i - (e-s+1)^{-1}\xi_{s,e}\right\} \nonumber\\
		&= \left(\frac{n}{e-s+1}\right)^{1/2} \left[ \C_n(k) - \C_n(s-1) - \frac{k-s+1}{e-s+1}\left \{ \C_n(e) - \C_n(s-1)\right \} \right] , \label{EQ local CUSUM}\\
	\V^{(\C)}_n(k \mid s, e) 
		&:= \frac{k-s+1}{(e-s+1)^{2}}\sum_{j = s}^k \L^{(\C)}_n(j \mid s, k)^2
			+ \frac{e-k}{(e-s+1)^{2}}\sum_{j = k+1}^e \L^{(\C)}_n(j \mid k+1, e)^2, \label{EQ: CUSUM self normalizer}
\end{align}
}respectively, where $1\leq s \leq k \leq e\leq n$. 
The indices $s$ and $e$ denote the starting and ending times of 
a local subsample $\{X_s, \ldots, X_e\}$, respectively.
The statistic $\L^{(\C)}_n(k\mid s,e)$ is used to detect a local difference between $\{X_s,\ldots, X_k\}$ and $\{X_{k+1}, \ldots, X_e\}$. Indeed, (\ref{EQ local CUSUM}) generalizes the global CUSUM process (\ref{EQ Global CUSUM}) because $\C_n(k) \equiv \L^{(\C)}_n(k\mid 1,n)$. 
Moreover, according to (\ref{EQ local CUSUM}), the localized CUSUM process allows fast recursive computation as it is a function of $\C_n(\cdot)$. Consequently, the \emph{locally self-normalized} (LSN) statistic is defined as 
\begin{equation}
\label{EQ: LSN stat CUSUM}
\begin{split}
	\T^{(\C)}_n(k \mid s, e) := \frac{\L^{(\C)}_n(k \mid s,e)^2}{\V^{(\C)}_n(k \mid s, e)} , 
	\qquad \text{for}\;1\leq s \leq k \leq e\leq n,
\end{split}
\end{equation}
and $\T^{(\C)}_n(k \mid s, e) = 0$ otherwise; 
see Remark \ref{rem:nonSN} for a non-SN approach. 
The LSN statistic generalizes \cite{shao2010testing}'s proposal in the sense that $\shao_n^{(1)}(k) \equiv \T^{(\C)}_n(k \mid 1, n)$. 
To infer how likely $k$ is a potential CP location, we consider all the symmetric local windows that center at $k$.
The resulting score function is
\begin{equation}
\label{EQ: CUSUM test score function}
	\T^{(\C)}_n(k) = \sup_{\lfloor n\epsilon  \rfloor \leq d \leq n} \T^{(\C)}_n(k \mid k-d, k+1+d),
\end{equation}
where $0<\epsilon<1/2$ is a fixed local tuning parameter, which is similarly used in, e.g., \cite{huang2015self} and \cite{zhang2018unsupervised};
see Remark \ref{rem:symmetricWindow} for discussing the choice of symmetric windows. 
In particular, they suggested to choose $\epsilon=0.1$. This suggestion is followed throughout this article. In essence, (\ref{EQ: LSN stat CUSUM}) compares 
the local subsamples of length $d+1$ before and after time $k$, i.e.,  
$\mathcal{S}_{\text{before}}=\{X_{k-d}, \ldots, X_{k}\}$ and 
$\mathcal{S}_{\text{after}}=\{X_{k+1}, \ldots, X_{k+1+d}\}$,
for each possible width $d$ that is not too small. The statistic $\T^{(\C)}_n(k)$ records the change corresponding to the most obvious $d$. We visualize our proposed score function $\T_n^{(\C)}(\cdot)$ 
and \cite{shao2010testing} proposal $\shao_n^{(1)}(\cdot)$ defined in \eqref{EQ: shao 1 ts} in Figure \ref{fig: compare pt est} when they are applied to a time series with 5 CPs. 
\cite{zhang2018unsupervised}'s method is not included because it does not map a time $k$ 
to a score. 
Clearly, $\T_n^{(\C)}(\cdot)$ attains the local maxima near all true CP locations, while $\shao_n^{(1)}(\cdot)$ cannot. Possible reasons are reduced CP detection capacity of global CUSUM by other changes, non-monotonic CP configuration, and non-robust SN.
By design, the score function can be extended to CP estimation; see Remark \ref{rem: pt est}.

\begin{figure}
\centering
\small
\includegraphics[width=0.95\linewidth]{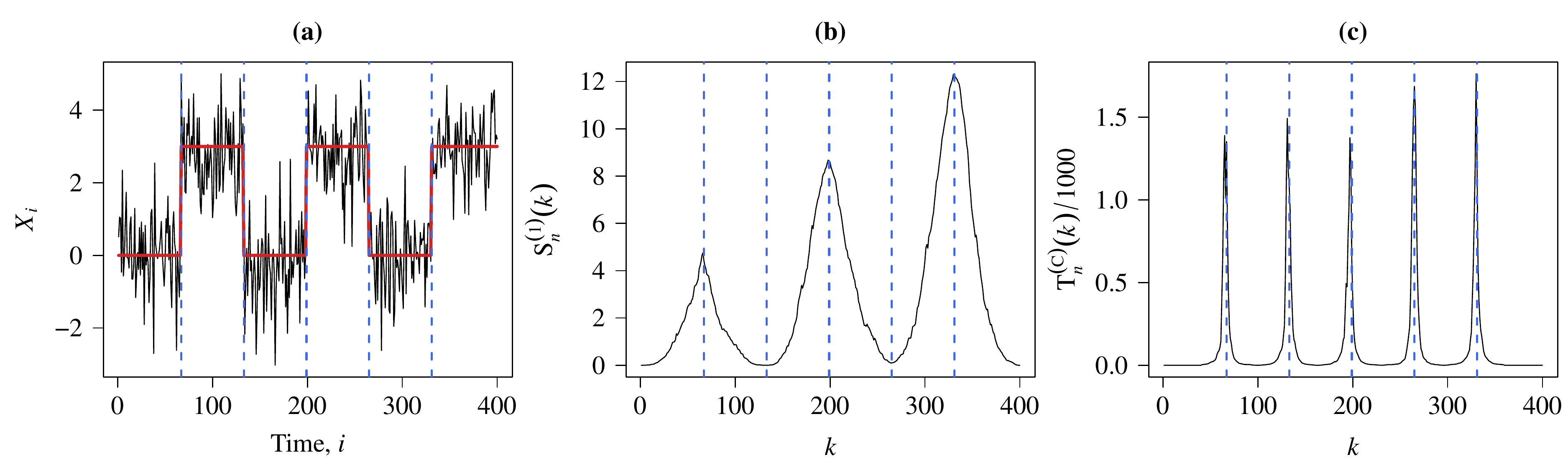}
\vspace{-0.3cm}
\caption{\footnotesize \label{fig: compare pt est}
(a) A realization of time series (black solid line) with 5 CPs (blue dashed lines) and its mean function (red solid line).
(b) The score function $\shao_n^{(1)}(k)$ employed by \cite{shao2010testing}. 
(c) The proposed score function $\T^{(\C)}_n(k)$.} 
\end{figure}

Finally, aggregating the scores $\T^{(\C)}_n(k)$ across $k$, 
we obtain the proposed test statistic:
\begin{equation}
\label{EQ: CUSUM ts}
\begin{split}
	\T^{(\C)}_n := \frac{1}{n - 2\lfloor \epsilon n \rfloor - 1}\sum_{k= \lfloor \epsilon n \rfloor + 1}^{n-\lfloor \epsilon n \rfloor - 1} \T^{(\C)}_n(k),
\end{split}
\end{equation}
which captures the effects of all potential CPs instead of 
merely the first and the last CPs as in \cite{zhang2018unsupervised}. Hence, this statistic is potentially more powerful for testing multiple CPs than the existing state-of-the-art SN test \citep{zhang2018unsupervised}. Define $0<\pi_1< \cdots<\pi_M< 1$ to be $M \geq 1$ relative CP times
so the CP occurs at times $\lfloor n\pi_1 \rfloor, \ldots, \lfloor n\pi_M \rfloor$.
For convenience, denote $\pi_0 = 0$ and $\pi_{M+1} = 1$. Theorem \ref{thm: CUSUM test} below states the limiting distribution of $\T_n^{(\C)}$ and the consistency of the test.

\begin{theorem}[Limiting distribution and consistency]
\label{thm: CUSUM test}
Suppose Assumption \ref{as:invariance principle} is satisfied. 
(i) Under $H_0$, 
we have
\begin{equation}
\label{EQ: CUSUM ts limiting dist} 
	\T^{(\C)}_n \inD \mathbb{T}:= (1-2\epsilon)^{-1} \int_{\epsilon}^{1-\epsilon} 
			\sup_{\delta > \epsilon} \frac{\mathbb{L}(t \mid t-\delta, t+\delta)^2}{\mathbb{V}(t\mid t-\delta,t+\delta)} \, \dd t,
\end{equation}
for any $0<\epsilon<1/2$,
where 
\begin{align}
	\mathbb{V}(t\mid t-\delta,t+\delta) 
		&= 	\frac{1}{4\delta} \left\{\int_{t- \delta}^{t} \mathbb{L}(\tau \mid t-\delta , t)^2 \dd \tau
			+ \int_{t}^{t+\delta} \mathbb{L}(\tau \mid t , t+\delta)^2 \dd \tau  \right\}, 
			\label{eqt:defVlimit}\\
	\mathbb{L}(t \mid \tau_1, \tau_2) 
		&= \frac{1}{\sqrt{\tau_2-\tau_1}} \left[ \BM(t) - \BM(\tau_1) - \frac{t-\tau_1}{\tau_2-\tau_1}\left\{\BM(\tau_2) - \BM(\tau_1) \right\} \right]. \label{eqt:defLlimit}
\end{align}
(ii) Under $H_{\geq 1}$, if there exist $j \in\{ 1, \ldots, M\}$ and $C>0$ such that
the $j$th CP time satisfies 
$\epsilon < \min(\pi_{j} - \pi_{j-1},\pi_{j+1} - \pi_{j})$ and the $j$th change magnitude satisfies 
$|\Delta_j|=|\mu_{k_j+1}-\mu_{k_j}|>C$ as $n\rightarrow\infty$, 
then $\lim_{n \rightarrow \infty}\T^{(\C)}_n = \infty$ in probability.
\end{theorem}

Under $H_0$, the test statistic $\T_n^{(\C)}$ has a pivotal limiting distribution, whose quantiles can be easily found by simulation; 
see Table \ref{tab: 95 Critical Values} in the appendix and Section \ref{sec: Implementation} for more discussion. Under $H_{\geq 1}$, the test is of power 1 asymptotically when the change magnitude is at least a non-zero constant. 
Moreover, Assumption \ref{as:invariance principle} is satisfied by short-range dependent data. It can be generalized to handle long-range dependent data; see Example \ref{exp: LRD}.
Theorem \ref{thm:localH1} below concerns a local alternative hypothesis that contains CP having diminishing change magnitude. 
It states the regime of the size of CP at which our test remains powerful. 

\begin{theorem}[Local limiting power]\label{thm:localH1}
Under $H_{\geq 1}$, parametrize the $j$th change magnitude $|\Delta_j(n)| = |\mu_{k_j+1}-\mu_{k_j}|$ as a function of $n$.
If there exist $j \in\{ 1, \ldots, M\}$ and $0.5 <\kappa\leq 1$ such that
the $j$th CP satisfies 
$\epsilon < \min(\pi_{j} - \pi_{j-1},\pi_{j+1} - \pi_{j})$ and 
$|\Delta_j(n)| \asymp n^{\kappa -1}$ as $n\rightarrow\infty$, 
then 
$\lim_{n \rightarrow \infty}\T^{(\C)}_n = \infty$ in probability.
\end{theorem}

The smallest change magnitude so that the test remains powerful satisfies $\kappa > 0.5$. This result is in line with \cite{zhang2018unsupervised}. Indeed, from the proof of Theorem \ref{thm:localH1}, 
it is easy to see that CP tests based on the CUSUM process are of power one asymptotically only if the order of change magnitude is larger than $n^{-1/2}$. Therefore, our proposed test preserves the change detecting ability of the CUSUM process.

\begin{remark}[Non-SN approach]\label{rem:nonSN}
Users may instead use a consistent estimator of the LRV, say $\widehat{\sigma}^2_n$, to replace the SN \eqref{EQ: CUSUM self normalizer}, $\V_n^{(\C)}(k \mid k-d, k+1+d)$, for all $k$ and $d$. 
The performance is expected to be affected by the robustness of CP estimators (if any) and bandwidth parameter as discussed in \cite{kiefer2005new}. 
Alternatively, the LRV can be estimated robustly by splitting the dataset as in (\ref{EQ: CUSUM self normalizer}). 
However, the estimator of $\sigma^2$ will become very noisy due to small subsample sizes. 
This extra variability in estimating $\sigma^2$ may accumulate and eventually ruin the test statistic. 
Our approach takes the variability of $\V_n^{(\C)}(k\mid k-d,k+1+d)$ into account, thus, it is expected to perform better. 
\end{remark}

\begin{remark}[Symmetric windows]\label{rem:symmetricWindow}
Define the score function based on non-symmetric windows as 
\begin{equation}
\label{EQ: LSN stat nonsymm window}
\tilde{\T}_n^{(\C)}(k) = \sup_{\lfloor \epsilon n \rfloor \leq d_0, d_1 \leq n} \T_n^{(\C)}(k \mid k - d_0, k+1+d_1).
\end{equation}
The corresponding test statistic is
$
\tilde{\T}_n^{(\C)} = \sum_{k = \lfloor \epsilon n \rfloor +1}^{n - \lfloor \epsilon n \rfloor - 1} \tilde{\T}_n^{(\C)}(k)/(n-2\lfloor \epsilon n \rfloor - 1)
$. To balance computational cost and power, we suggest using symmetric windows $\T_n^{(\C)}(k)$ instead of $\tilde{\T}_n^{(\C)}(k)$, and set the local subsamples before and after target time $k$, 
i.e., $\mathcal{S}_{\text{before}}$ and $\mathcal{S}_{\text{after}}$, to have the same width $d$. 
The test consistency only requires existence of one CP in the local windows. The change can still be captured under symmetric windows without significant power loss. 
Using non-symmetric windows significantly increases computational cost 
and is not expected to bring vast improvement in power. 
Some simulation evidence is provided in Sections \ref{sub-sec: test comparison discussion} and \ref{sec: justify symm windows}. 
\end{remark}

\begin{remark}[CP time estimation]\label{rem: pt est}
The score function $\T^{(\C)}_n(\cdot)$ in (\ref{EQ: CUSUM test score function}) 
can be used for CP estimation; see Figure \ref{fig: compare pt est} (c) for visualization. Motivated by the GOALS algorithm \citep{jiang2021CP}, the set of estimated CP times can be obtained by finding the local maximizers of the score functions:
\begin{equation}
\label{EQ: cp set pt estimate}
\small
\widehat{\Pi} = \left \{k \in \{ \lfloor n\epsilon \rfloor, \ldots, n-\lfloor n\epsilon \rfloor \}: \T^{(\C)}_n(k) = \max_{k-\lfloor n\epsilon \rfloor <j \leq k+\lfloor n\epsilon \rfloor}\T^{(\C)}_n(j) \text{ and } \T^{(\C)}_n(k) > \varrho \right \},
\end{equation}
where $\varrho$ is the decision threshold that is of order $o(n)$. The estimation may be improved through integrating other screening algorithms to perform model selection in order to avoid overestimating the number of CPs. For example, the screening and ranking algorithm (SaRa) \citep{niu2012screening} and the scanning procedure proposed by \cite{yau2016inference}. Alternatively, it is possible to apply our proposed test to stepwise estimation, for example, binary segmentation (BS) \citep{vostrikova1981detecting}, wild binary segmentation \citep{fryzlewicz2014wild}, etc; 
see Section \ref{sec: pt estimation} in the appendix for the detail of integrating our proposed scanning procedure with SaRa and BS.
\end{remark}

\section{General framework for constructing LSN statistics}
\label{sec: general framework}
The proposed test statistic $\T^{(\C)}_n$ in (\ref{EQ: CUSUM ts}) sheds light on how multiple-CP tests should be performed. 
In particular, we combine the advantages of localization and self-normalization.  
Although $\T^{(\C)}_n$ demonstrates to have appealing features as we see in Figure \ref{fig: compare pt est}, 
it is not fully general in terms of two aspects to be stated below. 
In this section, we lay out the underlying principles 
to define a general class of statistics for testing multiple CPs.

First, $\T^{(\C)}_n$ is a functional of the global CUSUM process $\C_n(\cdot)$. 
Thus, potential CPs are detected via differences between local averages of $X_i$'s, 
which however may not be the best choice, in particular for the data that have heavy-tailed distributions. 
Generally, one may consider any specific \emph{global change detecting process} 
$\{\D_n(\lfloor nt \rfloor) \}_{0 \leq t \leq 1}$ for the AMOC problem. 
Applying the localization trick (\ref{EQ local CUSUM}) and 
the self-normalization trick (\ref{EQ: CUSUM self normalizer}) 
to $\D_n(\cdot)$ instead of $\C_n(\cdot)$, we obtain 
the generalized LSN statistic $\T^{(\D)}_n(k\mid s,e) = \L^{(\D)}_n(k\mid s,e)^2/\V^{(\D)}_n(k\mid s,e)$, 
where 
\begin{align*}
	\L_n^{(\D)}(k \mid s, e) 
		&= \left(\frac{n}{e-s+1} \right)^{1/2} \left[ \D_n(k) - \D_n(s-1) - \frac{k-s+1}{e-s+1}\left\{ \D_n(e) - \D_n(s-1) \right\} \right],\\
	\V_n^{(\D)}(k \mid s, e) 
		&= \frac{k-s+1}{(e-s+1)^{2}}\sum_{j = s}^k \L_n^{(\D)}(j \mid s, k)^2
 		+ \frac{e-k}{(e-s+1)^{2}}\sum_{j = k+1}^e \L_n^{(\D)}(j \mid k+1, e)^2.
\end{align*}
If $\D_n(\cdot)$ converges weakly in Skorokhod space to $\sigma_{\D} \mathbb{D}(\cdot)$
for some $\sigma_{\D}\in(0, \infty)$,  where $\mathbb{D}(\cdot)$ is 
an empirical process whose distribution is pivotal, 
then $\T^{(\D)}_n(k\mid d)$ is asymptotically pivotal as in Theorem \ref{thm: CUSUM test}.
Examples of $\D_n(\cdot)$, based on Wilcoxon statistics, Hodges--Lehmann statistics, and influence functions, are provided in Examples \ref{exp: Wilcoxon}--\ref{exp: Influence functions}. Extensions to long-range dependent and multivariate time series are provided in Examples \ref{exp: LRD} and \ref{exp: multivariate time series}, respectively.

Second, there are two free indices, $d$ and $k$, in $\T^{(\D)}_n(k\mid k-d,k+1+d)$.
The maximum operator in (\ref{EQ: CUSUM test score function}) 
and the mean operator in (\ref{EQ: CUSUM ts}) are used to aggregate the indices $d$ and $k$, 
respectively. The introduction of $d$ is to let the data choose the window sizes. We remark that the performance of the MOSUM statistic, which is one special case of localized CUSUM with a fixed window size, depends critically on the window size \citep{eichinger2018mosum}. 
Hence, they suggested to consider multiple window sizes for robustness.  
We follow their suggestion to improve performance by aggregating all possible window sizes.  
Generally, users may choose their preferred aggregation operators,
e.g., median, trimmed mean, etc, for aggregation.
Formally, for any finite set of indices $S \subset \mathbb{Z}$ having size $|S|$, 
we say that $\mathcal{P}$ is an aggregation operator if 
$\mathcal{P}_{i\in S}a_i$ maps 
$\{a_i \in\mathbb{R}: i\in S\}$ to a real number, 
e.g., 
if $\mathcal{P}=\max$, $\mathcal{P}_{i\in S}a_i = \max_{i\in S} a_i$ gives the maximum; and  
if $\mathcal{P}=\mean$, $\mathcal{P}_{i\in S}a_i = \sum_{i\in S}a_i/|S|$ gives the sample mean. Hence, the most general form of our proposed test statistic is 
\begin{align*}
	\T_n^{(\D)}
		= \T_n^{(\D)}[\mathcal{P}^{(1)},\mathcal{P}^{(2)}]
		= \mathcal{P}^{(2)}_{\lfloor \epsilon n \rfloor  < k < n-\lfloor \epsilon n \rfloor}
			\left[ \mathcal{P}^{(1)}_{\lfloor n\epsilon  \rfloor \leq d \leq n} \left\{ \T^{(\D)}_n(k\mid k-d,k+1+d) \right\} \right],
\end{align*}
where $\mathcal{P}^{(1)}$ and $\mathcal{P}^{(2)}$ are some aggregation operators. 
In particular, $\T_n^{(\C)}[\max, \mean]$ is used in (\ref{EQ: CUSUM ts}).
It is practically convenient to use weighted sum and weighted average, i.e., 
\begin{align*}
	\mathcal{P}^{(1)}_{\lfloor n\epsilon  \rfloor \leq d \leq n}  a_d
		= \max_{\lfloor n\epsilon  \rfloor \leq d \leq n}  w_d^{(1)} a_d
	\qquad \text{and} \qquad
	\mathcal{P}^{(2)}_{\lfloor \epsilon n \rfloor  < k < n-\lfloor \epsilon n \rfloor} a_k
		= \frac{\sum_{\lfloor \epsilon n \rfloor  < k < n-\lfloor \epsilon n \rfloor}  w_k^{(2)} a_k}{\sum_{\lfloor \epsilon n \rfloor  < k < n-\lfloor \epsilon n \rfloor}  w_k^{(2)}},
\end{align*}
where $w_d^{(1)}$'s and $w_k^{(2)}$'s are some constants specified by users. 
The weights guarantee the local intervals are not too small 
in order to avoid degeneracy of the SN.
They are also used to reflect user's prior beliefs on $k$ and $d$. 
We propose to use $\mathcal{P}^{(2)}=\mean$ instead of $\mathcal{P}^{(2)}=\max$
for better empirical performance.
By construction, 
our proposed LSN statistics remain large over the neighborhoods of the true CPs. Thus, the averaged score capture evidence contributed by all CPs and their neighborhoods, 
whereas the maximum score can only capture evidence contributed by one single CP 
corresponding to the highest score. 
In this sense, $\mathcal{P}^{(2)}=\mean$ is more suitable for multiple-CP testing.

Our approach is more flexible as it supports 
a general CP detecting process $\D_n(\cdot)$ and general aggregation operators $\mathcal{P}^{(1)}$ and $\mathcal{P}^{(2)}$, 
Moreover, our proposed LSN statistic is a function of the supplied global CP detecting process. Therefore, recursive computation of the LSN statistics, regardless of the supplied global change detecting process, is possible for our approach. 
It results in lower computational cost;
see Section \ref{sec: recursive computational} in the appendix. 
In a nutshell, our framework is an automatic procedure for generalizing any one-CP detecting statistic $\D_n(\cdot)$ to a multiple-CP detecting statistic $\T_n^{(\D)}$. We demonstrate it via following examples.

\begin{example}[Wilcoxon rank statistics]\label{exp: Wilcoxon}
An outlier-robust non-parametric CP test 
can be constructed by using the Wilcoxon two-sample statistic. The corresponding global change detecting process is 
\begin{equation}
\label{EQ wilcoxon global}
\begin{split}
	\W_n(\lfloor nt \rfloor) 
		=  \frac{1}{n^{3/2}}\sum_{i = 1}^{\lfloor nt \rfloor} 
			\sum_{j = \lfloor nt \rfloor +1}^n \left( \mathbb{1}_{\{ X_i \leq X_j \}} - \frac{1}{2} \right)
		= n^{-3/2}\left(\sum_{i = 1}^{ \lfloor nt \rfloor} R_i - \frac{ \lfloor nt \rfloor}{n}\sum_{i' = 1}^n R_{i'} \right),
\end{split}
\end{equation}
for $t \in [0,1]$, 
where $R_i = \sum_{j = 1}^n \mathbb{1}_{\{X_j \leq X_i \}}$ is the rank of $X_i$ 
when there is no tie in $X_1, \ldots, X_n$. 
Assuming that the data is weakly dependent and can be represented as a functional of an absolutely regular process. 
Under some mixing conditions, 
we have
$\W_n(\lfloor nt \rfloor) \inD \sigma_{\W}\left\{\BM(t)- t\BM(1) \right\}$, 
where 
$\sigma_{\W}^2 = \sum_{k \in\mathbb{Z}} \Cov\left \{F(X_0), F(X_k) \right \}$ 
and $F(\cdot)$ is the distribution function of $X_1$ with a bounded density; 
see Theorem 3.1 in \cite{dehling2013SRD}. Using the principles in Section \ref{sec: general framework}, we may use \eqref{EQ wilcoxon global} to construct a self-normalized multiple-CP Wilcoxon test and eliminate the nuisance parameter $\sigma_{\W}$
by using the test statistic $\T_n^{\W} := \T_n^{\W}[\max, \mean]$, 
whose limiting distribution is stated as follows. 

\begin{corollary}[Limiting distribution of $\T_n^{\W}$]\label{thm: wilcoxon test}
Under the RCs in Theorem 3.1 of \cite{dehling2013SRD} and $H_0$, $\T^{(\W)}_n \inD \mathbb{T}$.
\end{corollary}
\end{example}

\begin{example}[Hodges--Lehmann order statistics]\label{exp: Hodges-Lehmann}
Hodges--Lehmann statistic is another popular alternative to the CUSUM statistic in (\ref{EQ Global CUSUM}). Its global change detecting process is 
\begin{equation}
\label{EQ: HL global}
\begin{split}
	\H_n(\lfloor nt \rfloor) = n^{-3/2}\lfloor nt \rfloor(n-\lfloor nt \rfloor)\median\left\{ (X_i - X_j): 1 \leq i \leq \lfloor nt \rfloor < j \leq n \right\},
\end{split}
\end{equation}
for $t\in[0,1]$,
where $\median(\mathcal{S})$ denotes the sample median of a set $\mathcal{S}$. 
It has a better performance than the CUSUM test under skewed and heavy-tailed distributions. Under the regularity conditions in \cite{dehling2015robust}, we have 
$H_n(\lfloor nt \rfloor) \inD \sigma_{\H} u(0)^{-1} \{\BM(t) - t\BM(1)\}$, where 
$u(x) = \int_{\mathbb{R}} f(y)f(x+y) \dd y$,
$\sigma_{\H}^2 = \sum_{k\in\mathbb{Z}} \Cov \{F(X_0), F(X_k) \}$, 
$f$ is the density of $X_1$, and 
$F$ is the distribution function of $X_1$. Similarly, we can apply the Hodges--Lehmann statistic to test for multiple CPs. The test statistic  is $\T^{(\H)}_n= \T^{(\H)}_n[\max, \mean]$,
whose limiting distribution is stated below. To our best knowledge, there is no existing self-normalized test that uses the Hodges--Lehmann statistic. It is included for detecting CP in heavy-tailed data and showing the generality of our proposed self-normalization framework.
\begin{corollary}[Limiting distribution of $\T^{(\H)}_n$]
\label{thm: HL test}
Under the RCs in Theorem 1 of \cite{dehling2015robust} and $H_0$, $\T^{(\H)}_n \inD \mathbb{T}$.
\end{corollary}
\end{example}

\begin{example}[Influence functions for testing general parameters]\label{exp: Influence functions}
Instead of testing changes in mean, one may be interested in other quantities,
e.g., variances, quantiles and model parameters. 
Let $\theta_i = P(F_i^{(h)})$, where $h\in\mathbb{N}$,
$P(\cdot)$ is a functional, and $F_{i}^{(h)}$ is the joint distribution function of 
$Y_i := (X_{i},\ldots,X_{i+h-1})^{\intercal}$ for $i=1, \ldots, n-h+1$. 
For example, for $h = 1$, 
$\mu_i = \int_{-\infty}^{\infty} x \dd F_{i}^{(1)}(x)$ and 
$\sigma^2_i = \int_{\mathbb{R}} x^2 \dd F_{i}^{(1)}(x) - \{\int_{\mathbb{R}} x \dd F_{i}^{(1)}(x) \}^2$
are the marginal mean and variance at time $i$ respectively.
For $h = 2$, the lag-1 autocovariance at time $i$ is $\gamma_i(1) = \int_{\mathbb{R}^2} (x_i - \mu_i)(x_{i-1} - \mu_{i-1}) \dd F_{i}^{(2)}(x_{i-1},x_{i})$.
The hypotheses in 
(\ref{eqt:H0}) and (\ref{EQ: multiple CP hypothesis test}) are redefined by replacing $\mu_i$'s by $\theta_i$'s. A possible global change detecting process is
\begin{equation}
\label{EQ: global parameter process}
	\G_n(\lfloor nt \rfloor) = \frac{\lfloor nt \rfloor(n-\lfloor nt \rfloor)}{n^{3/2}}\left( \widehat{\theta}_{1,\lfloor nt \rfloor} - \widehat{\theta}_{\lfloor nt \rfloor+1,n}\right),
\end{equation}
for $t\in[0,1]$; see, e.g., \cite{shao2010SNCI}, \cite{shao2010testing} and \cite{shao2015self}. We remark that $\widehat{\theta}$ is allowed to be of dimension $q > 1$; see Example \ref{exp: multivariate time series} for the construction of the corresponding global change detecting process.
Clearly, if $ \widehat{\theta}_{i,j} = (j-i+1)^{-1}\sum_{t = i}^j X_t$, 
then $\G_n(\cdot) = \C_n(\cdot)$.
So, (\ref{EQ: global parameter process}) generalizes \eqref{EQ Global CUSUM} 
from testing changes in $\mu_i$'s to testing changes in $\theta_i$'s. 
The final test statistic is $\T^{(\G)}_n := \T^{(\G)}_n[\max, \mean]$.

The limiting distribution of (\ref{EQ: global parameter process}) requires 
some standard regularity conditions in handling statistical functionals.
Define the empirical distribution of $Y_i$'s based on the sample $Y_s,\ldots, Y_e$
by $\widehat{F}_{s,e}^{(h)} = (e-s+1)^{-1} \sum_{i = s}^e \delta_{Y_i}$, 
where $\delta_y$ is a point mass at $y \in \mathbb{R}^h$.  
Assume that $\hat{\theta}_{s,e} := P(\widehat{F}_{s,e}^{(h)})$ is asymptotically linear in the following sense, as $e-s \rightarrow \infty$:
\begin{equation}
\label{EQ: parameter decomposition}
		P(\widehat{F}_{s,e}^{(h)}) 
		= P(F^{(h)}) + (e-s+1)^{-1}\sum_{i = s}^e \IF(Y_i \mid P, F^{(h)}) + R_{s,e},
\end{equation}
where $R_{s,e}$ is a remainder term,
$F^{(h)} = F^{(h)}_s = \cdots = F^{(h)}_e$, and 
$\IF(y \mid P, F^h) = \lim_{\epsilon \rightarrow 0} \{ P((1-\epsilon)F^h + \epsilon\,\delta_y) - P(F^h)\}/\epsilon$
is the influence function;
see Section 2.3 of \citet{Wasserman2006}.
The asymptotic linear representation (\ref{EQ: parameter decomposition}) is known as 
a nonparametric delta method.

\begin{corollary}[Limiting distribution of $\T^{(\G)}_n$]\label{thm: general parameter test}
Under $H_0$, 
$I_n(\lfloor nt \rfloor) := n^{-1/2}\sum_{i = 1}^{\lfloor nt \rfloor}\IF(Y_i \mid P, F^{(h)})$
for $t\in[0,1]$ and $h\in\mathbb{N}$, 
where $F^{(h)} = F^{(h)}_1 = \cdots = F^{(h)}_{n-h+1}$.
Suppose that (i) $\E\{\IF(Y_i \mid P, F^h)\} =0$ for all $i$, 
(ii) $\{I_n(t) : t\in[0,1]\} \Rightarrow \{ \sigma_{\G} \BM(t) : t\in[0,1]\}$ for some $0< \sigma_G <\infty$, and 
(iii) $\sup_{k = 1, \ldots, n-1} | R_{1:k}| + | R_{(k+1):n}| = o_p(n^{-1/2})$.
Then  $\T^{(\G)}_n \inD \mathbb{T}$.
\end{corollary}
\end{example}

\begin{example}[Long-range dependent time series]\label{exp: LRD}
Long-range dependence (LRD) 
is widely used for modeling time series data in, e.g., 
earth sciences, econometrics and traffic systems; see \cite{palma2007} for a review. 
The time series is of LRD 
if the autocovariances satisfy
\begin{equation}
\label{EQ: LRD autocaovariances}
\gamma(k) \sim |k|^{-D}\phi(k),
\end{equation}
as $k\rightarrow\infty$, 
where $D\in(0,1)$ is the LRD parameter, 
$\phi(\cdot)$ is a slowly varying function at infinity;
see, e.g., \cite{taqqu2017LRD}. In this section, we discuss two possible change detecting processes for handling LRD.

First, the global CUSUM process \eqref{EQ Global CUSUM} with suitable normalization converges under the LRD setting. 
Specifically, by Theorem 5.1 in \cite{taqqu1975weak}, the scaled CUSUM process satisfies 
\begin{equation*}
\C_n^*(\lfloor nt \rfloor) := (\sqrt{n}/d_n)\C_n(\lfloor nt \rfloor) \inD \sigma_{\C^*} \{\BM_H(t)- t\BM_H(1)\},
\end{equation*}
where $d_n^2 \sim 2n^{2-D}\phi(n)/\{(1-D)(2-D) \}$, $\sigma_{C^*}$ is a positive constant and $\BM_H(\cdot)$ is a standard fractional Brownian motion with a Hurst parameter $H = 1 - D/2$. 
Therefore, $\C_n^*(\cdot)$ can be used in our framework and the resulting test statistic is $\T^{(\C^*)}_n = \T^{(\C^*)}_n[\max , \mean]$. Second, since the Wilcoxon statistic is a partial sum of ranks, the invariance principle still holds under suitable normalization and some regularity conditions. 
By Theorem 2 in \cite{dehling2013LRD}, 
\begin{equation*}
\W_n^*(\lfloor nt \rfloor) := (\sqrt{n}/d_n)\W_n(\lfloor nt \rfloor) \inD (2\sqrt{\pi})^{-1}\left\{\BM_H(t) - t\BM_H(1)\right\}.
\end{equation*}
The resulting test statistic is $\T^{(\W^*)}_n := \T^{(\W^*)}_n[\max, \mean]$.
Note that users do not need to specify the unknown sequence $d_n$ 
since it is cancelled out by our SN. Thus, $\T^{(\C^*)}_n=\T^{(\C)}_n$ and $\T^{(\W^*)}_n=\T^{(\W)}_n$. 
However, by Corollary \ref{thm: LRD test} below, the limiting distribution 
is a function of a fractional Brownian motion instead of a Brownian motion. Note that \cite{shaoSN2011LRD} and \cite{betken2016testing} proposed global self-normalized test for the AMOC problem under long-range dependent time series. Our approach contributes by further extending it to the multiple-CP testing.

\begin{corollary}[Limiting distribution of $\T^{(\C^*)}_n$ and $\T^{(\W^*)}_n$]\label{thm: LRD test}
Let $X_i = \mu_i + h(U_i)$, where $\{U_i\}$ is a stationary zero-mean standard Gaussian process and its autocaovariances satisfy \eqref{EQ: LRD autocaovariances}, and $h$ is a measurable strictly monotone function such that $\E\{h(U_i)\} = 0$. Under $H_0$, we have
$\T^{(\C^*)}_n \inD \mathbb{T}_H$ and $\T^{(\W^*)}_n \inD \mathbb{T}_H$,
where
\begin{equation*}
	\mathbb{T}_H:= (1-2\epsilon)^{-1} \int_{\epsilon}^{1-\epsilon} 
			\sup_{\delta > \epsilon} \frac{\mathbb{L}_H(t \mid t-\delta, t+\delta)^2}{\mathbb{V}_H(t\mid t-\delta,t+\delta)} \, \dd t,
\end{equation*}
for any $0<\epsilon<1/2$,
where $\mathbb{L}_H$ is defined as $\mathbb{L}$ in (\ref{eqt:defLlimit}) 
but with $\mathbb{B}(\cdot)$ being replaced by $\mathbb{B}_H(\cdot)$; 
and 
$\mathbb{V}_H$ is defined as $\mathbb{V}$ in (\ref{eqt:defLlimit}) 
but with $\mathbb{L}(\cdot\mid \cdot, \cdot)$ being replaced by 
$\mathbb{L}_H(\cdot\mid \cdot, \cdot)$.
\end{corollary}
In practice, the unknown parameter $H$ can be consistently estimated through regression in the spectral domain \citep{palma2007} by, e.g., the local Whittle method \citep{KunschLocalWhittle1987}. \cite{shao2007localwhittle} showed the consistency of the local Whittle estimator.
\end{example}

\begin{example}[Multi-dimensional parameter]\label{exp: multivariate time series}
Our approach extends to the multivariate case naturally.
Let $\{\Q_n(\lfloor nt \rfloor)\}_{0 \leq t \leq 1}$ be a $q$-dimensional CP detecting process 
under the AMOC problem, where $q \in \mathbb{N}$.  
For example for testing changes in $q$-dimensional parameter $\Theta$, we can use a multivariate version of \eqref{EQ: global parameter process}:
\begin{align}
\label{EQ: q-dim global parameter process}
    \Q_n(\lfloor nt\rfloor) = \frac{\lfloor nt \rfloor(n-\lfloor nt \rfloor)}{n^{3/2}}\left( \widehat{\Theta}_{1,\lfloor nt \rfloor} - \widehat{\Theta}_{\lfloor nt \rfloor+1,n}\right),
\end{align}
where $\widehat{\Theta}_{s,e}$ is the estimator for $\Theta$ using the subsample $\{X_s,\ldots,X_e \}$.
We define the corresponding LSN statistic as $\T_n^{(\Q)}(k\mid s, e) = \L_n^{(\Q)}(k\mid s, e)^{\T} \V_n^{(\Q)}(k\mid s, e)^{-1} \L_n^{(\Q)}(k\mid s, e)$, where
\begin{align*}
\L_n^{(\Q)}(k\mid s, e) &= \left(\frac{n}{e-s+1}\right)^{1/2}\left[ \Q_n(k) - \Q_n(s-1) - \frac{k-s+1}{e-s+1}\left\{ \Q_n(e) - \Q_n(s-1)\right\}\right],\\
\V_n^{(\Q)}(k\mid s, e) &= \frac{k-s+1}{(e-s+1)^2}\sum_{j = s}^k \L_n^{(\Q)}(j\mid s, k)^{\otimes 2} +  \frac{e-k}{(e-s+1)^2}\sum_{j = k+1}^e \L_n^{(\Q)}(j\mid k+1, e)^{\otimes 2},
\end{align*}
and $A^{\otimes 2} = AA^{\T}$ for any matrix $A$. 
Aggregating $\T_n^{(\Q)}(k\mid s, e)$
as in (\ref{EQ: CUSUM test score function}) and (\ref{EQ: CUSUM ts}), 
we obtain the test statistic $\T_n^{(\Q)} :=  \T^{(\Q)}_n[\max, \mean]$. 
Its limiting distribution is stated as follows.

\begin{corollary}[Limiting distribution of $\T^{(\Q)}_n$]\label{thm: multivartiate test}
Assume
$\{\Q_n(\lfloor(nt\rfloor):t\in[0,1]\} 
	\Rightarrow \{\Sigma_{\Q}^{1/2}\{\BM^{(q)}(t) - t\BM^{(q)}(1) \} :t\in[0,1]\}$ 
for some unknown positive definite matrix $\Sigma_{\Q}^{1/2} \in \mathbb{R}^{q \times q}$,
and $\BM^{(q)}(\cdot)$ is $q$-dimensional standard Brownian motion. 
Under $H_0$, we have $\T^{(\Q)}_n \inD \mathbb{T}^{(q)}$ for any $0<\epsilon<1/2$, where 
\begin{equation*}
	\mathbb{T}^{(q)} 
		:= (1-2\epsilon)^{-1} \int_{\epsilon}^{1-\epsilon} 
			\sup_{\delta > \epsilon} \mathbb{L}^{(q)}(t \mid t-\delta, t+\delta)^{\T} \mathbb{V}^{(q)}(t\mid t-\delta,t+\delta)^{-1}\mathbb{L}^{(q)}(t \mid t-\delta, t+\delta) \, \dd t,
\end{equation*}
where $\mathbb{L}^{(q)}$ is defined as $\mathbb{L}$ in (\ref{eqt:defLlimit}) 
but with $\mathbb{B}(\cdot)$ being replaced by $\mathbb{B}^{(q)}(\cdot)$; 
and 
$\mathbb{V}^{(q)}$ is defined as $\mathbb{V}$ in (\ref{eqt:defLlimit}) 
but with $\mathbb{L}(\cdot\mid \cdot, \cdot)^2$ being replaced by 
$\mathbb{L}^{(q)}(\cdot\mid \cdot, \cdot)^{\otimes 2}$.
\end{corollary}
\cite{shao2010testing} proposed a general AMOC self-normalized change point test for the mean of multivariate time series. Our approach localizes their proposed CP detecting process and self normalizer for multiple-CP detection. Note that their proposed SN requires repeated calculation of the quantities of interest for each subsample while our SN reuses the global change detecting process which allows faster computation.

\begin{remark}
If the dimension $q$ increases with respect to the sample size $n$, then the proposed self normalizer $\V_n^{(\Q)}(\cdot)$ may not be invertible. To overcome the issue, \cite{shaoadaptiveHDsnmct2021} and \cite{wang2019HDmctsn} proposed univariate change detecting processes based on U-statistic for testing existence of CPs in the mean of high dimensional independent data. \cite{wang2019HDmctsn} further integrated their proposal into \cite{zhang2018unsupervised}'s framework and the test was shown to preserve power in the multiple-CP setting. By applying their proposal into our framework, the non-invertible issue can be avoided.
Note that \cite{wang2019HDmctsn} proposed a change detecting process that converges to a function of a centered Gaussian process. Thus, their test statistic is asymptotically pivotal. The critical values can be found through Monte Carlo simulation. The performance of the test is left for further investigation.
\end{remark}
\end{example}

\section{Discussion and Implementation}
\label{sec: Implementation}
\subsection{Comparison with existing methods}
\label{sub-sec: test comparison discussion}

\cite{shao2010testing} extended the self-normalized one-CP test \eqref{EQ: shao 1 ts} to 
a ``supervised'' multiple-CP test tailored for testing $m$ CPs, 
where $m$ is pre-specified.
The test statistic is  
\begin{eqnarray}\label{eqt:ShaoGeneral}
	\shao_n^{(m)} 
		= \sup_{(k_1, \ldots, k_m) \in \Omega_n^{(m)}(\epsilon)} 
			\sum_{j=1}^m \T_n^{(\C)}(k_1 \mid k_{j-1}, k_{j+1}), 
\end{eqnarray}
where $1 = k_0 < k_1 < \cdots < k_m < k_{m+1} = n$, and 
\begin{equation*}
\Omega_n^{(m)}(\epsilon) 
	= \left\{
		k_1, \ldots, k_m : 
		\forall\, j\in\{1,\ldots,m\},
		\lfloor \epsilon n \rfloor \leq \min(k_{j}-k_{j-1}, k_{j+1}-k_{j}) \right\}.
\end{equation*}
The trimmed region $\Omega_n^{(m)}(\epsilon)$ 
prevents computing estimates by too few observations.
\cite{zhang2018unsupervised} later proposed a ``unsupervised'' self-normalized multiple-CP test which does not require specifying $m$ and defined as,
\begin{eqnarray}\label{eqt:zhang}
	\zhang_n = \sup_{(k_1,k_2) \in \Omega_n^{(2)}(\epsilon)} \T_n^{(\C)}(k_1 \mid 1, k_2) 
				+ \sup_{(k_1,k_2) \in \Omega_n^{(2)}(\epsilon)} \T_n^{(\C)}(k_2 \mid k_1, n).
\end{eqnarray}	
Although both $\shao_n^{(m)}$, $\zhang_n$ and $\T_n^{(C)}$ use the LSN CUSUM statistics, 
i.e., $\T_n^{(\C)}(k \mid s, e)$ defined in \eqref{EQ: LSN stat CUSUM}, 
as building blocks, they have different restrictions on the local windows and aggregate the LSN statistics in different ways.

\cite{shao2010testing}'s test statistic is the maximum of sum of LSN statistics 
calculated on $m$ local windows with target time points $k_1, \ldots, k_m$, 
i.e., $[1, k_2], \ldots, [k_{m-1},n]$. 
This approach is also used by \cite{antoch2013testing} for constructing 
a non-self-normalized multiple-CP test. 
It has a strict control on the $m$ local windows 
since the boundaries of a window relate to previous and next windows. If the number of CPs $M$ is misspecified, some of the windows may not contain only one CP. 
If $m<M$, the self normalizers are not robust to the changes. 
If $m>M$, some degree of freedom is wasted to detect CPs that do not exist. 
Both two cases may lead to a huge loss in power. 

\cite{zhang2018unsupervised}'s approach sets the left end of the window to be 1 in the forward scan, and 
the right end to be $n$ in the backward scan. 
So, their approach tends to scan for the first CP $k_1\in[1, e]$, and the last CP 
$k_M\in[s, n]$ for some $e,s$; see Section \ref{sec: justify symm windows}. 
In contrast, our approach scans for all possible CPs.
We also allow the windows to start and end at any times
instead of restricting them to start at time 1 or end at time $n$. 

Besides, \cite{jiang2021CP} recently considered a self-normalized 
CP estimation problem for a quantile regression model, 
while we consider a general CP testing problem. 
Since we solve different inference problems,
the results are not directly comparable. 
So,
we only compare the self normalizers in these two approaches. 
First, their self normalizer is specifically designed for 
quantile regression model, but 
ours is built on a general framework, which supports 
any global change detecting process; 
see Section \ref{sec: general framework}.
One of our major contributions is to extend 
any one-CP test to multiple-CP self normalized test. 
Second, computing their self normalizer requires fitting the 
quantile regression model $O(n^2)$ times
on all possible local subsamples. 
In contrast, our self normalizer is a function of a 
global change detecting process, which can be computed 
by fitting the models on $O(n)$ local subsamples.
It significantly reduces computational burden. 
Although we focus on different problems, 
our proposed test statistic can be integrated
into the first step of their algorithm. 
The fused algorithm has a potential of further enhancing 
generality and computational flexibility. 
We leave it for further study.

Computational efficiency is usually a major concern in constructing multiple-CP tests.
Table \ref{tab: tests comparison} and Figure \ref{fig: time complexity}(a) present the average computational times for the test statistics 
$\KS_n$, $\shao_n^{(1)}$, $\shao_n^{(m)}$, $\zhang_n$, $\T_n^{(C)}$, $\T_n^{(W)}$ and $\T_n^{(H)}$. 
The computational cost of $\KS_n$ and $\shao_n^{(1)}$ is the least 
since they are tailor-made for the AMOC problem and have the simplest forms. Moreover, the computational cost of $\shao_n^{(m)}$
is scaled exponentially with $m$. 
Among the self-normalized multiple-CP tests, 
$\shao_n^{(2)}$, $\zhang_n$, $\T_n^{(\C)}$ and $\T_n^{(\W)}$ have the same computational complexity. 
Moreover, our framework allows fast recursive computation. 
It further reduces the computational time of $\T_n^{(\C)}$ and $\T_n^{(\W)}$; 
see Section \ref{sec: recursive computational} in the appendix. 
The higher time cost required by $\T_n^{(\H)}$ is solely due to higher time cost of computing the global Hodges--Lehmann change detecting statistic \eqref{EQ: HL global}.
Figure \ref{fig: time complexity}(b) shows that the computational cost of using nonsymmetric window in $\tilde{\T}_n^{(\C)}$ is significantly larger than that of $\T_n^{(\C)}$
because of the larger feasible set in computing the score function \eqref{EQ: LSN stat nonsymm window}. However, the gain in power is insignificant; see Section \ref{sec: justify symm windows}. The incremental power improvement does not justify the huge additional time cost. Therefore, using symmetric windows is recommended.

\begin{table}[t]
\setlength{\tabcolsep}{3pt}
\centering
\footnotesize
\caption{\footnotesize Comparisons among different CP tests in terms of 
(a) finite-sample size accuracy with respect to the nominal type-I error rate; 
(b) powerfulness under different actual numbers of CP $M$; 
(c) time complexity of computing the test statistics based on a sample of size $n$; 
(d) robustness against outliers; and 
(e) requirements of computing a LRV estimate $\widehat{\sigma}^2_n$ and 
specifying a target number of CPs $m$.
The comparisons in (a) and (b) are based on our simulation results 
in Section \ref{sec: simulation} and the appendix.} 
\label{tab: tests comparison}
\begin{tabular}{c c cc c c cc}
\toprule
 & (a) size accuracy & \multicolumn{2}{c}{(b) powerfulness} & (c) time complexity  & (d) robustness& \multicolumn{2}{c}{(e) requirement} \\
\cmidrule(r){3-4}  \cmidrule(r){7-8}
&& $M=1$ & $M>1$ &&  & $~~~\widehat{\sigma}^2_n~~~$ & $m$  \\
\cmidrule(r){2-8}  
$\KS_n$ & low & high & low & $O(n)$ & low & yes & no \\
$\shao_n^{(1)}$ & high & high & low & $O(n)$ & low& no & yes  \\
$\shao_n^{(m)}$ & medium & medium & depends on $m$ & $O(n^m)$  & low& no & yes \\
$\zhang_n$ & medium & medium & medium & $O(n^2)$  & low& no & no  \\
$\T_n^{(\C)}$ & high & high & high & $O(n^2)$ & low& no & no \\
$\T_n^{(\W)}$ & high & high & high & $O(n^2)$   & high & no & no \\
$\T_n^{(\H)}$ & high & high & high & $O(n^3\log n)$  & high & no & no \\
\bottomrule
\end{tabular} 
\end{table}

\begin{figure}[t]
\centering
\includegraphics[width=0.7\linewidth]{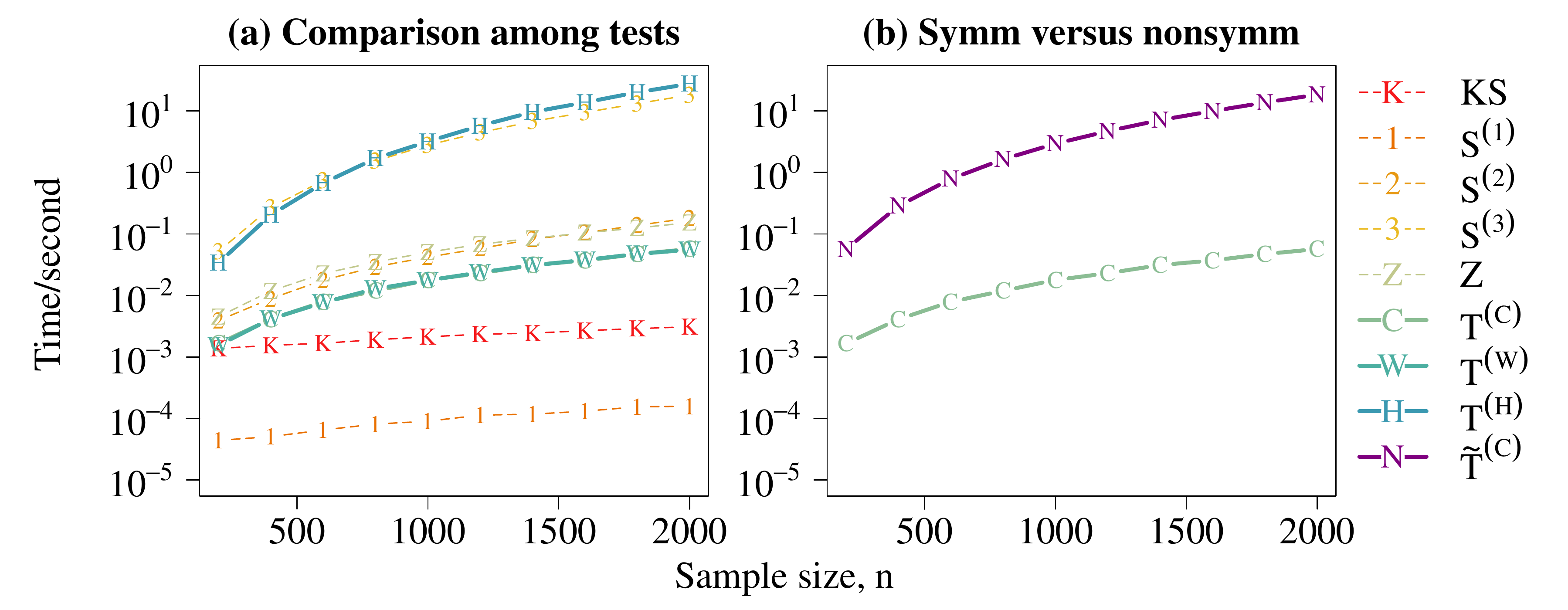}
\vspace{-0.4cm}
\caption{\footnotesize \label{fig: time complexity} 
Average computation times of (a) test statistics considered (including $\KS_n$, $\shao_n^{(1)}$, $\shao_n^{(2)}$, $\shao_n^{(3)}$, $\zhang$, $\T_n^{(\C)}$, $\T_n^{(\W)}$, $\T_n^{(\H)}$) and (b) proposed CUSUM test using symmetric windows (symm) $\T_n^{(\C)}$ versus nonsymmetric windows (nonsymm) $\tilde{\T}_n^{(\C)}$. 
All experiments are performed in R 4.1.0 on Intel Core i9-10900 CPU @ 2.8 GHz, and 64-bit operating system.
}
\end{figure}

\subsection{Finite-$n$ adjusted critical values}\label{sec: adjust critical values}
In the time series context, 
the accuracy of the invariance principle (i.e., Assumption \ref{as:invariance principle}) 
may deteriorate when the sample size is small \citep{kiefer2005new}.  
Therefore, the asymptotic theory in, e.g., Theorem \ref{thm: CUSUM test}, 
may not kick in. It may lead to severe size distortion. To mitigate this problem, finite-sample adjusted critical values can be used.
We propose 
to compute a critical value $c_{\alpha}(n,\rho)$ by matching the autocorrelation function (ACF) at lag one $\rho$ and $n$
for different specified level of significance $\alpha\in(0,1)$.  
The values of $c_{\alpha}(n,\rho)$ are tabulated in Section \ref{sec:finite_n_cVal} of the appendix
for $\alpha\in\{0.01, 0.05, 0.1\}$, $n\in\{100,200, \ldots, 1000, 2000, 3000, \ldots, 10000\}$, and $\rho\in\{0,\pm 0.1, \ldots, \pm, 0.9\}$.
The testing procedure is outlined as follows.

\begin{enumerate}[noitemsep]
	\item  Compute the sample lag-$1$ ACF $\widehat{\rho}$ of 
			$\{ X_{i + \lfloor n^{1/3}\rfloor} - X_i\}_{i =  1}^{n-\lfloor n^{1/3}\rfloor}$.
	\item Obtain the critical value $c_{\alpha}(n, \widehat{\rho})$.			
			Use interpolation if necessary.
	\item Reject the null if $\T_n^{(\aleph)} > c_{\alpha}(n, \widehat{\rho})$, where $\aleph \in \{\C, \W,\H, \G \}$. 
\end{enumerate}
The limiting critical value is approximated by using 
finite-$n$ critical value with $n=10000$. 
Since $\widehat{\rho}$ relies on 
the differenced data with a suitable lag, 
the change points (if any) have negligible effect on $\widehat{\rho}$.
The consistency of $\widehat{\rho}$ is developed on the following framework. 
Let $X_i = \mu_i + Z_i$, where $\mu_i$'s are deterministic and $Z_i$'s are zero-mean stationary noises. 
Define $Z_i = g(\ldots, \varepsilon_{i-1},\varepsilon_i)$, where $\varepsilon_i$'s are independent and identical ($\iid$) random variables and $g$ is some measurable function. Let $\varepsilon_i'$ be $\iid$ copy of $\varepsilon_i$ and $Z_i' = g(\ldots,\varepsilon_{-1},\varepsilon_0',\varepsilon_1,\ldots,\varepsilon_i)$. For $p>1$ and $\norm{Z_i}_p < \infty$, the physical dependence measure \citep{wu2011asymptotic} is defined as $\lambda_p(i) = \norm{Z_i - Z_i'}_p$, where $\norm{\cdot}_p = (\E|\cdot|^p)^{1/p}$. 
Theorem \ref{thm: rho consistency} below states that $\widehat{\rho}$
is a consistent estimator of $\rho := \gamma(1)/\gamma(0)$
even there exist a large number ($M$) of CPs,
where $\gamma(k) = \E(Z_1 Z_{1+k})$.

\begin{theorem}
\label{thm: rho consistency}
Assume $\E(Z_1^4) < \infty$ and $\sum_{i = 1}^{\infty} \lambda_4(i) < \infty$.
Define $|\Delta_{*}| = \sup_{1\leq j\leq M} |\Delta_j|$. 
Let $b = b_n$ be an $\mathbb{N}$-valued sequence. 
If $b/n + 1/b + (bM/n)\Delta_{*}^2 \rightarrow 0$ as 
$n \rightarrow \infty$, then  
$\widehat{\rho} \inP \rho$.
\end{theorem}

In particular, if $b=\lfloor n^{1/3} \rfloor$ and $|\Delta_{*}| <\infty$, then 
Theorem \ref{thm: rho consistency} guarantees that $\widehat{\rho}$ is consistent for $\rho$
provided that the actual number of CPs satisfies $M = o(n^{2/3})$. 
This testing procedure does not assume any parametric structure as we only intend to match the sample ACF at lag one. Although it is possible to further match ACF at a higher lag, the proposal is to balance the accuracy and the computational cost. Simulation result shows that the tests have accurate size under the above critical value adjustment procedure. We emphasize that this adjustment only affects the finite-sample performance because the critical value $\mathit{c}_{\alpha}(n, \rho)$ is a constant as a function of $\rho$ when $n\rightarrow\infty$.

\section{Simulation experiments}
\label{sec: simulation}
\subsection{Setting and overview}
Throughout Section \ref{sec: simulation}, the experiments are designed as follows.
The time series is generated from a signal-plus-noise model:
$X_i = \mu_i + Z_i$. 
The values of $\mu_i$'s will be specified in
Sections \ref{sec:sizePower} and \ref{sec:sim_allCPsEffect}. 
The zero-mean noises $Z_i$'s are simulated from a stationary
bilinear autoregressive (BAR) model: 
$Z_i =  (\varpi + \vartheta \varepsilon_i)Z_{i-1} + \varepsilon_i$ for $i=1, \ldots,n$, 
where $\varepsilon_i$'s are independent standard normal random variables, and
$\varpi, \vartheta \in\mathbb{R}$ such that $\varpi^2 + \vartheta^2 <1$.
Similar conclusions are obtained under other noise models, including 
autoregressive-moving-average model, 
threshold AR model, and 
absolute value nonlinear AR model.
Due to space constraint, their results are deferred to the appendix. 
The critical value is chosen according to the adjustment procedure described in Section \ref{sec: adjust critical values}.

\subsection{Size and power}\label{sec:sizePower}
In this subsection, we examine the size and power
of different CP tests when there exist various numbers of CPs. 
Following the suggestion of \cite{huang2015self} and \cite{zhang2018unsupervised}, 
we choose $\epsilon = 0.1$ in $\shao_n^{(2)}$, $\shao_n^{(3)}$, $\zhang_n$ and our proposals 
in order to have a fair comparison. 
Suppose that the CPs are evenly spaced and 
the mean change directions are alternating (increasing or decreasing):
\begin{equation}
\label{EQ: CP models, mean level}
	\mu_i = \Delta \sum_{j = 1}^M (-1)^{j+1} \mathbb{1}_{\{ i/n > j/(M+1) \}},
\end{equation} 
where $M \in \{1,\ldots, 6\}$ denotes the number of CPs,
and $\Delta\in\mathbb{R}$ controls the magnitude of the mean changes. 
Clearly, $\mu_i \equiv 0$ under $H_0$.

All tests are computed at nominal size $\alpha=5\%$.
The null rejection rates $\widehat{\alpha}$, 
for sample sizes $n \in\{ 200, 400\}$, are presented in Table \ref{tab: BAR11 size}. To summarize the result, we further report the sample root mean squared error (RMSE) of $\widehat{\alpha}$ 
over all cases of $\vartheta$ and $\varpi$ for each test and each $n$. 
The self-normalized tests generally have more accurate size 
than the non-self-normalized test $\KS_n$. The finding agrees with \cite{shao2015self}. 
The self-normalized approach is a special case of the fixed bandwidth approach 
with the largest possible bandwidth and thus achieves the smallest size distortion according to \cite{kiefer2002heteroskedasticity}. 
In comparison among the self-normalized tests, $\T_n^{(\W)}$ and $\T_n^{(\H)}$ control the size most accurately.
In particular, our proposed tests have the least severe under-size problem, 
which is observed in the existing self-normalized tests (i.e., $\shao_n^{(1)}$, $\shao_n^{(2)}$, $\shao_n^{(3)}$,
and $\zhang_n$) when $\varpi<0$. 
Moreover, \cite{shao2010testing}'s proposal $\shao_n^{(m)}$ suffers from 
an increasing size distortion as $m$ increases. 
This finding is in line with \cite{zhang2018unsupervised}.
The outperformance of $\T_n^{(\W)}$ and $\T_n^{(\H)}$ over $\T_n^{(\C)}$ 
can be attributed to the outlier robustness achieved by rank and order statistics. This 
advantage is particularly obvious in bilinear time series as this model 
is well known to produces sudden high amplitude oscillations to mimic structures in, e.g., 
explosion and earthquake data in seismology; 
see Section 5.2 of \cite{rao2012introduction}.
However, in the more standard ARMA models, 
$\T_n^{(\C)}$ performs as well as $\T_n^{(\W)}$ and $\T_n^{(\H)}$; see the appendix for the detailed 
results.

\begin{table}[t]
\setlength{\tabcolsep}{4pt}
\centering
\footnotesize
\caption{\footnotesize \label{tab: BAR11 size} Null rejection rates $\widehat{\alpha}$ at nominal size $\alpha=5\%$  under BAR model and mean function (\ref{EQ: CP models, mean level}). The sample RMSE of $\widehat{\alpha}$ over all cases of $\vartheta$ and $\varpi$ is computed for each test and each $n$. 
The values of $\widehat{\alpha}$ and RMSE are expressed in percentage. 
The smaller the RMSE, the better is the test in controlling the type-I error.
}
\begin{tabular}{ rr | r rrr r rrr | r rrr r rrr }
\toprule
\multicolumn{1}{c}{ } & \multicolumn{1}{c}{ } &  \multicolumn{8}{c}{$n=200$} & \multicolumn{8}{c}{$n=400$}\\ 
\cmidrule(r){3-10}\cmidrule(r){11-18}
$\vartheta$ &$\varpi$& $\KS_n$ & $\shao_n^{(1)}$ & $\shao_n^{(2)}$ & $\shao_n^{(3)}$ & $\zhang_n$ & $\T_n^{(\C)}$ & $\T_n^{(\W)}$ & $\T_n^{(\H)}$ & $\KS_n$ &  $\shao_n^{(1)}$ & $\shao_n^{(2)}$ & $\shao_n^{(3)}$ & $\zhang_n$ & $\T_n^{(\C)}$ & $\T_n^{(\W)}$ & $\T_n^{(\H)}$  \\ 
\cmidrule(r){1-18}
0.8 & 0.5 & 14.4 & 7.5 & 9.6 & 5.3 & 7.6 & 1.1 & 4.0 & 3.4 & 11.3 & 5.3 & 5.3 & 3.1 & 5.4 & 0.8 & 3.8 & 3.3\\
 & 0.3 & 13.5 & 5.6 & 5.8 & 3.1 & 5.1 & 1.2 & 4.5 & 3.4 & 10.7 & 4.9 & 4.5 & 2.7 & 4.5 & 1.3 & 4.0 & 3.4\\
 & 0 & 15.4 & 4.7 & 2.8 & 1.3 & 2.4 & 2.3 & 5.6 & 4.6 & 10.6 & 4.8 & 3.7 & 1.9 & 3.0 & 2.1 & 4.1 & 3.7\\
 & $-0.3$ & 23.9 & 4.6 & 1.7 & 0.4 & 1.5 & 2.1 & 7.8 & 6.3 & 15.6 & 3.9 & 2.1 & 0.8 & 1.9 & 2.1 & 5.6 & 5.0\\
 & $-0.5$ & 31.2 & 3.9 & 1.6 & 0.4 & 1.2 & 2.0 & 7.7 & 7.1 & 25.6 & 4.5 & 1.4 & 0.3 & 1.4 & 1.4 & 6.4 & 6.3\\
 \hline
0.5 & 0.8 & 31.8 & 11.9 & 26.1 & 27.2 & 27.1 & 3.6 & 11.0 & 9.0 & 31.8 & 8.0 & 14.9 & 12.8 & 14.3 & 2.2 & 5.7 & 3.6\\
 & 0.5 & 13.7 & 6.5 & 9.5 & 8.8 & 8.1 & 2.1 & 4.5 & 3.9 & 10.0 & 6.0 & 6.6 & 6.1 & 6.4 & 2.1 & 3.9 & 3.5\\
 & 0.3 & 11.0 & 5.1 & 5.5 & 4.9 & 5.6 & 3.0 & 5.0 & 4.6 & 7.6 & 4.8 & 5.7 & 4.1 & 5.2 & 2.7 & 4.5 & 4.1\\
 & 0 & 9.4 & 3.9 & 3.8 & 2.0 & 2.7 & 3.8 & 6.5 & 6.1 & 6.6 & 4.8 & 4.1 & 3.1 & 3.8 & 3.2 & 4.3 & 4.4\\
 & $-0.3$ & 17.9 & 3.5 & 1.7 & 0.4 & 1.5 & 4.3 & 6.3 & 5.9 & 10.0 & 4.2 & 3.2 & 1.8 & 2.5 & 3.1 & 5.3 & 4.8\\
 & $-0.5$ & 29.8 & 3.1 & 1.0 & 0.0 & 0.8 & 4.1 & 7.2 & 6.0 & 19.8 & 3.9 & 2.7 & 1.5 & 2.0 & 3.3 & 6.0 & 6.0\\
 & $-0.8$ & 44.1 & 2.6 & 0.2 & 0.0 & 0.1 & 4.3 & 8.7 & 7.3 & 37.3 & 2.9 & 0.9 & 0.1 & 0.4 & 3.2 & 6.9 & 6.5\\
\hline
$-0.5$ & 0.8 & 29.8 & 10.1 & 24.0 & 27.7 & 24.8 & 2.8 & 11.4 & 8.5 & 29.3 & 8.2 & 11.9 & 9.6 & 12.4 & 1.6 & 4.9 & 3.6\\
 & 0.5 & 12.2 & 5.6 & 8.6 & 7.1 & 9.0 & 1.7 & 3.4 & 2.6 & 10.0 & 5.1 & 5.4 & 4.4 & 5.9 & 2.3 & 4.0 & 3.4\\
 & 0.3 & 9.3 & 4.9 & 6.3 & 3.6 & 5.8 & 2.4 & 3.7 & 3.4 & 7.4 & 4.3 & 5.3 & 3.7 & 4.6 & 3.2 & 4.1 & 3.6\\
 & 0 & 9.1 & 5.1 & 3.4 & 2.1 & 3.0 & 3.6 & 4.7 & 4.3 & 6.4 & 4.8 & 4.4 & 3.2 & 4.0 & 3.3 & 5.8 & 5.3\\
 & $-0.3$ & 15.8 & 3.9 & 1.6 & 1.0 & 1.8 & 3.1 & 6.1 & 5.2 & 10.1 & 4.6 & 3.7 & 2.0 & 3.0 & 3.8 & 5.3 & 5.0\\
 & $-0.5$ & 27.9 & 3.6 & 1.0 & 0.6 & 1.1 & 4.3 & 6.2 & 6.0 & 19.5 & 4.4 & 3.1 & 1.6 & 3.0 & 4.3 & 5.4 & 5.1\\
 & $-0.8$ & 44.0 & 1.8 & 0.7 & 0.0 & 0.3 & 4.7 & 6.3 & 6.3 & 38.7 & 3.5 & 1.0 & 0.2 & 0.9 & 3.1 & 6.1 & 5.6\\
 \hline
$-0.8$ & 0.5 & 13.6 & 7.4 & 7.4 & 4.5 & 6.3 & 0.8 & 3.1 & 2.5 & 11.0 & 6.2 & 4.5 & 2.8 & 3.9 & 1.1 & 3.2 & 2.5\\
 & 0.3 & 13.2 & 5.9 & 4.6 & 2.3 & 3.6 & 1.0 & 3.2 & 2.7 & 9.1 & 5.7 & 3.3 & 2.1 & 2.7 & 1.9 & 4.2 & 3.4\\
 & 0 & 14.4 & 5.1 & 2.6 & 1.1 & 2.4 & 1.3 & 4.1 & 3.2 & 11.2 & 4.5 & 3.6 & 1.6 & 3.1 & 2.4 & 5.8 & 5.1\\
 & $-0.3$ & 21.1 & 5.4 & 1.9 & 0.5 & 1.3 & 1.6 & 6.3 & 4.7 & 17.9 & 4.9 & 2.4 & 1.0 & 1.8 & 2.7 & 6.8 & 5.5\\
 & $-0.5$ & 31.4 & 5.2 & 1.3 & 0.2 & 1.3 & 1.9 & 8.3 & 7.0 & 28.2 & 5.6 & 1.9 & 0.5 & 1.4 & 2.1 & 6.0 & 5.2\\
  \hline
&RMSE & 18.9 & 2.2 & 6.5 & 7.4 & 6.8 & 2.7 & 2.5 & 1.8 & 15.1 & 1.2 & 3.2 & 3.6 & 3.3 & 2.7 & 1.0 & 1.2\\
\bottomrule
\end{tabular} 
\end{table}

\begin{figure}[t]
\centering
\includegraphics[width=\linewidth]{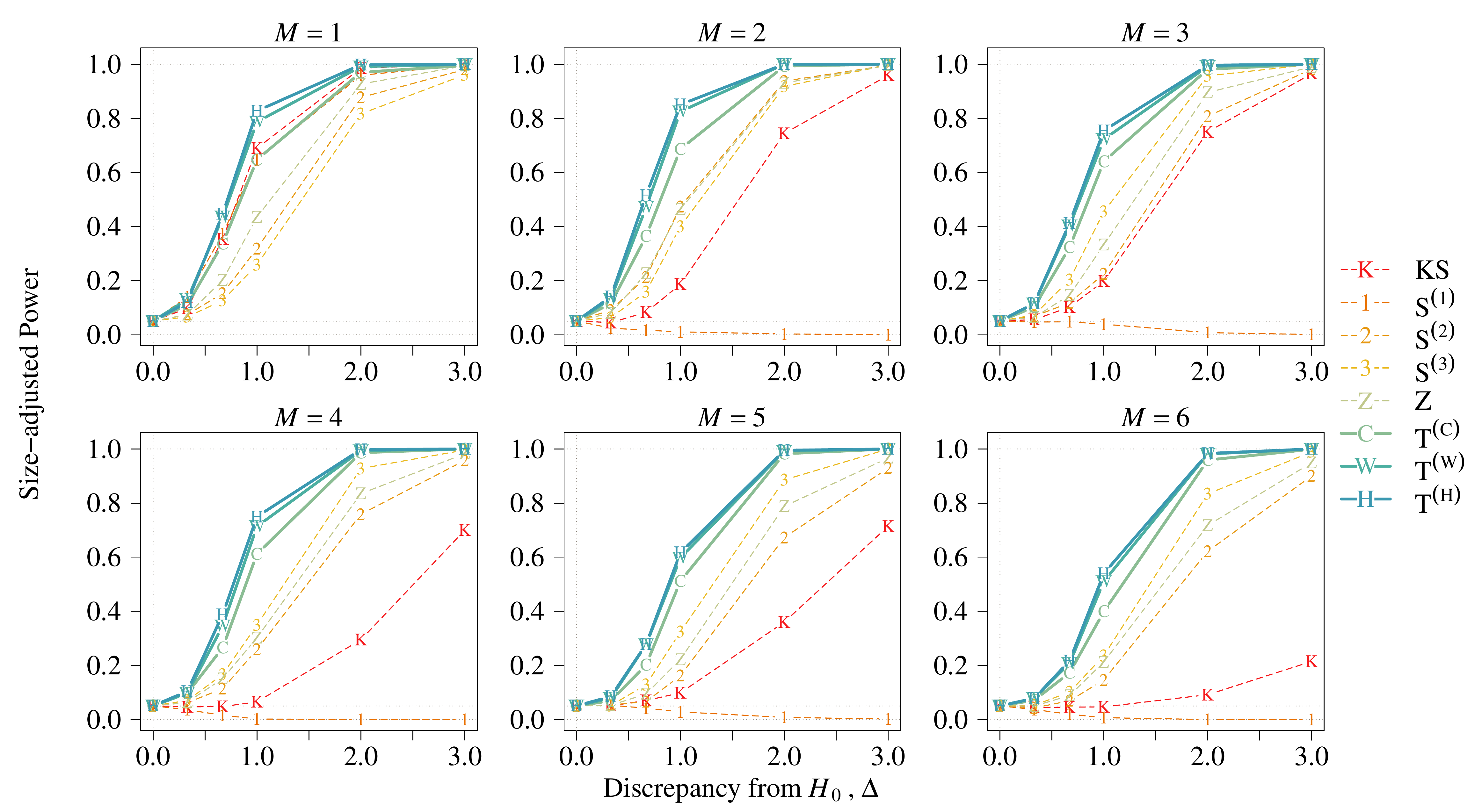}
\vspace{-0.3cm}
\caption{\label{fig: size adjusted power ar0.5g0.5} 
\footnotesize Size-adjusted power 
under BAR model with $\varpi = \vartheta = 0.5$, $n=200$ and mean function (\ref{EQ: CP models, mean level}).}
\end{figure}

The power curve against $\Delta$ is computed at 5\% nominal size
by using $2^{10}$ replications with $n = 200$. Figure \ref{fig: size adjusted power ar0.5g0.5} 
shows the size-adjusted power when $\varpi = \vartheta=0.5$.
The (unadjusted) power are presented in appendix.
We also repeat the experiments with different values of $\varpi$ and $\vartheta$. 
The results are similar, thus, they are deferred to the appendix.
In general, $\T_n^{(\W)}$ and $\T_n^{(\H)}$ have the highest power for all the cases. Their outperformance in power over $\T_n^{(\C)}$ is again attributed to the robustness against outliers. If only CUSUM-type tests are considered, the non-self-normalized test $\KS_n$ performs better than self-normalized tests when $M=1$,
however, it seriously under-performs when $M>1$ as $\KS_n$ is 
tailor-made for the one-CP alternative (\ref{eqt:H1_oneCP}). 
Although $\shao_n^{(1)}$ has the highest power among the self-normalized CUSUM tests when $M=1$, 
it suffers from the notorious non-monotonic power issue (with respect to $\Delta$) when $M>1$.
It is because its SN is not robust to the multiple CPs. 
Thus,
$\shao_n^{(1)}$ is no longer a consistent test when $M>1$.

It is surprising that our proposal $\T_n^{(\C)}$ outperforms 
\cite{shao2010testing}'s test $\shao^{(m)}_n$ even when 
the number of CPs is well-specified, i.e., $M=m$.  
The message behind this observation is that knowing the true number of CP does not give advantage because one can identify a structural change 
by observing the data around it without looking at the whole segmented series. 
It, indeed, partially motivates our localized framework. 
Moreover, 
as we discussed in Section \ref{sub-sec: test comparison discussion},
$\shao_n^{(m)}$ defines the local windows restrictively. 
It accumulates errors if the boundaries of the local windows are off from the actual CPs. 
It is also interesting to see that, compared with $\shao_n^{(1)}$, 
the tests $\shao_n^{(2)}$ and $\shao_n^{(3)}$ 
are less sensitive to misspecification of $M$
although they are still less powerful than our proposals.

\subsection{Ability to capture effects of all changes}\label{sec:sim_allCPsEffect}
This subsection investigates why our proposed 
tests have a higher power in a broader CP structures than the existing unsupervised multiple-CP test $\zhang_n$. Consider the 3-CP setting with the magnitudes of changes at the first and the last CPs are only half of the second CP. Specifically, we consider two mean functions: 

\begin{itemize}[noitemsep]
	\item (Case 1) equal magnitude: $\mu_i = \Delta \left\{ \mathbb{1}_{\{i > \lfloor  n/4 \rfloor\}} - \mathbb{1}_{\{i > \lfloor  2n/4 \rfloor\}} + \mathbb{1}_{\{ i > \lfloor  3n/4 \rfloor\}} \right \}$; and 
	\item (Case 2) unequal magnitude: $\mu_i = \Delta \left\{ \frac{1}{2}\cdot \mathbb{1}_{\{i > \lfloor  n/4 \rfloor\}} - \mathbb{1}_{\{i >  \lfloor  2n/4  \rfloor\}} + \frac{1}{2}\cdot \mathbb{1}_{\{i > \lfloor  3n/4  \rfloor\}} \right \}.
$
\end{itemize}

Figure \ref{fig: size adjusted power ar0.5g0.5, S3CP} shows the result when $\varpi=\vartheta=0.5$. The result under other time series models is provided in the appendix. Compared to Case 1, $\zhang_n$ and $\shao_n^{(2)}$ lose roughly half of the power in Case 2. In contrast, the power of our tests only reduces by $1/3$ when $\Delta \leq 1$ while remains about the same when $\Delta \geq 2$. Therefore, our tests are more powerful and are more robust to mean change structures than $\zhang_n$. It is because our proposals are able to capture all change structures whereas $\zhang_n$ only takes the first and the last CPs into account; see (\ref{eqt:zhang}).

\begin{figure}[t]
\centering
\includegraphics[width=0.7\linewidth]{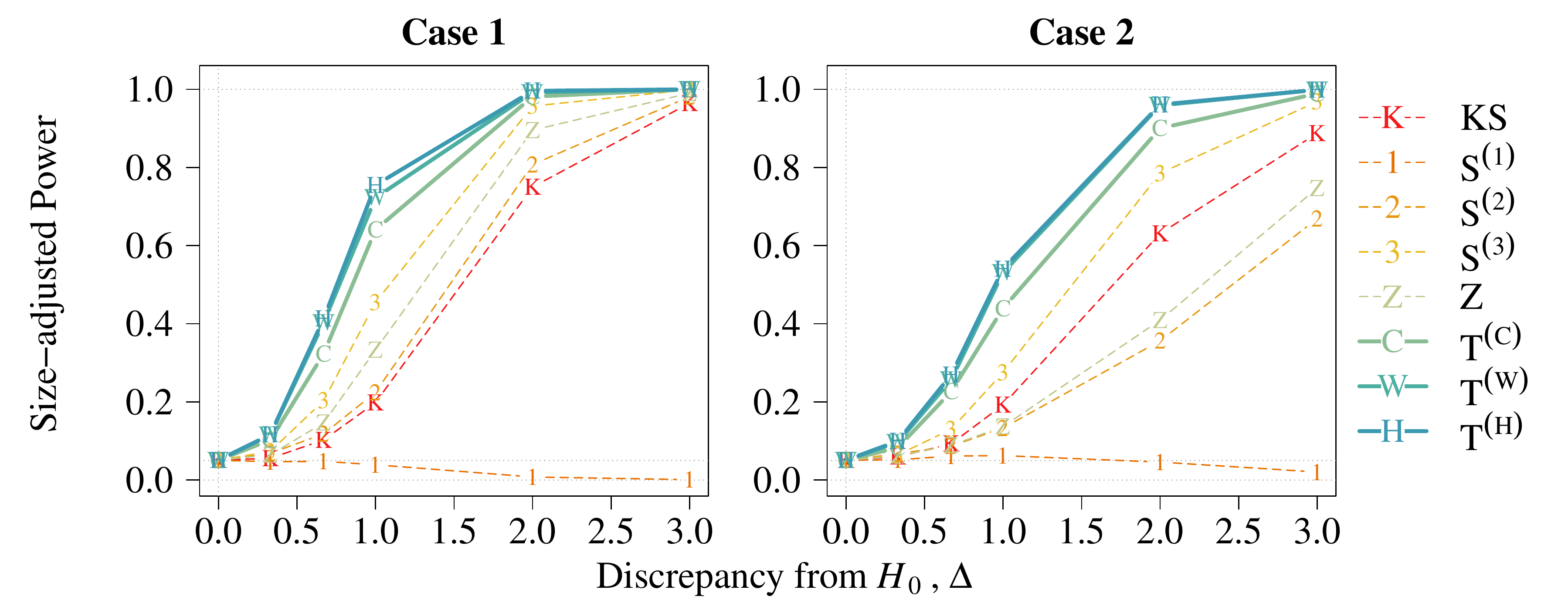}
\vspace{-0.3cm}
\caption{\footnotesize \label{fig: size adjusted power ar0.5g0.5, S3CP} Size-adjusted power under BAR model with $\varpi = \vartheta = 0.5$, $n=200$ and mean functions in Cases 1 and 2.}
\end{figure}

\subsection{Tradeoff of using symmetric windows}
\label{sec: justify symm windows}
To balance computational cost and power, 
we propose to use symmetric windows when constructing the LSN statistic; see Remark \ref{rem:symmetricWindow}. 
Intuitively, changes can still be detected by the contrast between samples of the same size.
In this subsection, we study the power loss of $\T_n^{(\C)}$ due to 
the use of symmetric windows.
We consider three-CP alternative with changes even and unevenly spaced in order to compare $\T_n^{(\C)}$ and $\tilde{\T}_n^{(\C)}$. Specifically, let the mean function with changes evenly spaced be Case 1 defined in Section \ref{sec:sim_allCPsEffect}, while the unevenly spaced CP model is defined as

\begin{itemize}[noitemsep]
	\item (Case 3) unevenly spaced CPs: $\mu_i = \Delta \left\{ \mathbb{1}_{\{i > \lfloor  0.15n \rfloor\}} - \mathbb{1}_{\{i > \lfloor  0.3n \rfloor\}} + \mathbb{1}_{\{ i >  \lfloor  0.85n \rfloor\}} \right \}$.
\end{itemize}

Figure \ref{fig: compare nonsym window, ar0.5g0.5} 
presents the size-adjusted power of $\T_n^{(\C)}$ and $\tilde{\T}_n^{(\C)}$
when $\varpi=\vartheta=0.5$. More simulation result is provided in the appendix. In both cases, $\T_n^{(\C)}$ performs similar to $\tilde{\T}_n^{(\C)}$. This indicates that using nonsymmetric windows does not bring significant advantage in power while it drastically increases computational burden; see Section \ref{sub-sec: test comparison discussion}. 
Therefore, we recommend symmetric windows to balance the power and the computational cost.
\begin{figure}[t]
\centering
\includegraphics[width=0.7\linewidth]{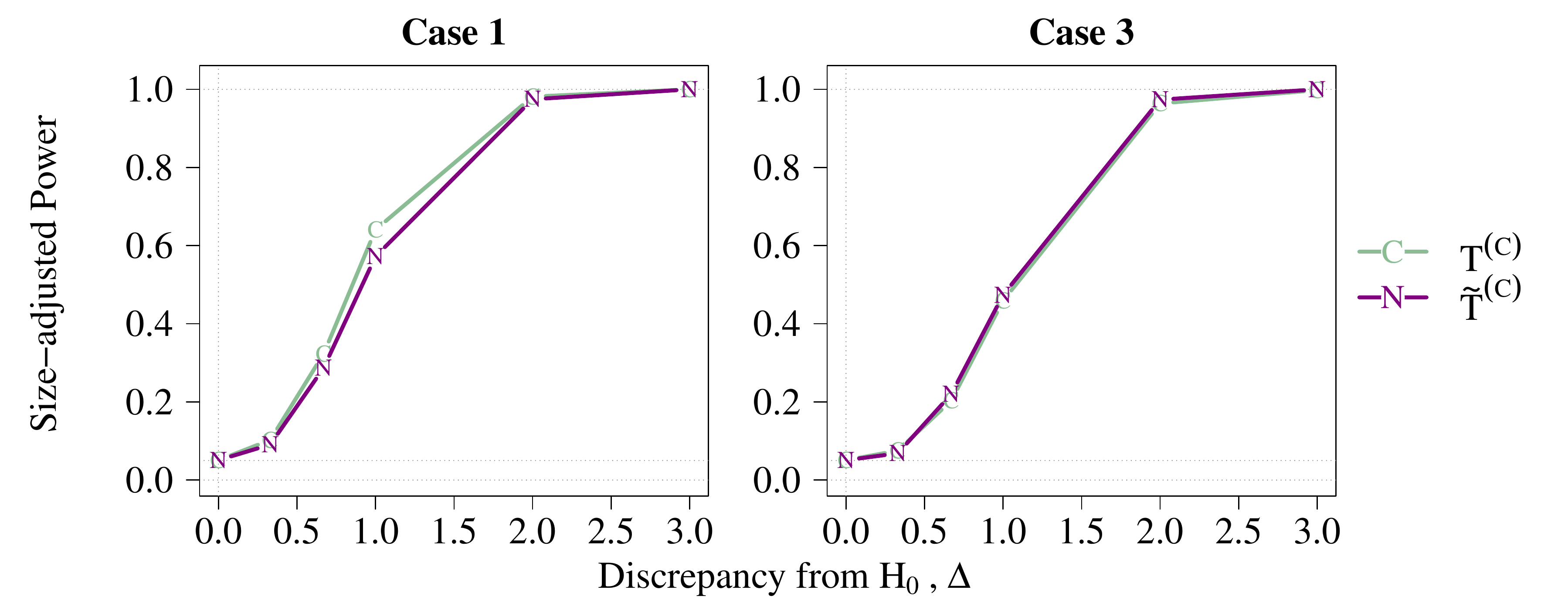}
\vspace{-0.3cm}
\caption{\footnotesize \label{fig: compare nonsym window, ar0.5g0.5} Size-adjusted power under BAR model with $\varpi = \vartheta = 0.5$, $n=200$ and mean functions in
Cases 1 and 3.}
\end{figure}

\section{Real data analysis}
\label{sec: real data analysis}
\subsection{NASDAQ call option volume}
\label{sec: nasdaq call vol}
We perform CP analysis for the Bloomberg's NASDAQ call option volume index (Bloomberg code: OPCVNQC index) from 24 March 2017 to 21 March 2022 ($n=1259$). The data are retrieved from Bloomberg\footnote{https://www.bloomberg.com/professional/product/market-data/}. Researchers and practitioners are interested in using the option volume to analyze the stock market. 
\cite{pan2006information} investigated the market price adjustment and identified strong evidence of informed trading in the option market. \cite{johnson2012option} found that option-to-stock-volume (O/S) ratio is significant in predicting the magnitude and sign of abnormal return resulting from earning surprises. 
It is consistent to the finding that the option volume reflects private information. Further divided the option volume into categories, \cite{ge2016does} founded that higher O/S ratio in general predicts lower stock return since more components of option volume negatively predict return.

Since the option volume is a significant factor of the market analysis, CP analysis on the constant mean option volume is thus of direct interest.
The test statistics and $p$-values are presented in Table \ref{tab: Nasdaq call  mean test result}. Since the result of our proposed tests reveal strong evidence against the hypothesis that no CPs exist, we further conduct mean change point location estimation, defined in \eqref{EQ: cp set pt estimate}, of which the estimates are indicated by red vertical lines in Figure \ref{fig:  Nasdaq call mean CP location}. Using the CUSUM process, the estimated CP dates are 27 November 2017, 30 January 2020, and 2 November 2020. The test using Wilcoxon statistics lead to similar results, while the test using Hodges--Lehmann statistics does not detect the first CP. The most dominant CP is the second CP at the end of January 2020. Studies have found a surge in trading activities worldwide and observed the trading from home effect \citep{chiah2020trading,ortmann2020covid}. \cite{chiah2020trading} explained the sudden growth of trading activities could be attributed to some retail investors regarded the stock market as gambling substitute. The jump in trading volume may also result from changing growth expectation and reaction to the post COVID-19 stimulus policies from the investors \citep{gormsen2020coronavirus}. 

\begin{table}[t]
\setlength{\tabcolsep}{3pt}
\centering
\footnotesize
\caption{\footnotesize \label{tab: Nasdaq call mean test result}Multiple-CP test on the mean of the NASDAQ Call Option Volume.}
\begin{tabular}{cc ccc c ccc}
\toprule
& $\KS_n$ & $\shao_n^{(1)}$ & $\shao_n^{(2)}$ & $\shao_n^{(3)}$ & $\zhang_n$ & $\T_n^{(\C)}$ & $\T_n^{(\W)}$ & $\T_n^{(\H)}$\\
\cmidrule(r){1-9}
Test statistic & 25.24 & 460.06 & 467.18 & 25.03 & 910.33 & 44.94 & 52.67 &31.46\\
$p$-value & $\leq 10^{-3}$ & $\leq 10^{-3}$ & $\leq 10^{-3}$ & $\leq 10^{-3}$& $\leq 10^{-3}$ & $\leq 10^{-3}$ & $\leq 10^{-3}$ & $\leq 10^{-3}$\\
\bottomrule
\end{tabular}  
\end{table}

\begin{figure}[t]
\centering
\includegraphics[width=1\linewidth]{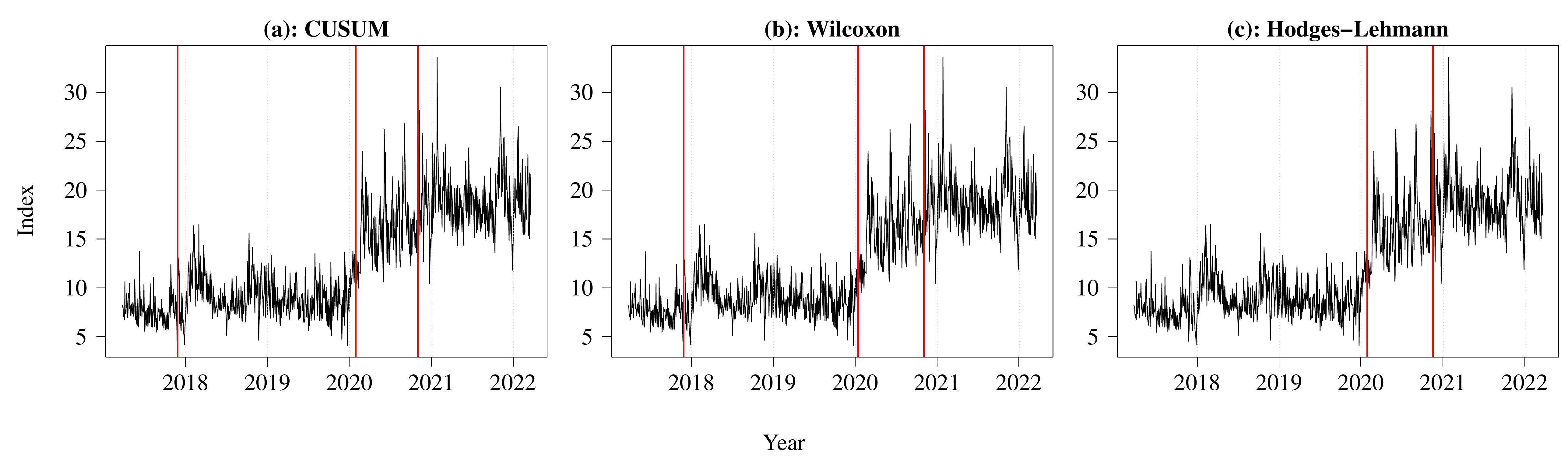}
\vspace{-0.3cm}
\caption{\footnotesize \label{fig: Nasdaq call mean CP location} 
Estimated change point locations are indicated by red vertical lines
by using (a) CUSUM process, (b) Wilcoxon statistics and (c) Hodges--Lehmann statistics. Data retrieved from Bloomberg}
\end{figure}
 
\subsection{Shanghai-Hong Kong Stock Connect turnover}
\label{sec: SHHK stock connect}
We apply our proposed method to perform CP analysis for the Shanghai-Hong Kong Stock Connect Southbound Turnover index (Bloomberg code: AHXHT index) from 23 March 2017 to 22 March 2022 ($n=1232$). The turnover index is based on the Hang Seng Stock Connect China AH (H) index which composites the most liquid shares eligible for the trading mechanism and listed in Hong Kong. The data are also retrieved from Bloomberg. 
The Stock Connect is a channel that allows both Mainland China and Hong Kong investors to access other stock markets mutually. The southbound is the direction for Shanghai investors to invest in Hong Kong stock market. Studies have shown that the Stock Connect improves mutual markets liquidity and capital market liberalization \citep{bai2017shanghai,huo2017return,xu2020stock}. Recently, the effect of the Stock Connect on the market risk and volatility has caught attention. \cite{ruan2018financial} found an increasing cross correlation between Shanghai and Hong Kong stock markets, suggesting higher risks after the implementation of the Stock Connect. \cite{huo2017return} also studied that the cointegration between the stock markets and concluded that volatility spillover effect is strengthened. The result from \cite{chan2018stock} further supported that the Stock Connect turnover is significant to the market volatility. As an indicator of foreign investment, the Stock Connect turnover is important for market volatility modeling. Therefore, we study whether CP exists in the Stock Connect Southbound daily turnover.

\begin{figure}[t]
\centering
\includegraphics[width=0.9\linewidth]{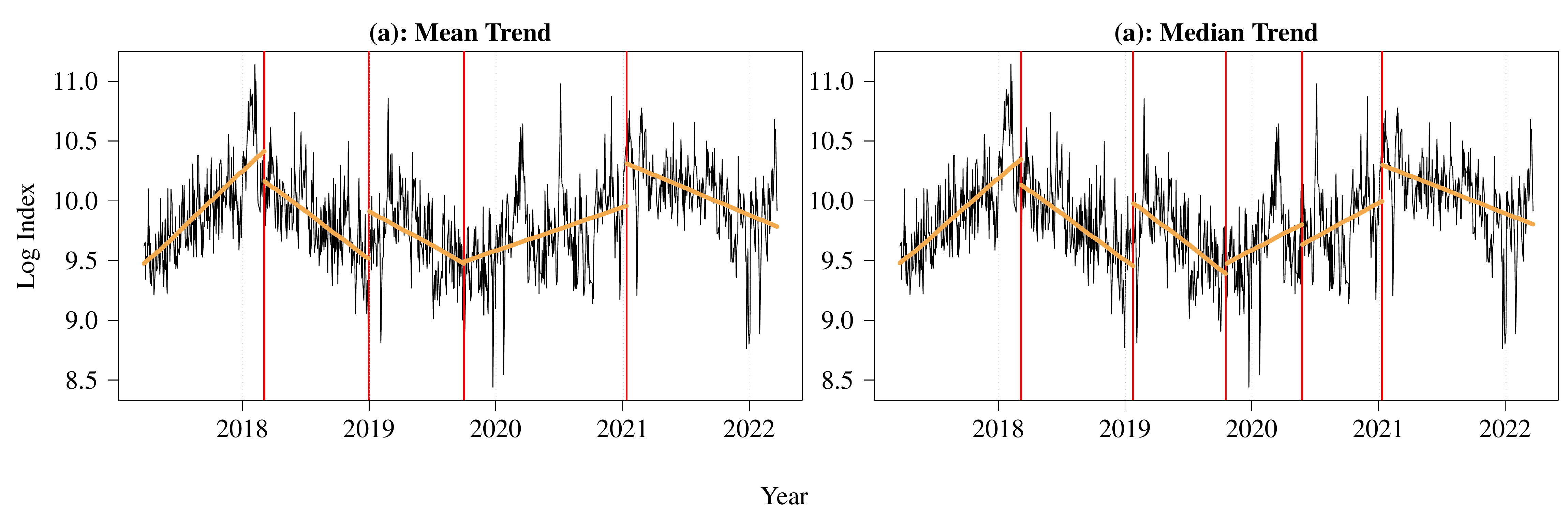}
\vspace{-0.3cm}
\caption{\footnotesize \label{fig: Stock connect trend CP location} 
Estimated (a) mean trend and (b) median trend CPs are indicated by red vertical lines. The orange lines indicate the fitted regression lines within each region separated by the estimated CPs. Data are retrieved from Bloomberg.}
\end{figure}

We investigate whether changes in trend exist in the log daily turnovers. 
Under our framework, it suffices to define a global change detecting process. 
It is, therefore, straightforward for practitioners. 
We consider the global change detecting process defined 
in \eqref{EQ: global parameter process} with 
two possible choices of $\hat{\theta}_{1:k}$ and $\hat{\theta}_{k+1:n}$
specified below.
\begin{enumerate}
    \item (Mean trend model) For each fixed $1\leq k\leq n$, 
    let $\E X_i = \left(\alpha_k^- + \beta_{k}^-\frac{i}{n}\right)\mathbb{1}_{\{i \leq k\}}+\left( \alpha_k^+ + \beta_{k}^+\frac{i}{n} \right)\mathbb{1}_{\{i > k\}} $
    for $i=1, \ldots, n$, 
    where $\alpha_k^-, \beta_k^-\in\mathbb{R}$ 
    are the intercept and slope of the mean trend at or before time $k$,  respectively, 
    whereas $\alpha_k^+, \beta_k^+\in\mathbb{R}$ are corresponding values after time $k$.
    Then $\hat{\theta}_{1:k}$ and $\hat{\theta}_{k+1:n}$ are defined as the ordinary least-squares estimators of $\beta_k^-$ and $\beta_k^+$, respectively. 
    \item (Median trend model) For each fixed $1\leq k\leq n$, let $\median(X_i) = \left(\alpha_k^- + \beta_{k}^-\frac{i}{n}\right)\mathbb{1}_{\{i \leq k\}}+\left( \alpha_k^+ + \beta_{k}^+\frac{i}{n} \right)\mathbb{1}_{\{i > k\}} $ for $i=1, \ldots, n$, where $\alpha_k^-, \beta_k^-, \alpha_k^+,\beta_k^+ \in\mathbb{R}$ are similarly interpreted as in the mean trend model. 
    Then $\hat{\theta}_{1:k}$ and $\hat{\theta}_{k+1:n}$ are  estimators of $\beta_k^-$ and $\beta_k^+$, respectively;
    see the modified Barrodale and Roberts algorithm \citep{koenker1987algorithm}. The estimators are computed by the R package \texttt{"quantreg"}.
    
\end{enumerate}
The resulting test statistics are 47.46 and 40.64. Both $p$-values are less than 1\%. The trend change point location estimates
are indicated by red vertical lines in Figure \ref{fig: Stock connect trend CP location}. 
The mean and median trend analysis agree on four estimated CPs, which are estimated on 7 March 2018, 24 January 2019, 2 October 2019 and 12 January 2021 under the median trend analysis. In same time, it further detects one CP on 25 May 2020. The first CP is likely to be the China-United States trade war. After the trade war began in January 2018, the Stock Connect turnover became downtrend. Until the COVID-19 events, an uptrend is detected after the end of 2019. Similar to the finding in the NASDAQ call option volume, the trading from home effect is observed. The uptrend stopped after the beginning of 2021.

\section{Conclusion}
\label{sec:conc}
Our proposed method improves existing CP tests and brings several advantages: 
(i) it has high power and controls size well, 
(ii) neither specification of number of CPs nor consistent estimation of nuisance parameters is required, 
(iii) change point location estimate can be naturally produced, 
(iv) general change detecting statistics, e.g., rank and order statistics, are allowed 
to enhance robustness
(v) it is capable to test change in general parameter of interest other than mean; and 
(vi) both short-range and long-range dependence are supported. 
Table \ref{tab: tests comparison} summarizes the properties of the proposed tests.  
Besides, our proposal is driven by intuitive principles, so that as long as a one-change detecting statistic is 
provided, our proposal can generalize it to a multiple-change detecting statistic. We anticipate that our framework can be applied to non-time-series data, e.g., spatial data and 
spatial–temporal data, in future work.

\appendix 
\section{Proofs of theorems}
\begin{proof}[Proof of Theorems \ref{thm: CUSUM test} and \ref{thm:localH1}]
(i) Under $H_0$
and by continuous mapping theorem, Assumption \ref{as:invariance principle} implies that 
$\{\C_n(\lfloor nt \rfloor) : t\in[0,1]\} \Rightarrow \{\sigma (\BM(t)- t\BM(1)): t\in [0,1]\}$.
Note that $\T_n^{(\C)}$ is
a composite function of $\C_n(\cdot)$ 
via (\ref{EQ local CUSUM}), (\ref{EQ: CUSUM self normalizer}), (\ref{EQ: LSN stat CUSUM}), (\ref{EQ: CUSUM test score function}), and (\ref{EQ: CUSUM ts}), 
each of which is a continuous and measurable map. 
By continuous mapping theorem, 
we obtain $\T_n^{(\C)}\inD \mathbb{T}$ in 
(\ref{EQ: CUSUM ts limiting dist}).
The limiting distribution $\mathbb{T}$ 
is well-defined
because 
$\V^{(\C)}_n(k \mid k-d,k+1+d)$ converges to a non-negative and non-degenerate distribution for any 
$\lfloor n\epsilon \rfloor \leq d\leq n$ and 
$\lfloor \epsilon n \rfloor + 1 \leq k \leq n-\lfloor \epsilon n \rfloor - 1$, provided that $\epsilon>0$.

(ii) 
For consistency, we consider 
the $j$th relative CP time $\pi_{j}$ such that 
the corresponding change magnitude satisfies 
$|\Delta_j| \asymp n^{\kappa-1}$, where $0.5 < \kappa \leq 1$. 
Denote $\Delta_* = \Delta_j$ and $k_* = \lfloor n\pi_{j} \rfloor$.
Under the assumption, there exist $c \in \mathbb{R}^+$ and $3/2 - \kappa \leq \Upsilon < 1$ such that 
$\epsilon + \lfloor cn^{\Upsilon}\rfloor/n < \min(\pi_{j} - \pi_{j-1},\pi_{j+1} - \pi_{j})$ is satisfied for a large enough $n$. Therefore, there is only one CP in the interval $k_* \pm (\lfloor \epsilon n \rfloor + \lfloor cn^{\Upsilon}\rfloor)$. 
It suffices to consider $d = \lfloor \epsilon n \rfloor$. For all $k$ such that $|k-k_*| \leq \lfloor bn^{\Upsilon}\rfloor$ where $0 < b < c$, we can decompose
\begin{align}
    \L_n^{(\C)}(k \mid k - d, k+d+1)^2 &= \frac{d+1}{8}\left(\bar{X}_{(k-d):k} - \bar{X}_{(k+1):(k+d+1)} \right)^2\nonumber \\
    &= \frac{d+1}{8}\left(\left\{\bar{Z}_{(k-d):k} - \bar{Z}_{(k+1):(k+d+1)} \right\}- \left\{ 1-\frac{|k_* - k|}{d+1}\right\}\Delta_* \right)^2\nonumber \\
    &= \frac{1}{8}\left[\left(\sqrt{d+1}\left\{\bar{Z}_{(k-d):k} - \bar{Z}_{(k+1):(k+d+1)} \right\}\right)-\left\{ \sqrt{d+1}-\frac{\lfloor bn^{\Upsilon}\rfloor}{\sqrt{d+1}}\right\}\Delta_* \right]^2 \nonumber \\
    &\geq \left\{ O_p(1) + \Xi_{1,n} \right\}^2,\label{eqt:LnC_order}
\end{align}
where $|\Xi_{1,n}| \asymp n^{\kappa-1/2}$.
Next, for the SN $\V_n^{(\C)}(k \mid k - d, k+d+1)$, 
we consider two cases: $k \geq k_*$ and $k < k_*$. 
If $k \geq k_*$, there is only one CP in the interval $[k-d,k]$.
Following similar calculations as in \eqref{eqt:LnC_order},
we know that $\L_n^{(\C)}(j\mid k-d,k)^2 \leq \L_n^{(\C)}(k_* \mid k-d, k)$
for all $j = k-d,\ldots,k$, provided that $n$ is large enough.
Moreover, since there is no CP in the interval $[k+1,k+1+d]$,
$\L_n^{(\C)}(j \mid k+1,k+1+d)^2 = O_p(1)$ for all $j=k+1, \ldots, k+1+d$.
Therefore,
\begin{align*}
    \V_n^{(\C)}(k \mid k - d, k+d+1) &= \frac{1}{4(d+1)}\left\{\sum_{j = k-d}^k \L_n^{(\C)}(j\mid k-d,k)^2 + \sum_{j = k+1}^{k+1+d} \L_n^{(\C)}(j\mid k+1,k+1+d)^2  \right\}\\
    &\leq \frac{1}{4}\L_n^{(\C)}(k_* \mid k-d, k)^2 + O_p(1)\\
    &= \frac{(d+1-\lfloor bn^{\Upsilon} \rfloor)^2(\lfloor bn^{\Upsilon} \rfloor)^2}{2(d+1)^3}\left\{\bar{Z}_{(k-d):k_*} - \bar{Z}_{k_*:k} - \Delta_* \right\}^2 + O_p(1)\\
    &\leq \left\{ O_p(1) + \Xi_{2,n} \right\}^2 ,
\end{align*}
where $|\Xi_{2,n}| \asymp n^{\Upsilon+\kappa -3/2}$.
For the case when $k < k_*$, the analysis is similar and $\V_n^{(\C)}$ are of the same order. From above, there exists a constant $c' \neq 0$ such that 
\[
    \T^{(\C)}_n(k) 
    \geq \left\{\frac{O_p(1) + \Xi_{1,n}}{O_p(1) + \Xi_{2,n}}\right\}^2
    = \left\{ o_p(1) + c' n^{1-\Upsilon} \right\}^2,
\]
for all $k = k_* - \lfloor bn^{\Upsilon}\rfloor, \ldots, k_* + \lfloor bn^{\Upsilon}\rfloor$.
Since $\T_n^{(\C)}(k) \geq 0$ for all $k$, we have
\begin{align}\label{eqt:boundTnC}
	\T^{(\C)}_n 
		\geq n^{-1}\sum_{k = k_* - \lfloor bn^{\Upsilon}\rfloor}^{k_*+ \lfloor bn^{\Upsilon}\rfloor} \T^{(\C)}_n(k)
		= c'' n^{1-\Upsilon} \{o_p(1)+1\}^2,
\end{align}
for a large enough $n$, where $c'' > 0$
is a constant. Consequently,  $\T^{(\C)}_n \rightarrow \infty$ in probability as $n \rightarrow \infty$ because $\Upsilon<1$.
\end{proof}

\begin{proof}[Proof of Theorem \ref{thm: rho consistency}]
Define $D_t = (X_{t+b} - X_t)/\sqrt{2}$ for $t = 1, \ldots, N$, where $N = n-b+1$. 
The ACF estimator at lag 1 based on $\{D_t \}$ is 
$\widehat{\rho}=\widehat{\gamma}(1)/\widehat{\gamma}(0)$, where
$\widehat{\gamma}(k) = \sum_{t = 1}^{N-k}(D_t - \bar{D})(D_{t+k} - \bar{D})/N$ for $k=0,1$, 
and $\bar{D} = N^{-1}\sum_{t = 1}^N D_t$.
Define $M_t^{(\alpha)} = \alpha_{t+b} - \alpha_t -  \left( \bar{\alpha}_{1+b,N+b} - \bar{\alpha}_{1,N}\right)$
for $\alpha \in \{\mu, Z\}$, 
where $\bar{\alpha}_{s,e} = (e-s+1)^{-1}\sum_{i = s}^e \alpha_i$. 
Then, we rewrite
\begin{equation}
\label{EQ: gamma 1 decomposition}
\widehat{\gamma}(1) = \frac{1}{2N}\sum_{t = 1}^{N-1} (M_t^{(Z)} + M_t^{(\mu)})(M_{t+1}^{(Z)} + M_{t+1}^{(\mu)}) = \frac{1}{2}(\widehat{\gamma}_N^{(ZZ)} + \widehat{\gamma}_N^{(\mu Z)} + \widehat{\gamma}_N^{(Z\mu)} +\widehat{\gamma}_N^{(\mu\mu)}),
\end{equation}
where $\widehat{\gamma}_N^{(\alpha\beta)} = N^{-1}\sum_{t = 1}^{N-1} M_t^{(\alpha)}M_{t+1}^{(\beta)}$ and $\alpha, \beta \in \{\mu, Z\}$. 
We derive the probability limits of $\widehat{\gamma}_N^{(ZZ)}$, $\widehat{\gamma}_N^{(\mu Z)}$, 
$\widehat{\gamma}_N^{(Z\mu)}$ and $\widehat{\gamma}_N^{(\mu\mu)}$ one by one. 
First, for $\widehat{\gamma}_N^{(ZZ)}$, we have
\begin{align*}
\widehat{\gamma}_N^{(ZZ)} &= \frac{1}{N}\sum_{t = 1}^{N-1}(Z_{t+b} - \bar{Z}_{1+b,N+b})(Z_{t+b+1} - \bar{Z}_{1+b,N+b}) - \frac{1}{N}\sum_{t = 1}^{N-1}(Z_{t+b} - \bar{Z}_{1+b,N+b})(Z_{t+1} - \bar{Z}_{1,N})\\
&\quad - \frac{1}{N}\sum_{t = 1}^{N-1}(Z_{t} - \bar{Z}_{1,N})(Z_{t+b+1} - \bar{Z}_{1+b,N+b}) + \frac{1}{N}\sum_{t = 1}^{N-1}(Z_{t} - \bar{Z}_{1,N})(Z_{t+1} - \bar{Z}_{1,N})\\
&=: \widehat{\gamma}_{N,1}^{(ZZ)} - \widehat{\gamma}_{N,2}^{(ZZ)}  - \widehat{\gamma}_{N,3}^{(ZZ)}  + \widehat{\gamma}_{N,4}^{(ZZ)}.
\end{align*}
Note that $\sum_{i = 1}^{\infty} \lambda_4(i) < \infty$ implies $\sum_{k = 1}^{\infty} |\gamma(k)| < \infty$. 
By Theorem 7 in \cite{wu2011asymptotic}, ergodic theorem and Slutsky's lemma,
\begin{align*}
\widehat{\gamma}_{N,1}^{(ZZ)} = \frac{1}{N}\sum_{t = 1}^{N-1} Z_{t+b}Z_{t+b+1} - \bar{Z}_{1+b,N+b}\left\{\frac{1}{N}\sum_{t = 1}^{N-1} \left( Z_{t+b} + Z_{t+1+b} \right) - \left(1- \frac{1}{N}\right)\bar{Z}_{1+b,N+b} \right\} \inP \gamma(1),
\end{align*}
where ``$\inP$'' denotes convergence in probability. Similarly, we have $\widehat{\gamma}_{N,4}^{(ZZ)}  \inP \gamma(1)$. The asymptotic behavior for sample autocovariance function at a divergent lag is different to that at a fixed lag. 
Nevertheless, by Theorem 8 in \cite{wu2011asymptotic}, ergodic theorem and Slutsky's lemma, 
we have $\widehat{\gamma}_{N,2}^{(ZZ)} - \gamma(b-1) \inP 0$ and $\widehat{\gamma}_{N,3}^{(ZZ)} - \gamma(b+1) \inP 0$. By Slutsky's lemma again, we have $\widehat{\gamma}_N^{(ZZ)} \inP 2\gamma(1)$. 
Next, we consider $\widehat{\gamma}_N^{(\mu Z)}$ and $\widehat{\gamma}_N^{(Z\mu)}$. 
Note that $M_t^{(\mu)}$ is deterministic and by Markov inequality. 
So,  for all $\epsilon > 0$, we have
\begin{equation*}
\pr(|\widehat{\gamma}_n^{(\mu Z)} | > \epsilon) \leq \frac{\E(|\widehat{\gamma}_n^{(\mu Z)} |)}{\epsilon} \leq \frac{1}{2N\epsilon}\sum_{t = 1}^{N-1} \big[ \big(M_t^{(\mu)}\big)^2 + \E\big\{\big(M_{t+1}^{(Z)}\big)^2 \big\} \big].
\end{equation*}
Since $\{\mu_{t+b} - \mu_t \}_{1 \leq t \leq N}$ has at most $bM$ non-zero elements, we have
\begin{equation}
\label{EQ average of M^mu square negligible}
0 \leq \frac{1}{N}\sum_{t=1}^{N-1} \big(M_t^{(\mu)}\big)^2 \leq \frac{1}{N}\sum_{t = 1}^N \big(\mu_{t+b} - \mu_t \big)^2 \leq \frac{bM}{N}\Delta^2 \rightarrow 0.
\end{equation}
Moreover, expanding $\sum_{t=1}^{N-1} \E \big\{\big(M_{t+1}^{(Z)}\big)^2 \big\}/N$, we have 
\begin{align*}
	 \frac{2}{N^2}\sum\limits_{t=2}^{N}\sum\limits_{i=1}^N \{\gamma(|t+b-i|) + \gamma(|t-b-i|)\} 
	- \frac{2N-2}{N}\left\{\gamma(b) + \frac{1}{N^2}\sum\limits_{i,j = 1}^N \gamma(|i - j - b|) \right\},
\end{align*}
which converges to 0 in probability by the absolute summability of $\gamma(k)$. 
Therefore, we have $\widehat{\gamma}_N^{(\mu Z)} \inP 0$. By similar arguments, we also have $\widehat{\gamma}_N^{(Z\mu)} \inP 0$. 
Then, we consider $\widehat{\gamma}_N^{(\mu\mu)}$. 
By the result in (\ref{EQ average of M^mu square negligible}), we have 
$
	0 \leq \widehat{\gamma}_N^{(\mu\mu)} 
	\leq \sum_{t=1}^{N-1} \big\{ \big(M_t^{(\mu)}\big)^2 + \big(M_{t+1}^{(\mu)}\big)^2 \big\} /N\rightarrow 0. 
$
Putting all these results back to (\ref{EQ: gamma 1 decomposition}),
we obtain $\widehat{\gamma}(1) \inP \gamma(1)$. For $\widehat{\gamma}(0)$, we have
\begin{equation}
\label{EQ: gamma 0 decomposition}
\widehat{\gamma}(0) = \frac{1}{2N}\sum_{t=1}^{N} \big\{ \big(M_t^{(Z)}\big)^2 + \big(M_t^{(\mu)}\big)^2 + 2M_t^{(Z)}M_t^{(\mu)} \big\}.
\end{equation}
Consider the first term of (\ref{EQ: gamma 0 decomposition}), we have, by Theorems 7 and 8 in \cite{wu2011asymptotic}, that
\begin{equation*}
\frac{1}{N}\sum\limits_{t=1}^{N} \big(M_t^{(Z)}\big)^2 = \frac{\sum_{t = 1}^N (Z_{t+b}^2 + Z_t^2 - 2Z_{t+b}Z_{t})}{N} - (\bar{Z}_{1+b,N+b} - \bar{Z}_{1,N})^2 \inP 2\gamma(0).
\end{equation*}
By similar arguments as above, we also have $N^{-1}\sum_{t=1}^{N} \big(M_t^{(\mu)}\big)^2  \rightarrow 0$ and $N^{-1}\sum_{t=1}^{N}  M_t^{(Z)}M_t^{(\mu)} \inP 0$ as  $b \rightarrow \infty$. Therefore, $\widehat{\gamma}(0) \inP \gamma(0)$ as $b \rightarrow \infty$. 
Hence, we have $\widehat{\rho} \inP \gamma(1)/\gamma(0) = \rho$.
\end{proof}

\section{Proofs of Corollaries}

\begin{proof}[Proof of Corollaries \ref{thm: wilcoxon test} and \ref{thm: HL test}] 
The proof is similar to the CUSUM case and we highlight the difference. 
By Theorem 3.1 in \cite{dehling2013SRD}, 
we have $\W_n(\lfloor nt \rfloor) \inD \sigma_{\W}^2\{\BM(t)- t\BM(1)\}$. 
Moreover by Theorem 1 in \cite{dehling2015robust}, 
we have $\{\H_n(\lfloor nt \rfloor) : t\in[0,1] \}\inD \{\left(\sigma_{\H}/u(0)\right)(\BM(t) - t\BM(1) \}) : t\in[0,1]\}$. 
The result follows by using similar argument as in the proof of Theorem \ref{thm: CUSUM test} and 
replacing $\sigma$ by $\sigma_{\W}^2$ for $\T_n^{(\W)}$ and $\sigma_{\H}/u(0)$ for $\T_n^{(\H)}$.
\end{proof}

\begin{proof}[Proof of Corollary \ref{thm: general parameter test}] 
Using \eqref{EQ: parameter decomposition}, we can rewrite
\begin{equation*}
\begin{split}
\G_n(k) =\frac{k(n-k)}{n^{3/2}}\left( \widehat{\theta}_{1:k} - \widehat{\theta}_{k+1:n}\right)
	= I_n(t)  - \frac{k}{n}I_n(1) + \frac{k(n-k)}{n^{3/2}}(R_{1:k} - R_{k+1:n}).
\end{split}
\end{equation*}
The last term vanishes uniformly over $k$ as $n \rightarrow \infty$ since $\sup_{k = 1, \ldots, n-1} \left| R_{1:k}\right| + | R_{(k+1):n}|= o_p(n^{-1/2})$. 
Therefore, $\{ \G_n(\lfloor nt \rfloor) : t\in[0,1] \}\inD \{\sigma (\BM(t) - t\BM(1)):t\in[0,1]\}$ 
by the continuous mapping theorem. The result follows by using similar arguments as in the proof of Theorem \ref{thm: CUSUM test}.
\end{proof}

\begin{proof}[Proof of Corollaries \ref{thm: LRD test} and \ref{thm: multivartiate test}] 
Note that $\{\C_n^*(\lfloor nt \rfloor): t\in[0,1] \} \Rightarrow \{\sigma_{\C^*}\BM_H(t) - t\BM_H(1)\}: t\in[0,1] $ and 
$\{\W_n^*(\lfloor nt \rfloor): t\in[0,1] \} \Rightarrow \{(2\sqrt{\pi})^{-1}\BM_H(t) - t\BM_H(1) : t\in[0,1] \}$. Moreover, $\T_n^{(\C^*)}$ and $\T_n^{(\W^*)}$ is a continuous and measurable function of $\C_n^*(\cdot)$ and $\W_n^*(\cdot)$ respectively. The above arguments also apply and the limiting distribution is a functional of standard fractional Brownian motion $\BM_H(t)$ instead of the standard Brownian motion $\BM(t)$.

Regarding to $\T_n^{(\Q)}$, since 
$\{\Q_n(\lfloor(nt\rfloor): t\in[0,1] \} \Rightarrow \{\sigma_{\Q}\{\BM^q(t) - t\BM^q(1) \} : t\in[0,1] \}$ and 
$\T_n^{(\Q)}$ is a continuous and measurable function of $\Q_n(\cdot)$, we have the result by 
continuous mapping theorem.
\end{proof}

\section{Supplementary Simulation Result}
\subsection{Setting}
In this session, the size and power of our proposed tests ($\T_n^{(\C)}$, $\T_n^{(\W)}$ and $\T_n^{(\H)}$), \cite{shao2010testing} proposed tests ($\shao_n^{(1)}$, $\shao_n^{(2)}$ and $\shao_n^{(3)}$) as well as \cite{zhang2018unsupervised} proposed test ($\zhang_n$) are further investigated under the standard autoregressive (AR) model, autoregressive-moving-average (ARMA) model, threshold AR (TAR) model and absolute value nonlinear AR (NAR) model in Sections \ref{sec: AR model}, \ref{sec: ARMA model}, \ref{sec: TAR model} and \ref{sec: NAR model} respectively. Extra results under the bilinear AR (BAR) model are also provided in \ref{sec: BAR model}. The performance for the tests under non-Gaussian time series is investigated in Section \ref{sec: non-g time series}. The power comparison between all the test concerned when the magnitudes of changes are not equal is reported in Section \ref{sec: Unequal change}. Last but not least, the power comparison after the use of nonsymmetric windows is presented in Section \ref{sec: Non-symmetric windows}. All results are obtained using $2^{10}$ replications and the nominal size is 5\%. The sample size is $n = 200$ unless specified. We define $\varepsilon_i$ to be independent standard normal random variables for $i = 1, \ldots, n$. We again consider the signal-plus-noise model, $X_i = \mu_i + Z_i$, and let $\mu_i$ to be the signal, specified in equation \eqref{EQ: CP models, mean level} from the main article, i.e.,
\begin{equation}
\label{EQ mean function normal}
\mu_i = \Delta \sum_{j = 1}^M (-1)^{j+1} \mathbb{1}_{\{ i/n \geq j/(M+1) \}},
\end{equation}
where $M \in \{1,\ldots, 6\}$ denotes the number of CPs, and $\Delta\in\mathbb{R}$ is the magnitude of the mean changes. 

\subsection{AR Model}
\label{sec: AR model}
In this subsection, we investigate the performance under the AR(1) model. Specifically, the data has the form $X_i - \mu_i =  \varpi (X_{i-1} - \mu_{i-1}) + \varepsilon_i$, where the AR parameter $|\varpi|<1$. The null rejection rates are documented in Table \ref{tab: AR size}. Our proposed tests have accurate size in most cases. Although $\shao_n^{(1)}$ has more accurate size when $\varpi =  0.8$, our tests are not severely under-sized when $\varpi < 0$ as observed in other self-normalized tests. The non-self-normalized Kolmogorov--Smirnov (KS) test has larger size distortion than the self-normalized tests.

Figures \ref{fig: power ar0.5} and \ref{fig: power ar-0.5} present the power while Figures \ref{fig: size adjusted power ar0.5} and \ref{fig: size adjusted power ar-0.5} present the size-adjusted power. In the single CP case, our proposed tests lose the least power among multiple-CP tests comparing to $\shao_n^{(1)}$ and $\KS_n$, which are tailored for the at-most-one-change (AMOC) problem. However, $\shao_n^{(1)}$ suffers non-monotonic power in multiple-CP case. The issue is also observed in the KS test when $\varpi = -0.5$, yet it does not appear in the unadjusted power. Under multiple-CP case, our proposed tests have the largest power. In particular, $\shao_n^{(2)}$ and $\shao_n^{(3)}$ have lower power than our proposed tests even the number of CP is well-specified.

\begin{table}[t]
\setlength{\tabcolsep}{3pt}
\centering
\footnotesize
\caption{\footnotesize \label{tab: AR size} Null rejection rates (\%) at 5\% nominal size under AR model and mean function \eqref{EQ mean function normal}.}
\vspace{-0.3cm}
\begin{tabular}{ r  r rrr r rrr  r rrr r rrr}
\toprule
&  \multicolumn{8}{c}{$n=200$} & \multicolumn{8}{c}{$n=400$}\\
\cmidrule(r){2-9}\cmidrule(r){10-17}
$\varpi$& $\KS_n$ & $\shao_n^{(1)}$ & $\shao_n^{(2)}$ & $\shao_n^{(3)}$ & $\zhang_n$ & $\T_n^{(\C)}$ & $\T_n^{(\W)}$ & $\T_n^{(\H)}$ & $\KS_n$ &  $\shao_n^{(1)}$ & $\shao_n^{(2)}$ & $\shao_n^{(3)}$ & $\zhang_n$ & $\T_n^{(\C)}$ & $\T_n^{(\W)}$ & $\T_n^{(\H)}$  \\ 
\cmidrule(r){1-17}
0.8 & 53.6 & 9.0 & 25.9 & 43.3 & 27.2 & 16.1 & 23.9 & 21.8 & 52.6 & 6.0 & 12.9 & 23.4 & 14.6 & 9.3 & 12.7 & 12.1\\
0.5 & 13.7 & 5.7 & 9.8 & 10.9 & 10.4 & 5.5 & 7.3 & 6.5 & 9.4 & 4.5 & 5.7 & 8.3 & 5.4 & 5.1 & 5.5 & 5.1\\
0.3 & 10.1 & 5.4 & 6.7 & 5.6 & 6.8 & 4.5 & 5.3 & 5.4 & 7.7 & 4.2 & 5.0 & 5.8 & 4.8 & 5.2 & 5.0 & 4.8\\
0 & 7.4 & 5.2 & 3.5 & 2.5 & 3.1 & 4.1 & 4.7 & 4.6 & 5.0 & 4.2 & 4.5 & 4.6 & 3.8 & 5.2 & 4.7 & 4.2\\
$-0.3$ & 13.0 & 4.3 & 2.5 & 0.8 & 2.2 & 4.4 & 5.3 & 4.7 & 6.6 & 3.3 & 2.8 & 3.0 & 2.5 & 5.4 & 4.7 & 4.9\\
$-0.5$ & 23.7 & 3.7 & 1.5 & 0.4 & 1.4 & 5.0 & 5.8 & 5.3 & 14.7 & 3.1 & 2.1 & 1.9 & 2.1 & 5.0 & 4.4 & 4.8\\
$-0.8$ & 46.0 & 2.4 & 0.5 & 0.1 & 0.1 & 7.8 & 9.2 & 9.3 & 38.9 & 2.3 & 1.0 & 0.1 & 1.0 & 6.6 & 7.7 & 7.3\\

\bottomrule
\end{tabular} 
\end{table}

\begin{figure}[t]
\centering
\includegraphics[width=0.95\linewidth]{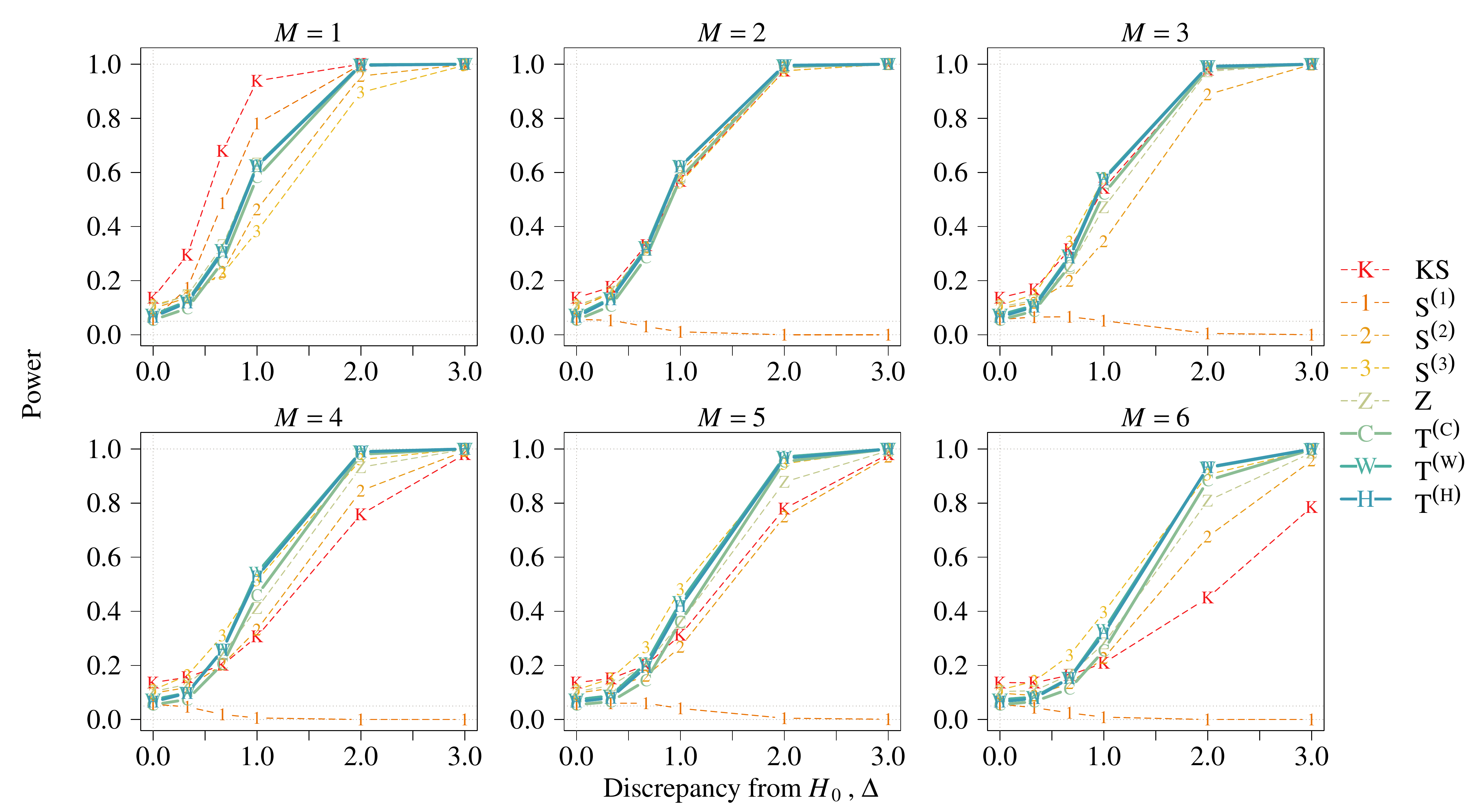}
\vspace{-0.3cm}
\caption{\footnotesize \label{fig: power ar0.5} Power under AR model with $\varpi = 0.5$ and mean function \eqref{EQ mean function normal}.}
\end{figure}

\begin{figure}[H]
\centering
\includegraphics[width=\linewidth]{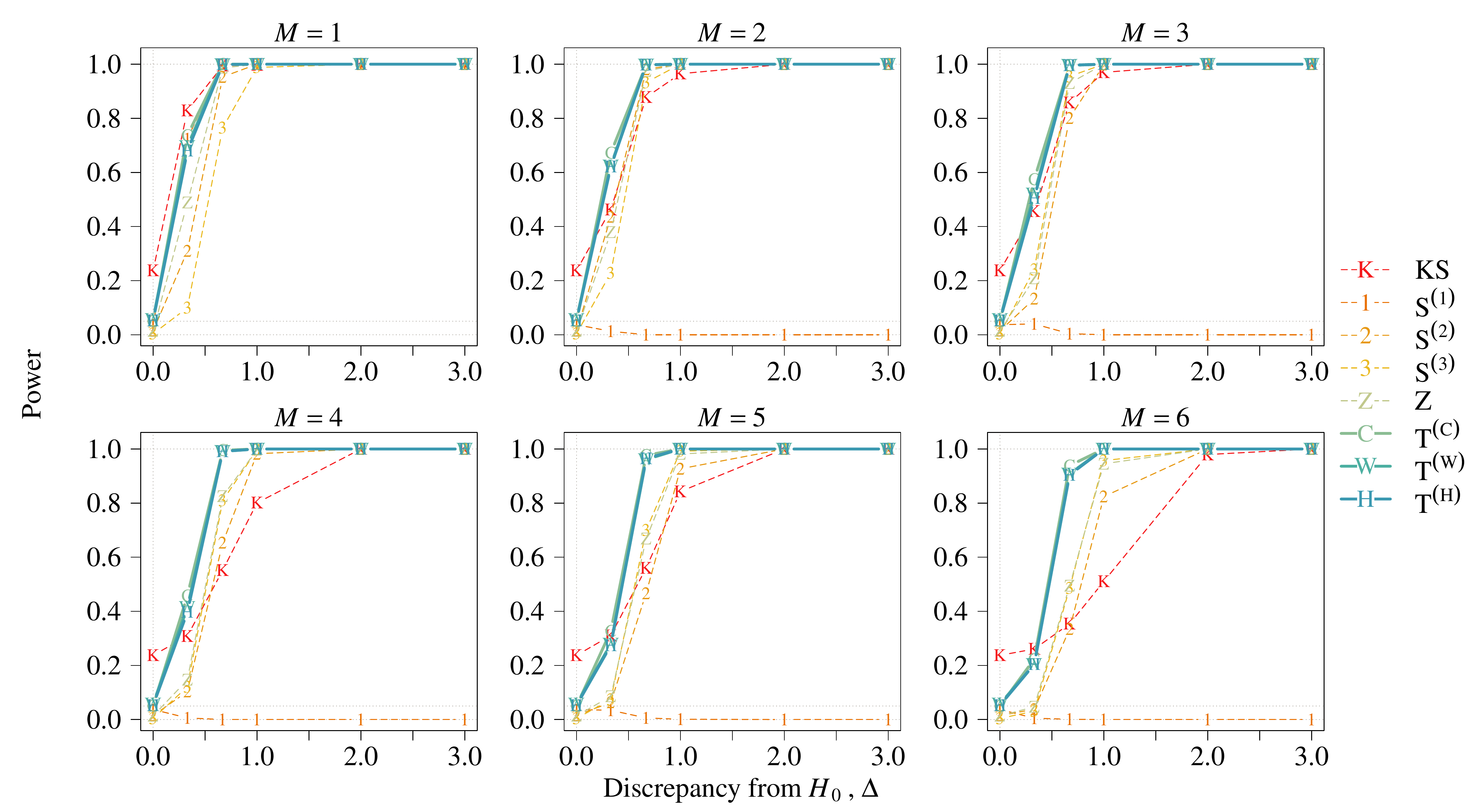}
\vspace{-0.3cm}
\caption{\footnotesize \label{fig: power ar-0.5} Power under AR model with $\varpi = -0.5$ and mean function \eqref{EQ mean function normal}.}
\end{figure}

\begin{figure}[H]
\centering
\includegraphics[width=\linewidth]{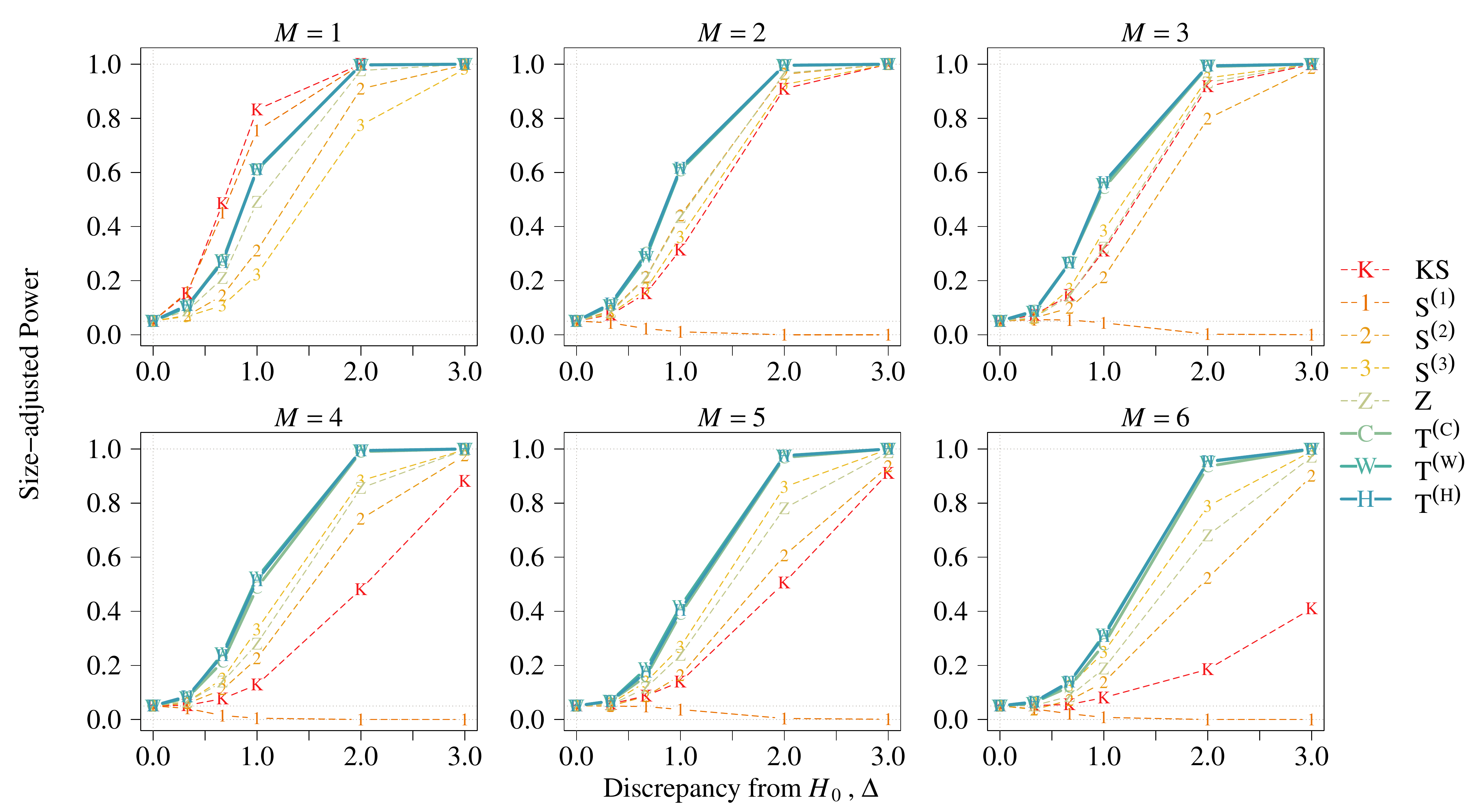}
\vspace{-0.3cm}
\caption{\footnotesize \label{fig: size adjusted power ar0.5} Size-adjusted power under AR model with $\varpi = 0.5$ and mean function \eqref{EQ mean function normal}.}
\end{figure}

\begin{figure}[H]
\centering
\includegraphics[width=0.95\linewidth]{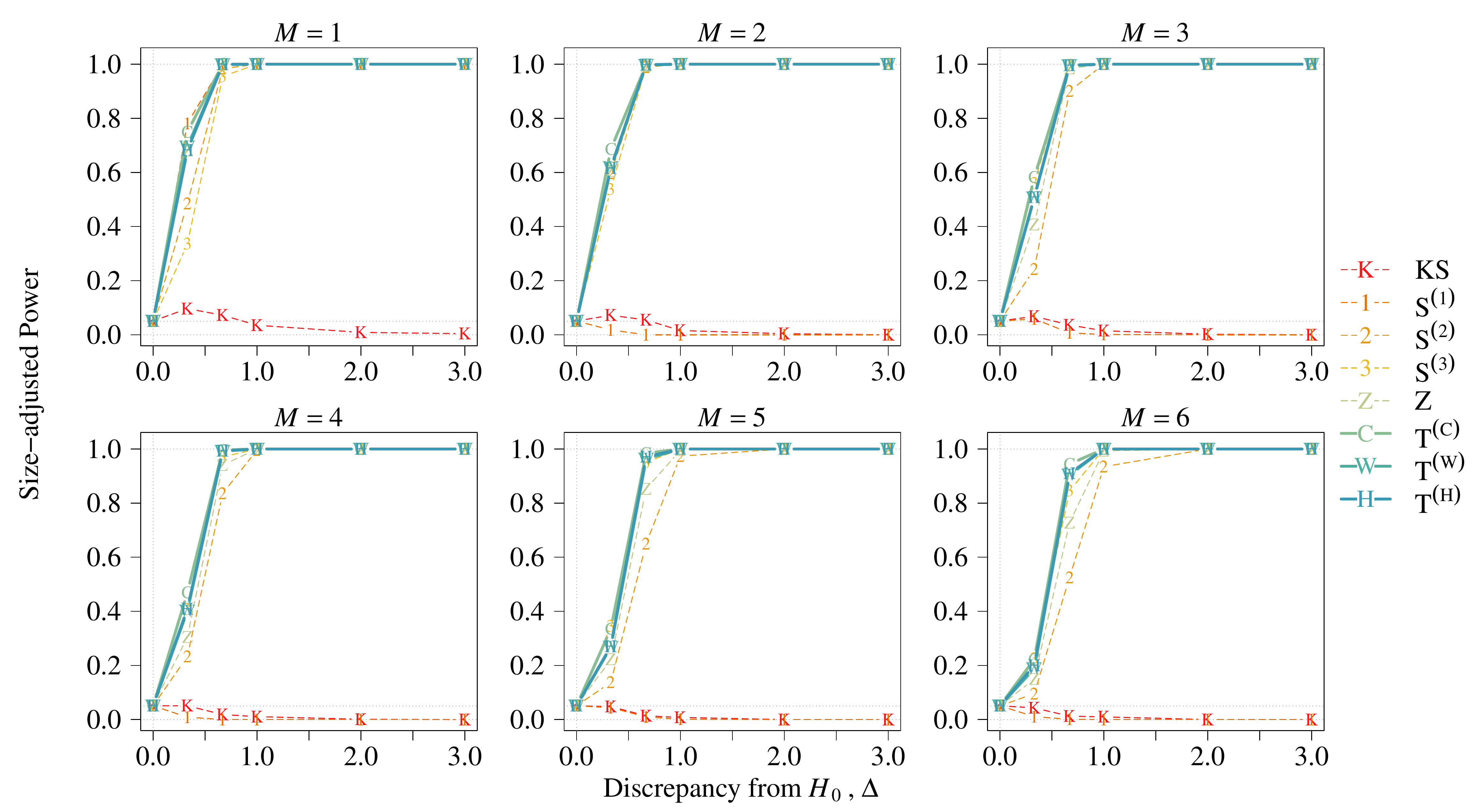}
\vspace{-0.3cm}
\caption{\footnotesize \label{fig: size adjusted power ar-0.5} Size-adjusted power under AR model with $\varpi = -0.5$ and mean function \eqref{EQ mean function normal}.}
\end{figure}

\subsection{ARMA model}
\label{sec: ARMA model}
The ARMA model is defined as $X_i - \mu_i =  \varpi (X_{i-1} - \mu_{i-1}) + \varepsilon_i + \varphi \varepsilon_{i-1}$, where $\varpi, \varphi \in (-1, 1)$. The size accuracy is reported in Table \ref{tab: ARMA size}. The self-normalized approach generally has more accurate size than the non-self-normalized approach. Roughly speaking, $\KS_n$ is seriously over-sized in most cases, while the self-normalized tests are under-sized in some cases. Overall, $\shao_n^{(1)}$ has the most accurate size, but the size distortion of $\shao_n^{(m)}$ is increasing with m. The size performance of self-normalized multiple-CP tests is mixed. Our proposed tests in general are under-sized while other multiple-CP tests, i.e., $\shao_n^{(2)}$, $\shao_n^{(3)}$ and $\zhang_n$, are over-sized when both $\varphi$ and $\varpi$ are positive. Consider when $\varpi = 0.8$. Our proposed tests have significantly smaller size distortion than other multiple-CP tests when $\varphi$ is positive. Except when $\varphi = -0.5$, our proposed tests have larger size distortion, but nevertheless all multiple-CP tests are seriously over-sized. When both $\varpi$ and $\varphi$ are negative, self-normalized tests commit lower type-I error than the nominal size while non-self-normalized $\KS_n$ commits higher type-I error.

Figures \ref{fig: raw power ar0.5ma0.5}--\ref{fig: size adjusted power ar-0.5ma-0.5} present the power. A similar conclusion in Section \ref{sec: AR model} is drawn. Our proposed tests have largest size-adjusted power in the multiple-CP case and lose the least power comparing to AMOC type tests, i.e., $\shao_n^{(1)}$ and $\KS_n$, in the one-CP case.

\begin{table}[t]
\setlength{\tabcolsep}{3pt}
\centering
\footnotesize
\caption{\footnotesize \label{tab: ARMA size} Null rejection rates (\%) at 5\% nominal size under ARMA model and mean function \eqref{EQ mean function normal}.}
\vspace{-0.3cm}
\begin{tabular}{ rr  r rrr r rrr  r rrr r rrr}
\toprule
&  \multicolumn{8}{c}{$n=200$} & \multicolumn{8}{c}{$n=400$}\\
\cmidrule(r){3-10}\cmidrule(r){11-18}
$\varphi$&$\varpi$& $\KS_n$ & $\shao_n^{(1)}$ & $\shao_n^{(2)}$ & $\shao_n^{(3)}$ & $\zhang_n$ & $\T_n^{(\C)}$ & $\T_n^{(\W)}$ & $\T_n^{(\H)}$ & $\KS_n$ &  $\shao_n^{(1)}$ & $\shao_n^{(2)}$ & $\shao_n^{(3)}$ & $\zhang_n$ & $\T_n^{(\C)}$ & $\T_n^{(\W)}$ & $\T_n^{(\H)}$  \\ 
\cmidrule(r){1-18}
0.8 & 0.8 & 54.5 & 9.0 & 26.0 & 44.4 & 27.8 & 2.6 & 7.5 & 7.3 & 53.6 & 6.1 & 13.0 & 23.4 & 14.4 & 1.7 & 3.5 & 3.0\\
 & 0.5 & 13.5 & 6.1 & 10.4 & 13.0 & 11.3 & 0.8 & 1.9 & 1.6 & 9.0 & 4.6 & 5.9 & 8.8 & 5.8 & 1.6 & 2.4 & 2.5\\
 & 0.3 & 11.7 & 5.5 & 7.8 & 8.1 & 8.1 & 0.8 & 1.7 & 1.4 & 7.8 & 4.5 & 5.3 & 6.5 & 5.2 & 2.6 & 2.9 & 2.7\\
 & 0 & 7.4 & 5.2 & 3.5 & 2.5 & 3.1 & 4.1 & 4.7 & 4.6 & 5.0 & 4.2 & 4.5 & 4.6 & 3.8 & 5.2 & 4.7 & 4.2\\
 & $-0.3$ & 8.7 & 5.1 & 5.1 & 3.7 & 5.2 & 2.1 & 2.7 & 2.6 & 5.8 & 4.1 & 4.0 & 4.9 & 3.9 & 4.0 & 4.2 & 4.2\\
 & $-0.5$ & 7.8 & 4.6 & 5.1 & 3.8 & 5.2 & 3.0 & 2.9 & 3.2 & 5.4 & 4.0 & 4.2 & 4.7 & 3.6 & 4.5 & 4.3 & 4.2\\

\cmidrule(r){1-18}

0.5 & 0.8 & 54.6 & 9.0 & 26.1 & 44.3 & 28.2 & 4.1 & 10.1 & 9.7 & 53.7 & 6.1 & 13.1 & 23.5 & 14.5 & 2.6 & 4.5 & 4.5\\
 & 0.5 & 13.3 & 6.0 & 10.4 & 12.5 & 11.0 & 1.5 & 2.3 & 2.1 & 9.0 & 4.5 & 5.7 & 8.7 & 6.0 & 2.1 & 3.0 & 3.1\\
 & 0.3 & 11.7 & 5.4 & 7.7 & 7.8 & 7.8 & 1.4 & 2.2 & 2.4 & 7.4 & 4.5 & 5.1 & 6.4 & 5.2 & 3.2 & 3.7 & 3.1\\
 & 0 & 7.4 & 5.2 & 3.5 & 2.5 & 3.1 & 4.1 & 4.7 & 4.6 & 5.0 & 4.2 & 4.5 & 4.6 & 3.8 & 5.2 & 4.7 & 4.2\\
 & $-0.3$ & 8.1 & 4.9 & 4.5 & 3.2 & 5.1 & 3.3 & 3.7 & 3.6 & 5.1 & 3.9 & 3.8 & 4.7 & 3.9 & 4.4 & 4.5 & 4.7\\

 & $-0.8$ & 10.2 & 4.4 & 3.3 & 1.5 & 2.5 & 8.4 & 9.0 & 8.7 & 6.1 & 4.0 & 3.4 & 3.1 & 2.6 & 6.7 & 6.6 & 6.4\\
\cmidrule(r){1-18}

$-0.5$ & 0.8 & 43.7 & 8.6 & 20.3 & 29.9 & 22.5 & 32.4 & 37.8 & 36.6 & 43.5 & 6.0 & 12.3 & 19.7 & 14.0 & 17.5 & 20.3 & 20.7\\

 & 0.3 & 11.8 & 3.6 & 1.4 & 0.4 & 1.4 & 1.5 & 2.1 & 2.2 & 6.3 & 2.9 & 2.0 & 1.6 & 2.1 & 2.7 & 2.7 & 2.9\\
 & 0 & 7.4 & 5.2 & 3.5 & 2.5 & 3.1 & 4.1 & 4.7 & 4.6 & 5.0 & 4.2 & 4.5 & 4.6 & 3.8 & 5.2 & 4.7 & 4.2\\
 & $-0.3$ & 43.9 & 0.9 & 0.0 & 0.0 & 0.0 & 0.2 & 0.2 & 0.1 & 40.0 & 1.8 & 0.1 & 0.0 & 0.3 & 0.2 & 0.9 & 0.9\\
 & $-0.5$ & 53.3 & 1.0 & 0.0 & 0.0 & 0.0 & 0.1 & 0.4 & 0.2 & 50.7 & 1.0 & 0.1 & 0.0 & 0.0 & 0.0 & 0.4 & 0.2\\
 & $-0.8$ & 51.6 & 0.0 & 0.0 & 0.0 & 0.0 & 0.0 & 0.2 & 0.0 & 53.2 & 0.2 & 0.0 & 0.0 & 0.0 & 0.0 & 0.1 & 0.1\\
\cmidrule(r){1-18}

$-0.8$ & 0.5 & 24.2 & 1.2 & 0.1 & 0.0 & 0.0 & 0.0 & 0.1 & 0.0 & 20.3 & 1.3 & 0.0 & 0.0 & 0.0 & 0.0 & 0.1 & 0.1\\
 & 0.3 & 38.7 & 0.4 & 0.0 & 0.0 & 0.0 & 0.0 & 0.0 & 0.0 & 35.3 & 0.7 & 0.0 & 0.0 & 0.0 & 0.0 & 0.0 & 0.0\\
 & 0 & 7.4 & 5.2 & 3.5 & 2.5 & 3.1 & 4.1 & 4.7 & 4.6 & 5.0 & 4.2 & 4.5 & 4.6 & 3.8 & 5.2 & 4.7 & 4.2\\
 & $-0.3$ & 55.8 & 0.0 & 0.0 & 0.0 & 0.0 & 0.0 & 0.0 & 0.0 & 52.8 & 0.0 & 0.0 & 0.0 & 0.0 & 0.0 & 0.0 & 0.0\\
 & $-0.5$ & 57.9 & 0.0 & 0.0 & 0.0 & 0.0 & 0.0 & 0.0 & 0.0 & 60.0 & 0.0 & 0.0 & 0.0 & 0.0 & 0.0 & 0.0 & 0.0\\
 & $-0.8$ & 52.1 & 0.0 & 0.0 & 0.0 & 0.0 & 0.0 & 0.0 & 0.0 & 54.6 & 0.0 & 0.0 & 0.0 & 0.0 & 0.0 & 0.0 & 0.0\\

\bottomrule
\end{tabular} 
\end{table}

\begin{figure}[H]
\centering
\includegraphics[width=0.9\linewidth]{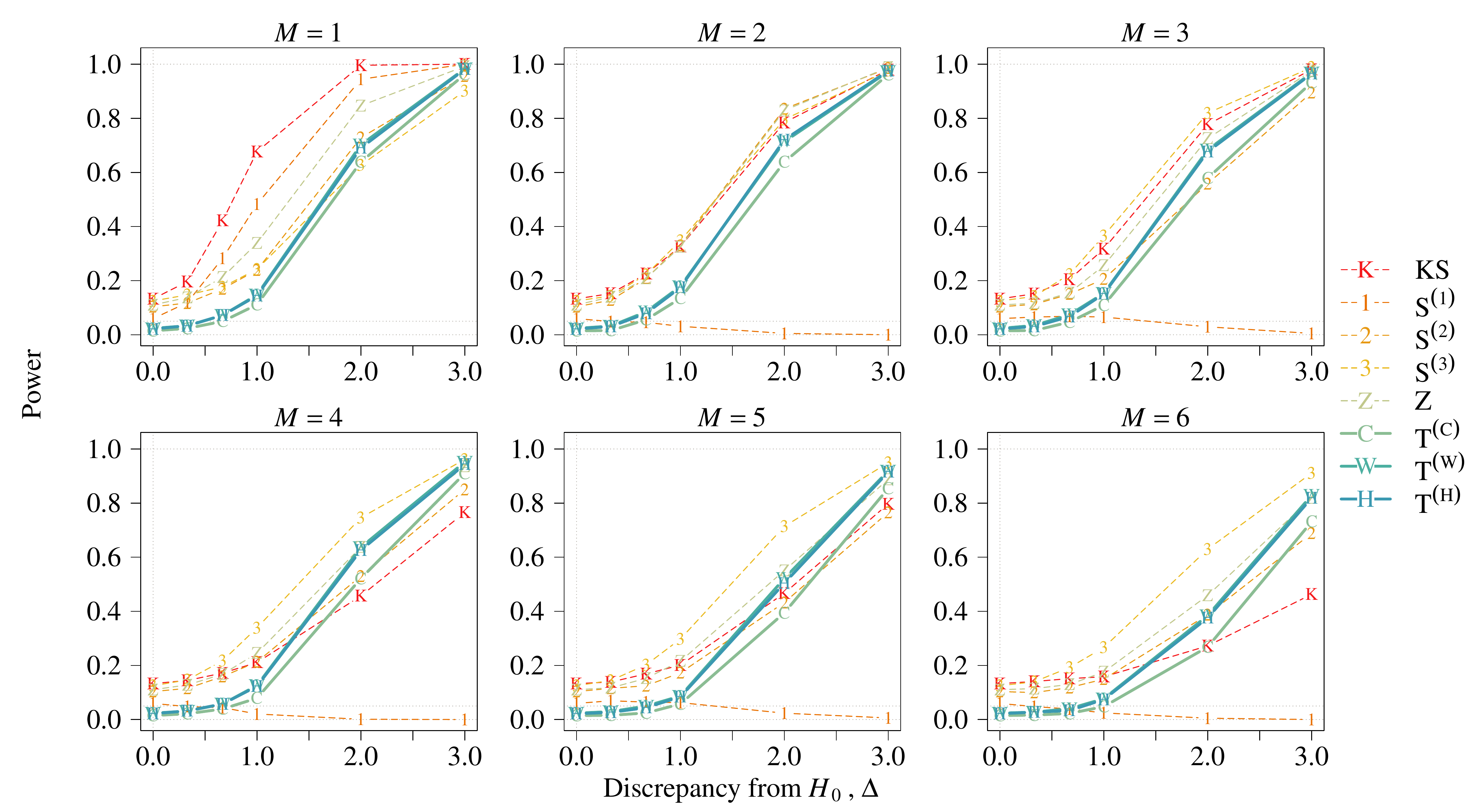}
\vspace{-0.3cm}
\caption{\footnotesize \label{fig: raw power ar0.5ma0.5} Power under ARMA model with $\varpi = \varphi = 0.5$ and mean function \eqref{EQ mean function normal}.}
\end{figure}

\begin{figure}[H]
\centering
\includegraphics[width=0.9\linewidth]{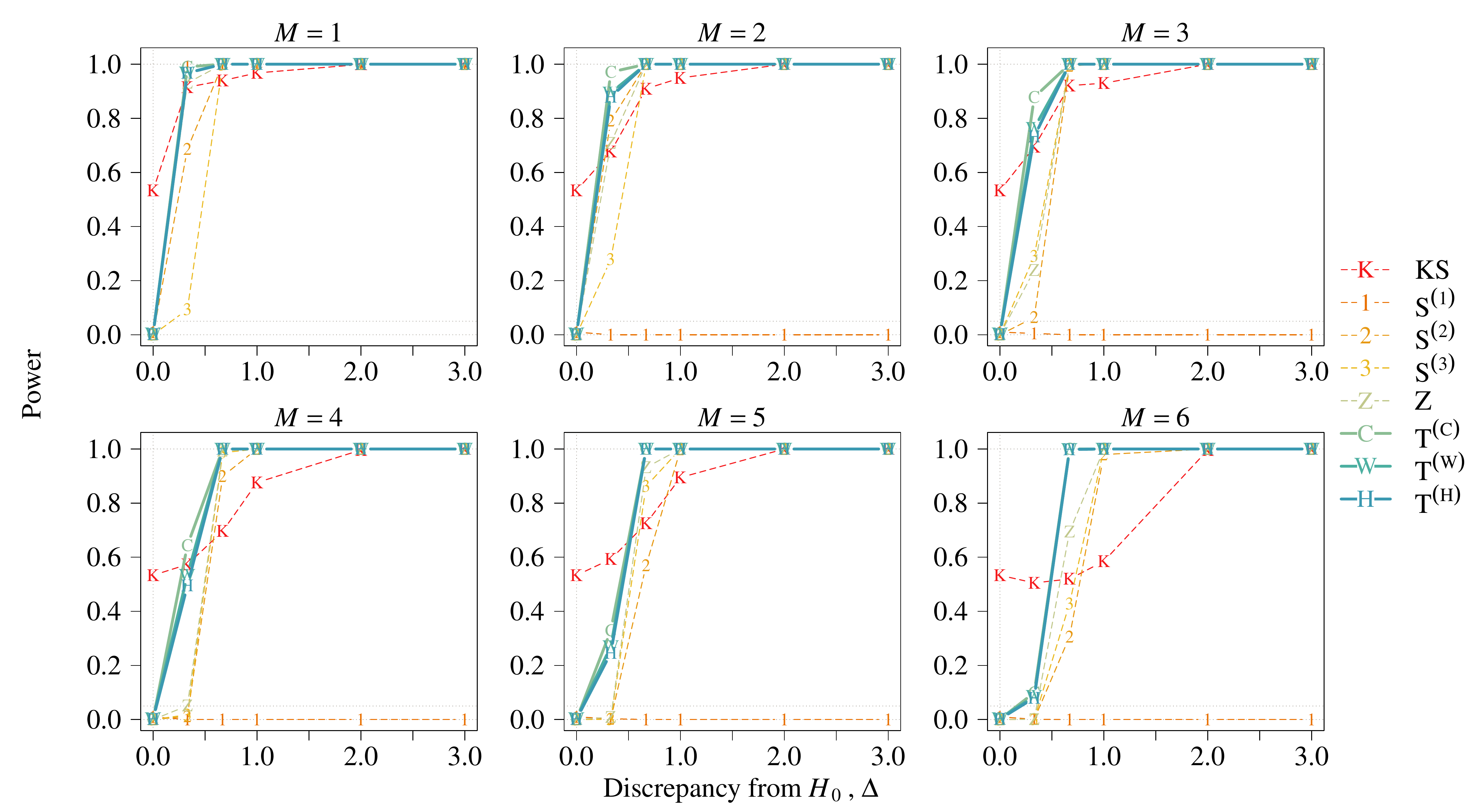}
\vspace{-0.3cm}
\caption{\footnotesize \label{fig: raw power ar-0.5ma-0.5} Power under ARMA model with $\varpi = \varphi = -0.5$ and mean function \eqref{EQ mean function normal}.}
\end{figure}

\begin{figure}[H]
\centering
\includegraphics[width=\linewidth]{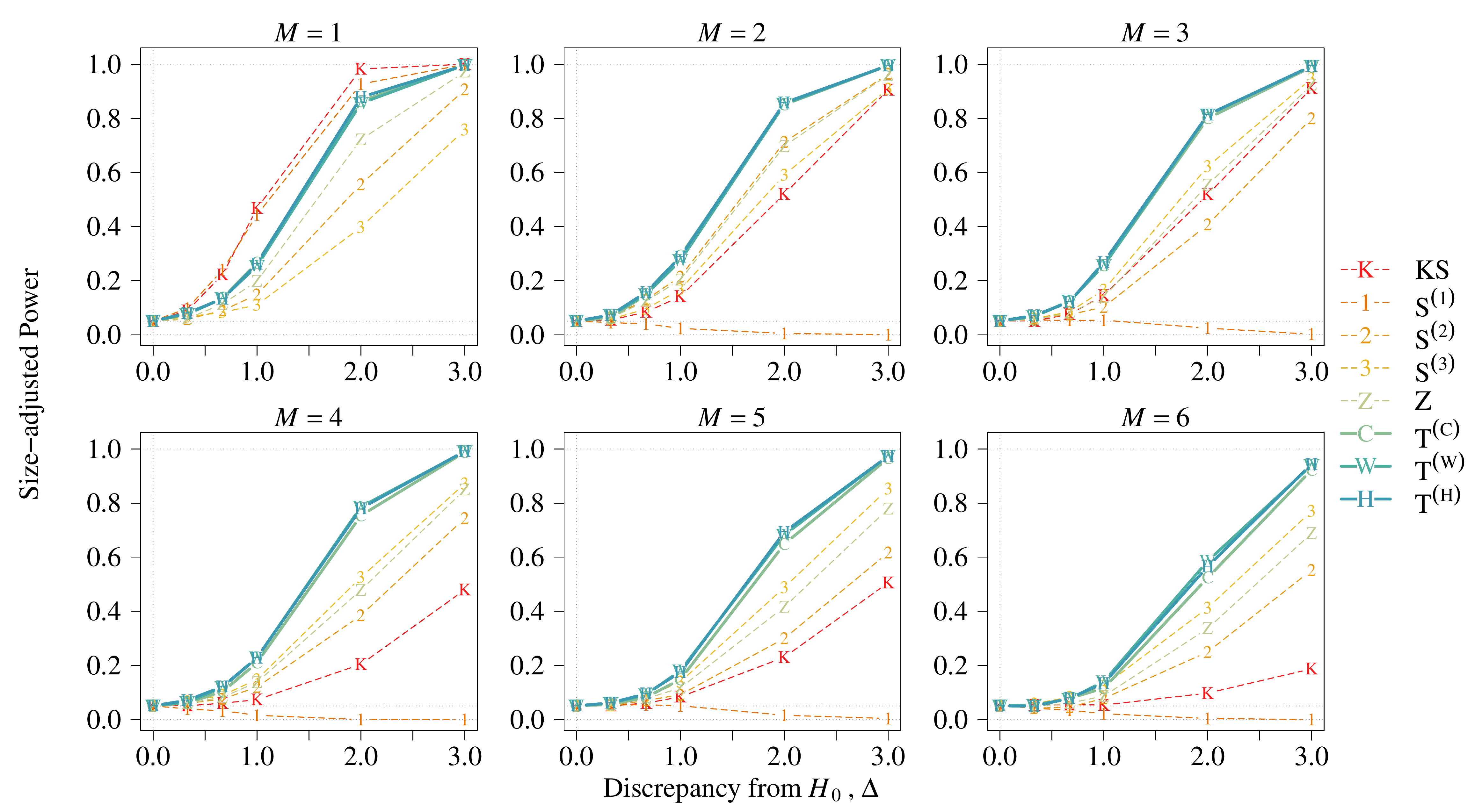}
\vspace{-0.3cm}
\caption{\footnotesize \label{fig: size adjusted power ar0.5ma0.5} Size-adjusted power under ARMA model with $\varpi = \varphi = 0.5$ and mean function \eqref{EQ mean function normal}.}
\end{figure}

\begin{figure}[H]
\centering
\includegraphics[width=\linewidth]{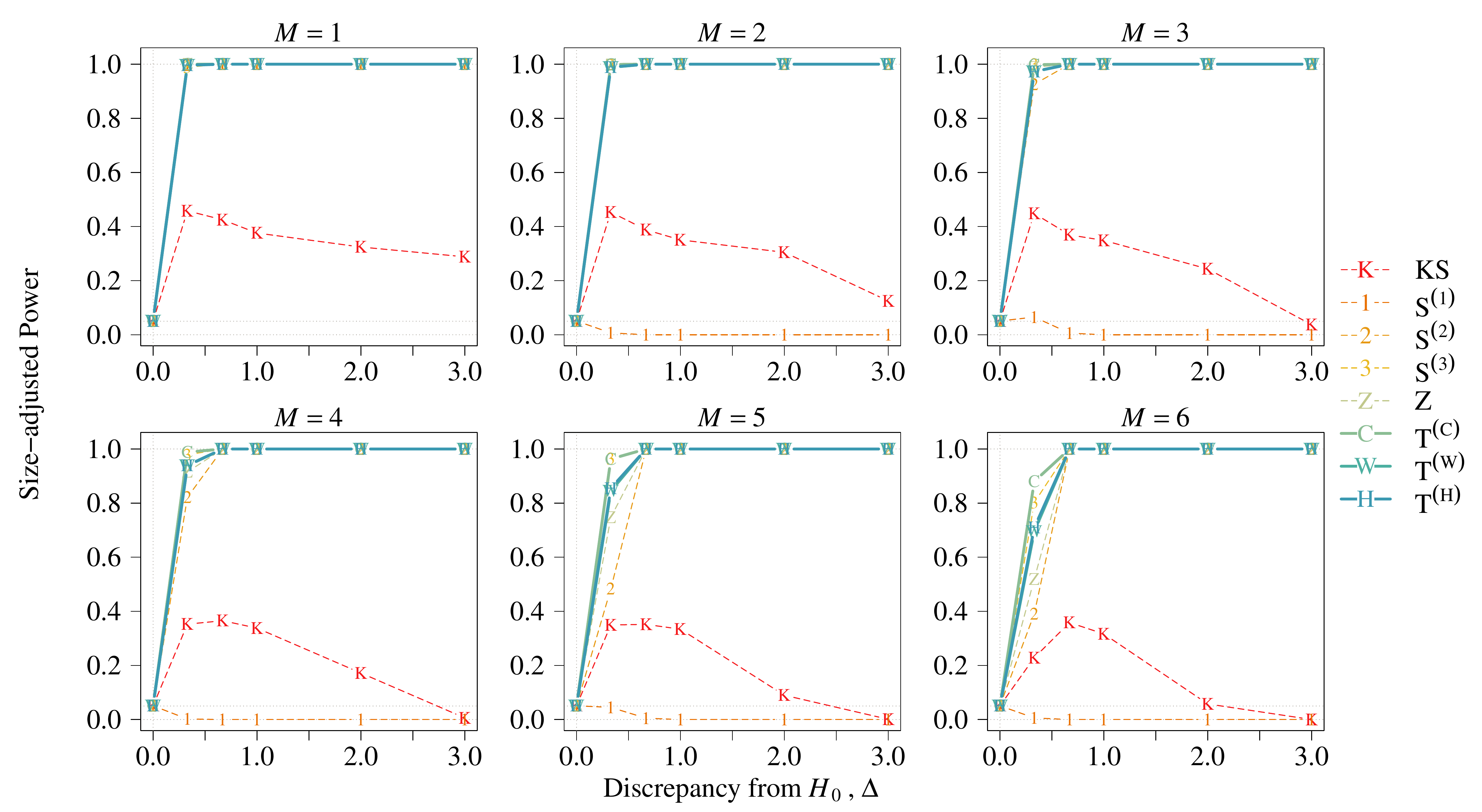}
\vspace{-0.3cm}
\caption{\footnotesize \label{fig: size adjusted power ar-0.5ma-0.5} Size-adjusted power under ARMA model with $\varpi = \varphi = -0.5$ and mean function \eqref{EQ mean function normal}.}
\end{figure}

\subsection{TAR Model}
\label{sec: TAR model}
The threshold AR model is defined to be
\begin{equation*}
X_i - \mu_i=  \{\varpi_1 I_{i-1} + \varpi_2 (1-I_{i-1})\}(X_{i-1} - \mu_{i-1}) + \varepsilon_i,
\end{equation*}
where $I_i = \mathbb{1}_{\{X_{i} - \mu_i \geq 0 \}}$ and $\varpi_1, \varpi_2  \in (-1,1)$. In this subsection, we report the performance under different choices of $\varpi_1$ and $\varpi_2$. Table \ref{tab: TAR size} tabulates the null rejection rates. Similar to Section \ref{sec: AR model}, $\KS_n$ has larger size distortion than the self-normalized tests. Our proposed tests have more accurate size in general. Although our tests have larger size distortion than $\shao_n^{(1)}$ when $\varpi_1 = 0.8$, our proposed tests have accurate size while other self-normalized tests are seriously under-sized when both $\varpi_1$ and $\varpi_2$ are negative. In contrast to the result in BAR model, $\T_n^{(\C)}$ has more accurate size than $\T_n^{(\W)}$ and $\T_n^{(\H)}$ in general.  A similar result in power is also observed in Figures \ref{fig: power TAR ar0.5 TAR0.5}--\ref{fig: size adjusted power TAR ar-0.5 tar-0.5}. Our proposed tests perform the best under multiple-CP case but they slightly underperform $\shao_n^{(1)}$ and $\KS_n$ under the single CP case.

\begin{figure}[H]
\centering
\includegraphics[width=\linewidth]{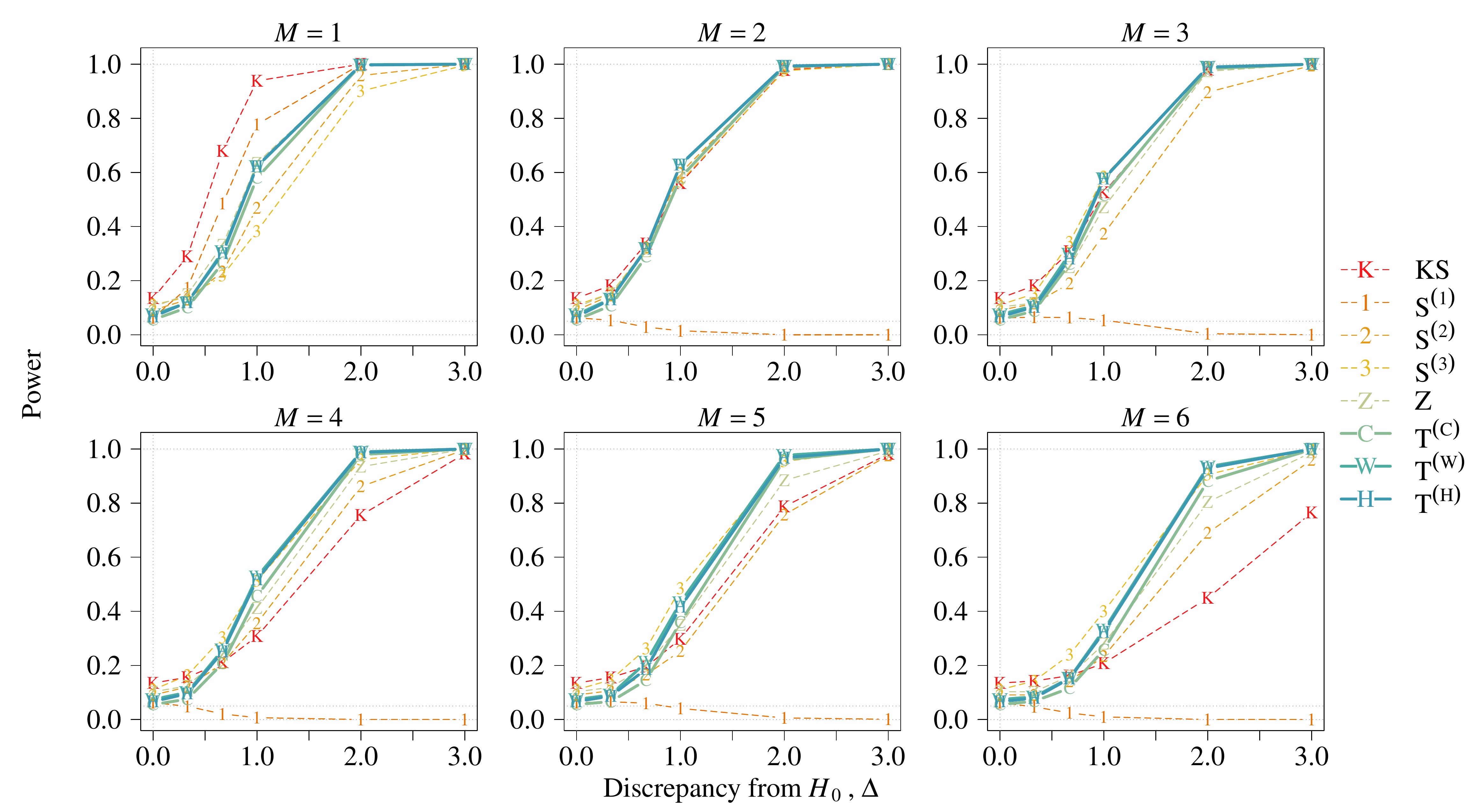}
\vspace{-0.3cm}
\caption{\footnotesize \label{fig: power TAR ar0.5 TAR0.5} Power under TAR model with $\varpi_1 = \varpi_2 = 0.5$ and mean function \eqref{EQ mean function normal}.}
\end{figure}

\begin{table}[t]
\setlength{\tabcolsep}{3pt}
\centering
\footnotesize
\caption{\footnotesize \label{tab: TAR size} Null rejection rates (\%) at 5\% nominal size under TAR model and mean function \eqref{EQ mean function normal}.}
\vspace{-0.3cm}
\begin{tabular}{ rr  r rrr r rrr  r rrr r rrr}
\toprule
&  \multicolumn{8}{c}{$n=200$} & \multicolumn{8}{c}{$n=400$}\\
\cmidrule(r){3-10}\cmidrule(r){11-18}
$\varpi_2$ & $\varpi_1$& $\KS_n$ & $\shao_n^{(1)}$ & $\shao_n^{(2)}$ & $\shao_n^{(3)}$ & $\zhang_n$ & $\T_n^{(\C)}$ & $\T_n^{(\W)}$ & $\T_n^{(\H)}$ & $\KS_n$ &  $\shao_n^{(1)}$ & $\shao_n^{(2)}$ & $\shao_n^{(3)}$ & $\zhang_n$ & $\T_n^{(\C)}$ & $\T_n^{(\W)}$ & $\T_n^{(\H)}$  \\ 
\cmidrule(r){1-18}

0.8 & 0.8 & 52.0 & 8.7 & 27.1 & 43.0 & 30.9 & 15.0 & 23.6 & 21.6 & 52.8 & 6.2 & 14.1 & 23.3 & 14.5 & 9.6 & 13.1 & 12.2\\
 & 0.5 & 30.5 & 7.2 & 18.8 & 29.7 & 19.6 & 10.0 & 15.7 & 14.6 & 26.7 & 5.0 & 11.7 & 14.8 & 11.9 & 7.0 & 9.1 & 8.8\\
 & 0.3 & 26.4 & 6.9 & 17.2 & 27.1 & 18.1 & 8.9 & 14.1 & 12.6 & 23.7 & 5.2 & 11.6 & 14.3 & 11.6 & 6.3 & 8.3 & 7.9\\
 & 0 & 23.6 & 6.8 & 16.1 & 24.5 & 17.3 & 8.3 & 12.4 & 11.8 & 20.2 & 5.4 & 11.5 & 12.8 & 11.4 & 5.8 & 8.4 & 7.9\\
 & $-0.3$ & 21.1 & 6.7 & 14.5 & 22.6 & 15.7 & 7.7 & 11.6 & 11.4 & 18.5 & 5.3 & 10.8 & 12.4 & 10.4 & 6.0 & 8.1 & 8.2\\
 & $-0.5$ & 20.0 & 6.7 & 14.3 & 21.1 & 14.5 & 7.3 & 10.5 & 10.6 & 17.2 & 5.7 & 10.0 & 12.0 & 10.2 & 6.0 & 8.0 & 7.7\\
 & $-0.8$ & 18.2 & 6.4 & 13.4 & 18.7 & 14.4 & 6.4 & 10.3 & 9.9 & 16.3 & 5.4 & 9.3 & 10.6 & 9.5 & 6.0 & 7.1 & 6.9\\
\cmidrule(r){1-18}

0.5 & 0.8 & 30.2 & 7.1 & 18.8 & 28.9 & 22.0 & 10.9 & 14.6 & 14.9 & 27.8 & 5.6 & 10.1 & 16.7 & 9.8 & 7.2 & 8.3 & 9.0\\
 & 0.5 & 13.5 & 6.3 & 9.1 & 11.0 & 10.3 & 5.6 & 7.4 & 6.6 & 9.5 & 4.4 & 5.5 & 8.1 & 5.7 & 5.1 & 5.4 & 5.2\\
 & 0.3 & 11.7 & 5.5 & 7.9 & 9.0 & 8.5 & 4.7 & 5.8 & 5.9 & 9.0 & 4.6 & 4.8 & 6.9 & 5.5 & 5.4 & 5.0 & 5.4\\
 & 0 & 10.6 & 5.4 & 6.8 & 7.0 & 7.6 & 4.3 & 5.4 & 5.3 & 8.2 & 4.3 & 5.3 & 6.9 & 5.1 & 5.1 & 4.8 & 5.1\\
 & $-0.3$ & 8.5 & 5.0 & 5.9 & 5.7 & 5.9 & 4.1 & 4.9 & 5.2 & 7.5 & 4.9 & 5.2 & 5.9 & 4.4 & 5.1 & 4.6 & 5.2\\
 & $-0.5$ & 8.3 & 4.9 & 5.4 & 5.3 & 5.1 & 4.2 & 4.8 & 5.1 & 6.5 & 5.2 & 5.2 & 5.8 & 4.2 & 5.5 & 4.9 & 5.1\\
 & $-0.8$ & 7.3 & 4.3 & 4.8 & 4.1 & 4.2 & 4.1 & 4.9 & 4.8 & 5.1 & 5.4 & 4.8 & 4.8 & 3.8 & 5.0 & 4.8 & 4.8\\
\cmidrule(r){1-18}

$-0.5$ & 0.8 & 19.8 & 6.9 & 14.7 & 21.0 & 14.8 & 8.3 & 10.7 & 10.4 & 16.4 & 5.3 & 7.7 & 10.4 & 9.0 & 5.1 & 7.1 & 6.7\\
 & 0.5 & 8.2 & 4.3 & 5.8 & 4.4 & 5.3 & 4.1 & 5.2 & 4.5 & 6.3 & 3.9 & 4.1 & 4.6 & 3.8 & 4.4 & 5.0 & 4.9\\
 & 0.3 & 7.6 & 4.2 & 3.5 & 2.4 & 3.7 & 4.0 & 4.5 & 4.6 & 5.4 & 4.0 & 3.6 & 3.5 & 3.2 & 4.5 & 4.8 & 4.6\\
 & 0 & 10.3 & 4.4 & 2.4 & 1.4 & 2.0 & 4.1 & 4.7 & 4.3 & 5.9 & 3.8 & 2.7 & 2.8 & 2.7 & 5.3 & 4.8 & 5.0\\
 & $-0.3$ & 16.6 & 4.2 & 1.9 & 0.6 & 1.4 & 4.1 & 5.1 & 5.1 & 9.5 & 3.3 & 2.6 & 2.7 & 2.2 & 5.6 & 5.0 & 5.0\\
 & $-0.5$ & 24.4 & 3.5 & 1.4 & 0.5 & 1.3 & 4.8 & 5.6 & 5.6 & 14.4 & 3.1 & 2.3 & 2.0 & 2.1 & 5.1 & 4.6 & 4.9\\
 & $-0.8$ & 36.6 & 3.4 & 0.9 & 0.2 & 0.5 & 6.0 & 6.4 & 6.3 & 28.4 & 4.1 & 1.6 & 0.9 & 1.4 & 5.2 & 4.3 & 4.8\\
\cmidrule(r){1-18}

$-0.8$ & 0.8 & 17.9 & 6.5 & 13.1 & 19.4 & 14.1 & 7.6 & 10.2 & 9.6 & 16.1 & 4.9 & 7.1 & 9.6 & 7.5 & 4.8 & 7.1 & 6.4\\
 & 0.5 & 8.5 & 4.8 & 4.7 & 4.3 & 4.8 & 3.9 & 4.5 & 4.3 & 5.9 & 4.2 & 3.4 & 4.3 & 3.3 & 4.2 & 4.7 & 4.9\\
 & 0.3 & 8.4 & 4.3 & 3.5 & 2.2 & 3.3 & 3.6 & 4.0 & 4.8 & 5.9 & 3.9 & 3.5 & 2.8 & 2.8 & 4.3 & 4.3 & 4.4\\
 & 0 & 14.6 & 4.3 & 1.8 & 0.9 & 1.8 & 3.9 & 4.3 & 4.2 & 7.3 & 4.1 & 2.9 & 2.1 & 2.5 & 4.8 & 4.9 & 5.2\\
 & $-0.3$ & 26.6 & 3.4 & 1.4 & 0.4 & 1.1 & 4.3 & 5.4 & 5.1 & 17.0 & 3.1 & 2.1 & 2.0 & 2.1 & 5.0 & 4.6 & 4.7\\
 & $-0.5$ & 36.7 & 3.1 & 1.2 & 0.3 & 0.9 & 5.6 & 6.3 & 5.8 & 27.1 & 3.2 & 2.1 & 1.1 & 1.4 & 5.4 & 5.0 & 5.2\\
 & $-0.8$ & 43.7 & 2.1 & 0.4 & 0.1 & 0.2 & 8.3 & 9.5 & 8.8 & 39.0 & 2.4 & 0.9 & 0.2 & 0.6 & 7.2 & 7.7 & 6.9\\

\bottomrule
\end{tabular} 
\end{table}

\clearpage

\begin{figure}[H]
\centering
\includegraphics[width=\linewidth]{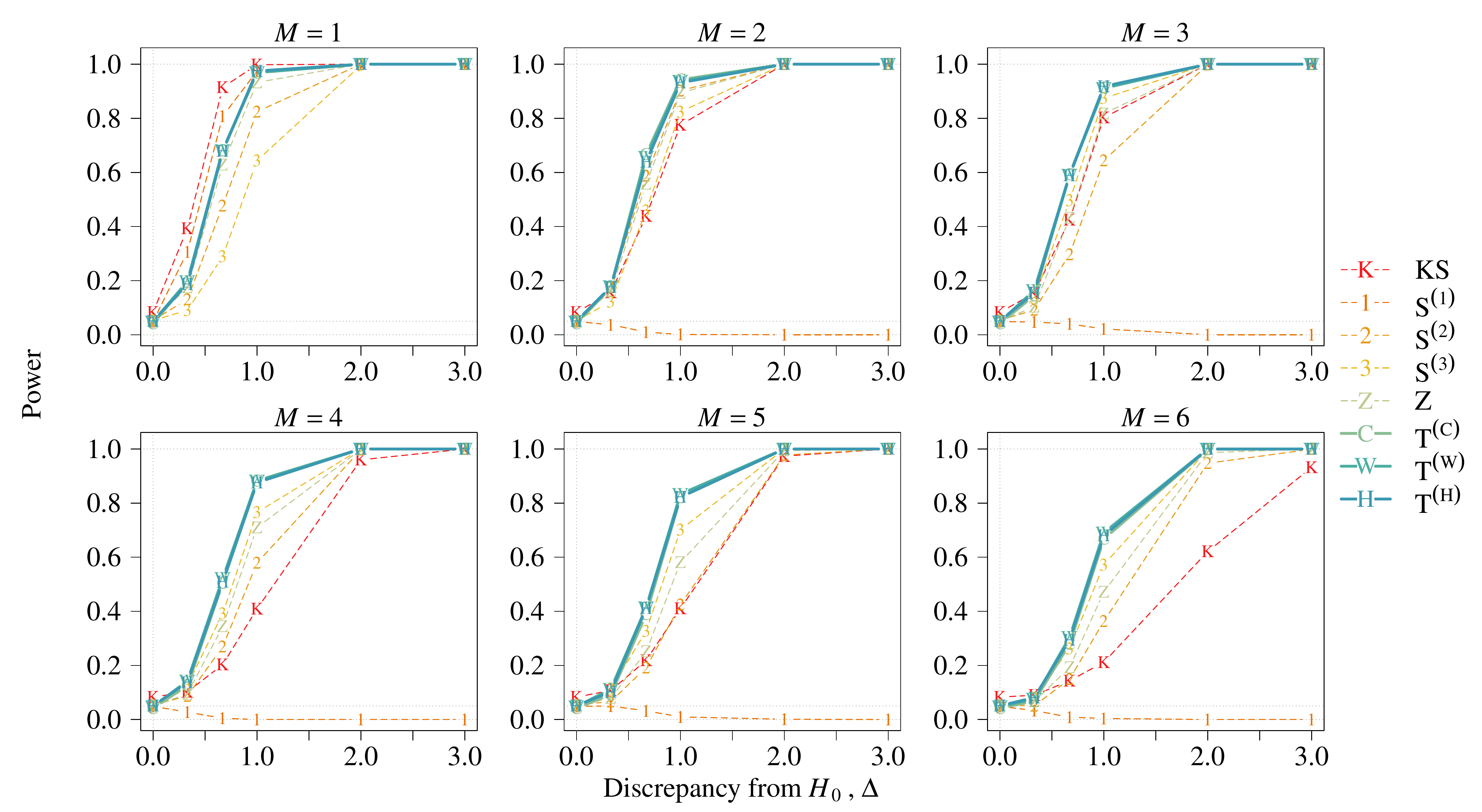}
\vspace{-0.3cm}
\caption{\footnotesize \label{fig: power TAR ar-0.5 TAR0.5} Power under TAR model with $-\varpi_1 = \varpi_2 = 0.5$ and mean function \eqref{EQ mean function normal}.}
\end{figure}

\begin{figure}[H]
\centering
\includegraphics[width=\linewidth]{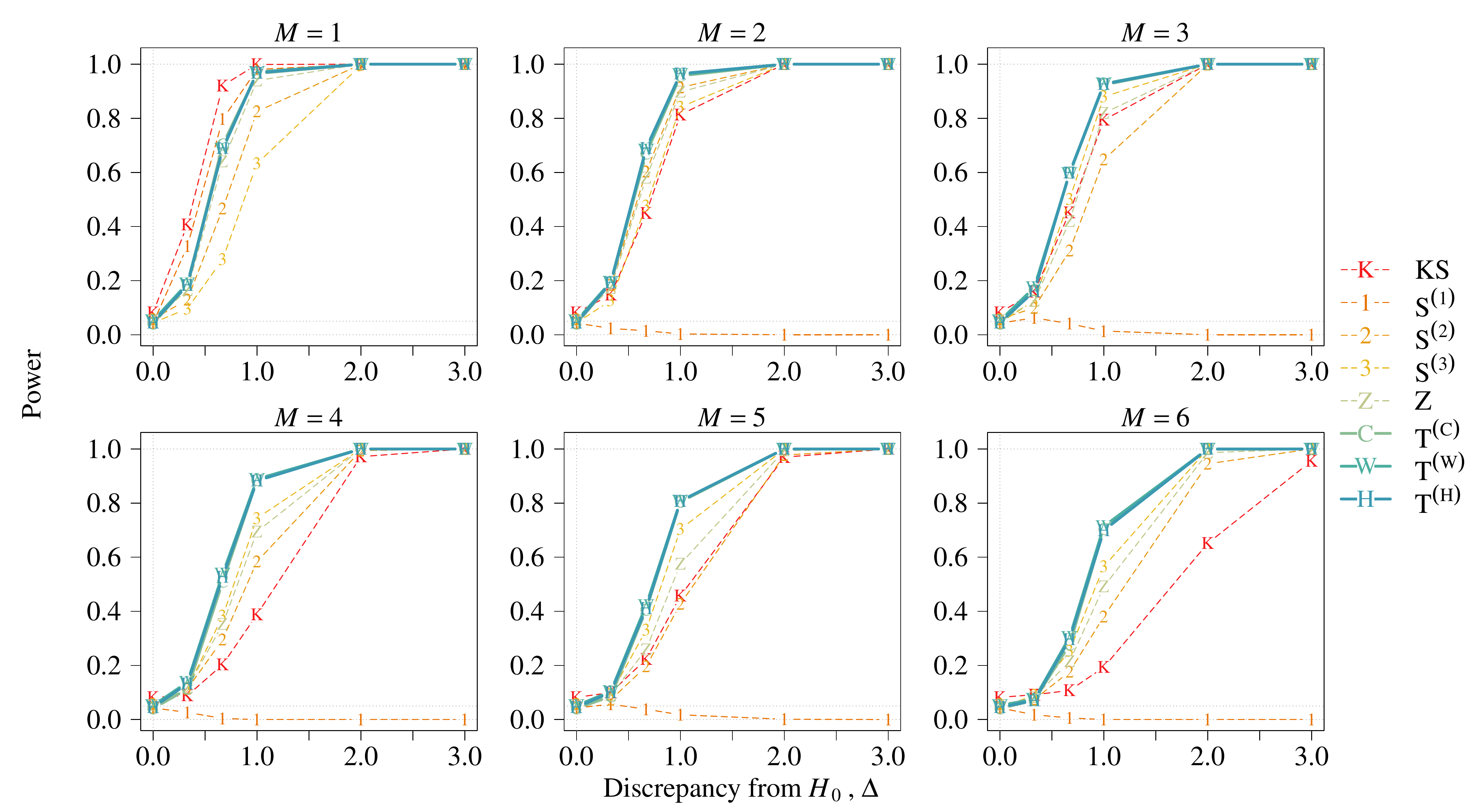}
\vspace{-0.3cm}
\caption{\footnotesize \label{fig: power TAR ar0.5 TAR-0.5} Power under TAR model with $\varpi_1 = -\varpi_2 = 0.5$ and mean function \eqref{EQ mean function normal}.}
\end{figure}

\begin{figure}[H]
\centering
\includegraphics[width=\linewidth]{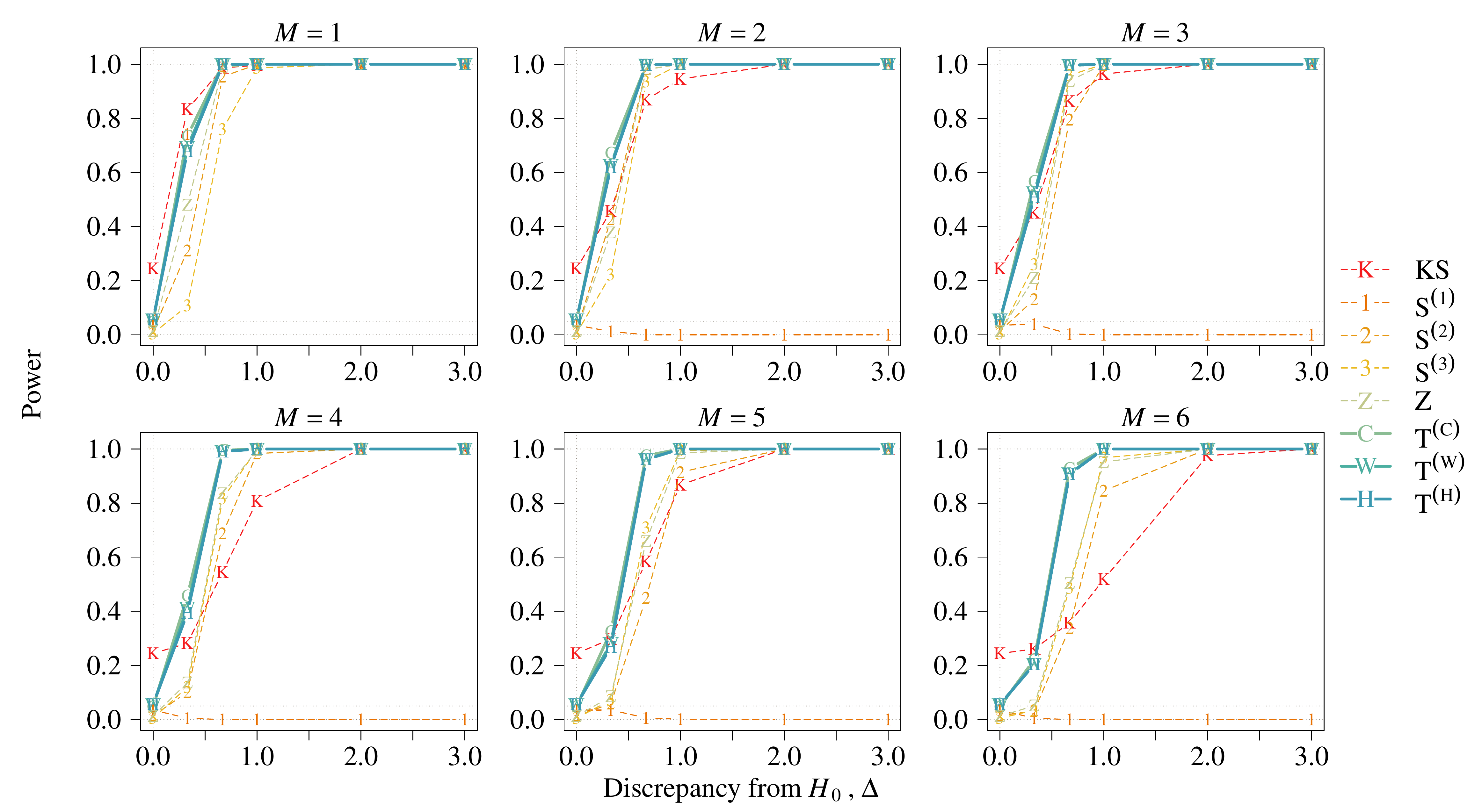}
\vspace{-0.3cm}
\caption{\footnotesize \label{fig: power TAR ar-0.5 tar-0.5} Power under TAR model with $\varpi_1 = \varpi_2 = -0.5$ and mean function \eqref{EQ mean function normal}.}
\end{figure}

\begin{figure}[H]
\centering
\includegraphics[width=\linewidth]{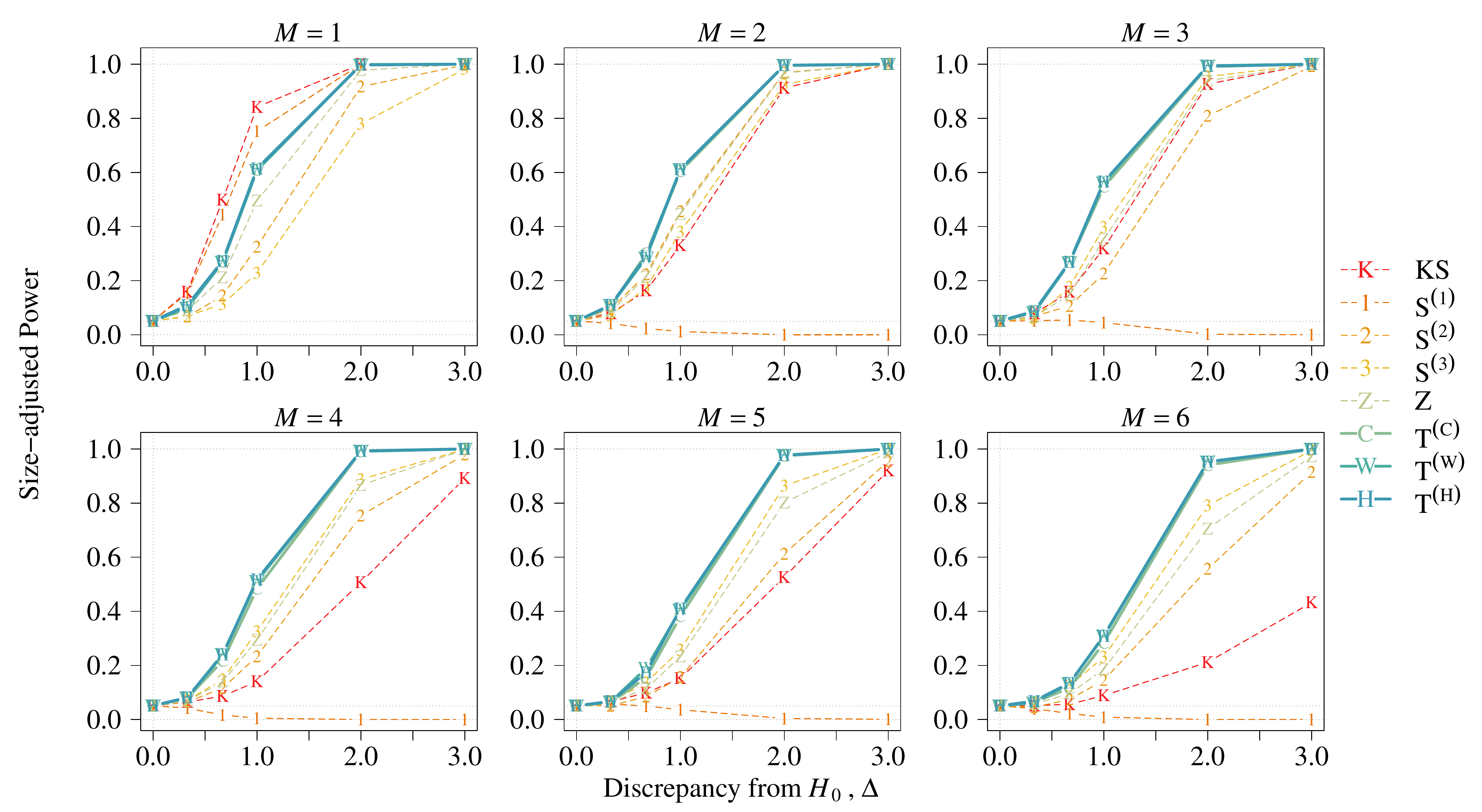}
\vspace{-0.3cm}
\caption{\footnotesize \label{fig: size adjusted power TAR ar0.5 TAR0.5} Size-adjusted power under TAR model with $\varpi_1 = \varpi_2 = 0.5$ and mean function \eqref{EQ mean function normal}.}
\end{figure}

\begin{figure}[H]
\centering
\includegraphics[width=\linewidth]{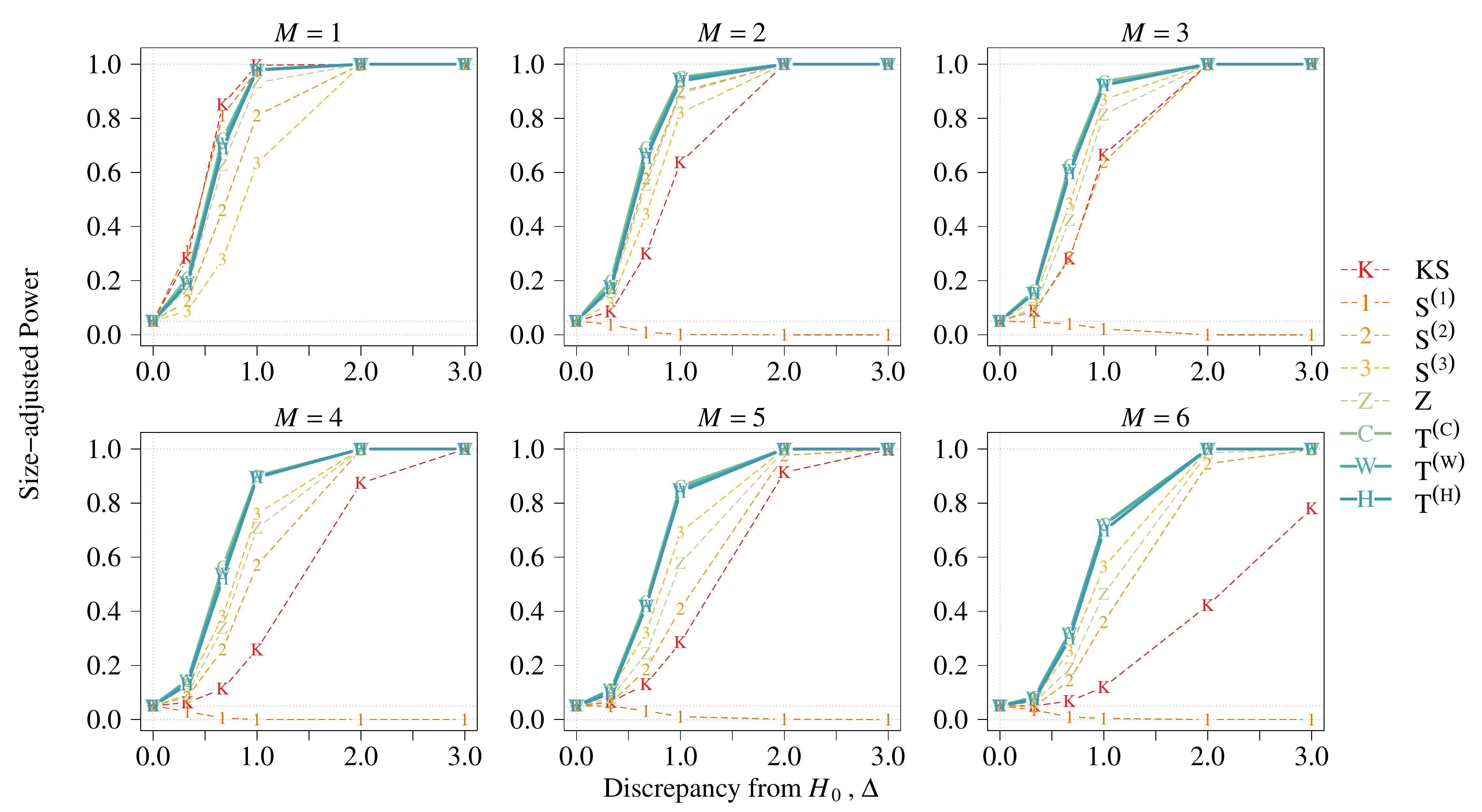}
\vspace{-0.3cm}
\caption{\footnotesize \label{fig: size adjusted power TAR ar-0.5 TAR0.5} Size-adjusted power under TAR model with $-\varpi_1 = \varpi_2 = 0.5$ and mean function \eqref{EQ mean function normal}.}
\end{figure}

\begin{figure}[H]
\centering
\includegraphics[width=\linewidth]{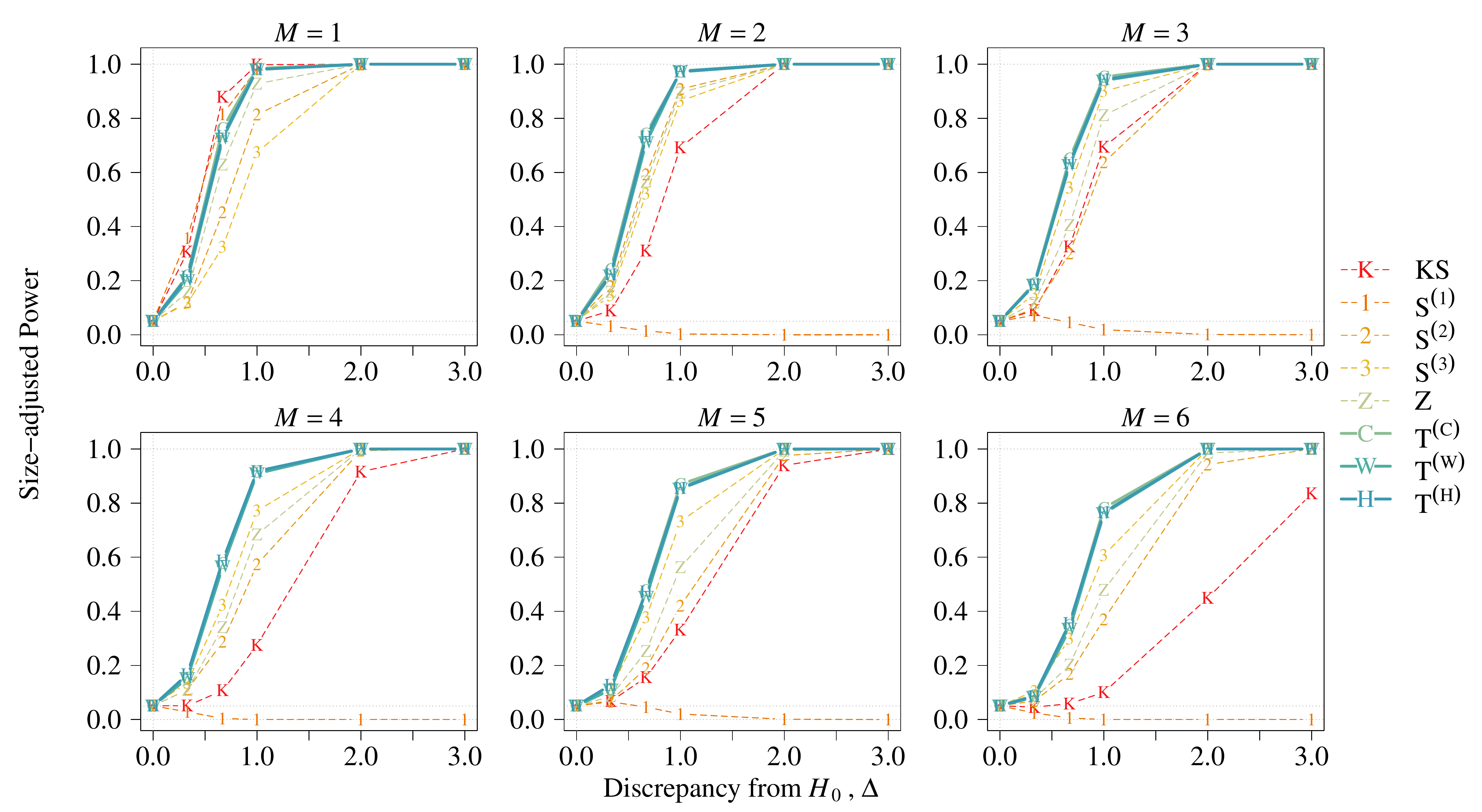}
\vspace{-0.3cm}
\caption{\footnotesize \label{fig: size adjusted power TAR ar0.5 TAR-0.5} Size-adjusted power under TAR model with $\varpi_1 = -\varpi_2 = 0.5$ and mean function \eqref{EQ mean function normal}.}
\end{figure}

\begin{figure}[H]
\centering
\includegraphics[width=\linewidth]{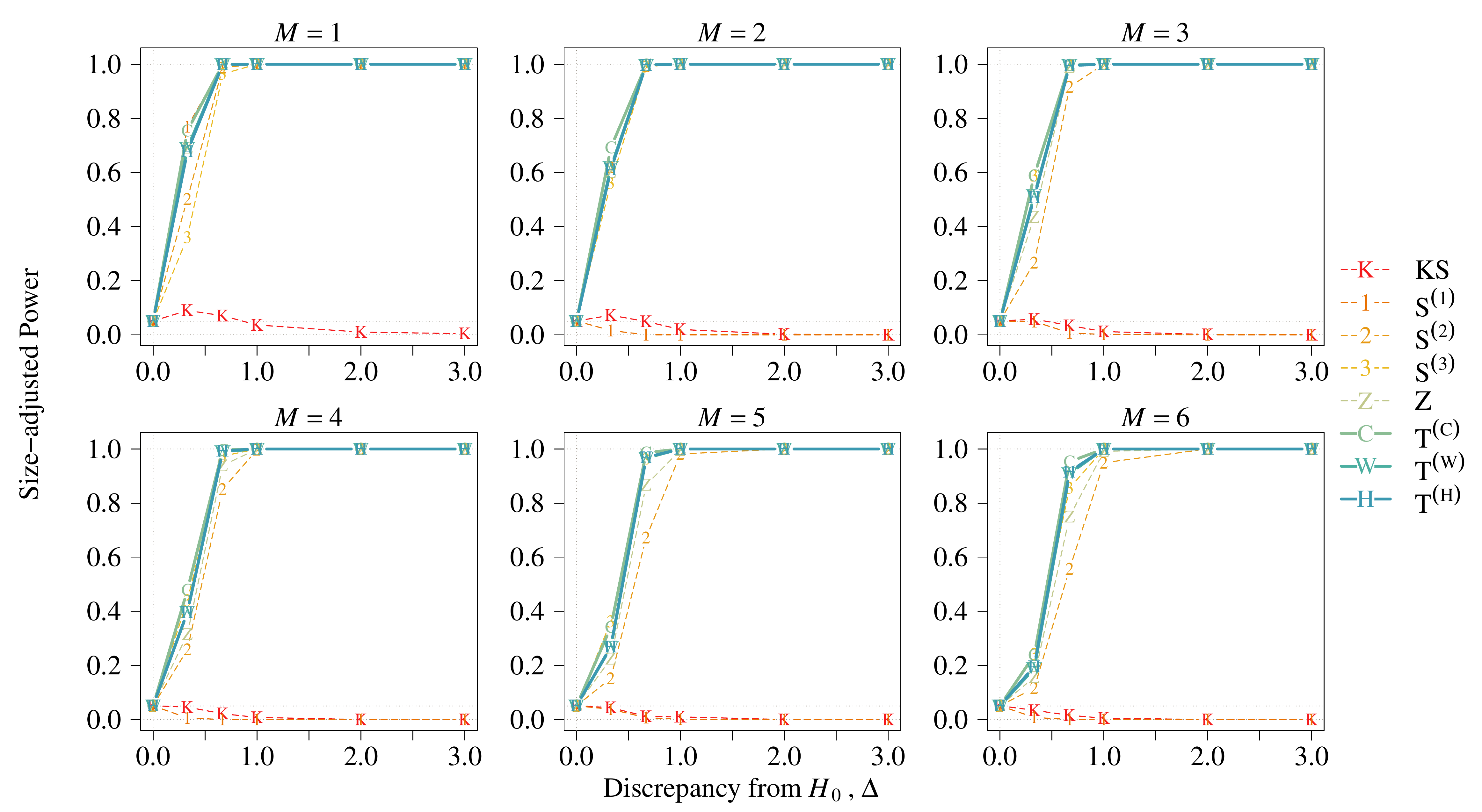}
\vspace{-0.3cm}
\caption{\footnotesize \label{fig: size adjusted power TAR ar-0.5 tar-0.5} Size-adjusted power under TAR model with $\varpi_1 = \varpi_2 = -0.5$ and mean function \eqref{EQ mean function normal}.}
\end{figure}

\subsection{NAR Model}
\label{sec: NAR model}
The data under the NAR model has form $X_i - \mu_i = \varpi\left| X_{i-1} - \mu_{i-1} \right| + \varepsilon_i\sqrt{1 - \varpi^2}$, where $\varpi \in (0,1)$. From Table \ref{tab: NAR size}, the self-normalized tests have similar size distortion when $\left| \varpi \right| \leq 0.5$, while $\KS_n$ has slightly larger size. When $\left| \varpi \right| = 0.8$, $\shao_n^{(1)}$ and $\T_n^{(\C)}$ are more accurate than other tests while $\T_n^{(\W)}$ and $\T_n^{(\H)}$ are less size distorted comparing to $\shao_n^{(2)}$, $\shao_n^{(3)}$ and $\zhang_n$. A similar result regarding to the power is observed in Figures \ref{fig: power nar0.5}--\ref{fig: size adjusted power nar-0.5}. After adjusting the size, our proposed tests have the highest power in the multiple-CP case and lose the least power among the multiple-CP tests in the single CP case.

\begin{table}[t]
\setlength{\tabcolsep}{3pt}
\centering
\footnotesize
\caption{\footnotesize \label{tab: NAR size} Null rejection rates (\%) at 5\% nominal size under NAR model and mean function \eqref{EQ mean function normal}.}
\vspace{-0.3cm}
\begin{tabular}{ r  r rrr r rrr  r rrr r rrr}
\toprule
&  \multicolumn{8}{c}{$n=200$} & \multicolumn{8}{c}{$n=400$}\\ 
\cmidrule(r){2-9}\cmidrule(r){10-17}
$\varpi$& $\KS_n$ & $\shao_n^{(1)}$ & $\shao_n^{(2)}$ & $\shao_n^{(3)}$ & $\zhang_n$ & $\T_n^{(\C)}$ & $\T_n^{(\W)}$ & $\T_n^{(\H)}$ & $\KS_n$ &  $\shao_n^{(1)}$ & $\shao_n^{(2)}$ & $\shao_n^{(3)}$ & $\zhang_n$ & $\T_n^{(\C)}$ & $\T_n^{(\W)}$ & $\T_n^{(\H)}$  \\ 
\cmidrule(r){1-17}
0.8 & 18.1 & 6.4 & 13.1 & 19.3 & 14.2 & 7.6 & 10.2 & 9.6 & 16.0 & 4.9 & 7.2 & 9.5 & 7.6 & 4.8 & 7.1 & 6.4\\
0.5 & 8.2 & 4.3 & 5.8 & 4.4 & 5.3 & 4.1 & 5.2 & 4.5 & 6.3 & 3.9 & 4.1 & 4.6 & 3.8 & 4.4 & 5.0 & 4.9\\
0.3 & 7.2 & 4.7 & 4.0 & 3.1 & 3.9 & 4.1 & 4.6 & 4.7 & 5.7 & 4.2 & 3.4 & 4.4 & 3.3 & 5.0 & 4.8 & 4.7\\
0 & 7.3 & 4.8 & 4.4 & 2.5 & 3.3 & 4.1 & 5.0 & 4.9 & 5.0 & 3.9 & 3.2 & 4.8 & 3.3 & 5.1 & 4.6 & 4.5\\
$-0.3$ & 7.0 & 5.1 & 4.7 & 3.1 & 4.6 & 4.3 & 4.5 & 4.7 & 4.5 & 4.2 & 4.1 & 4.6 & 3.7 & 4.8 & 4.0 & 4.3\\
$-0.5$ & 8.3 & 4.9 & 5.4 & 5.3 & 5.1 & 4.2 & 4.8 & 5.1 & 6.5 & 5.2 & 5.2 & 5.8 & 4.2 & 5.5 & 4.9 & 5.1\\
$-0.8$ & 18.4 & 6.3 & 13.2 & 18.7 & 14.2 & 6.4 & 10.4 & 9.9 & 16.4 & 5.4 & 9.3 & 10.6 & 9.5 & 6.0 & 7.1 & 6.8\\

\bottomrule
\end{tabular} 
\end{table}

\begin{figure}[t]
\centering
\includegraphics[width=\linewidth]{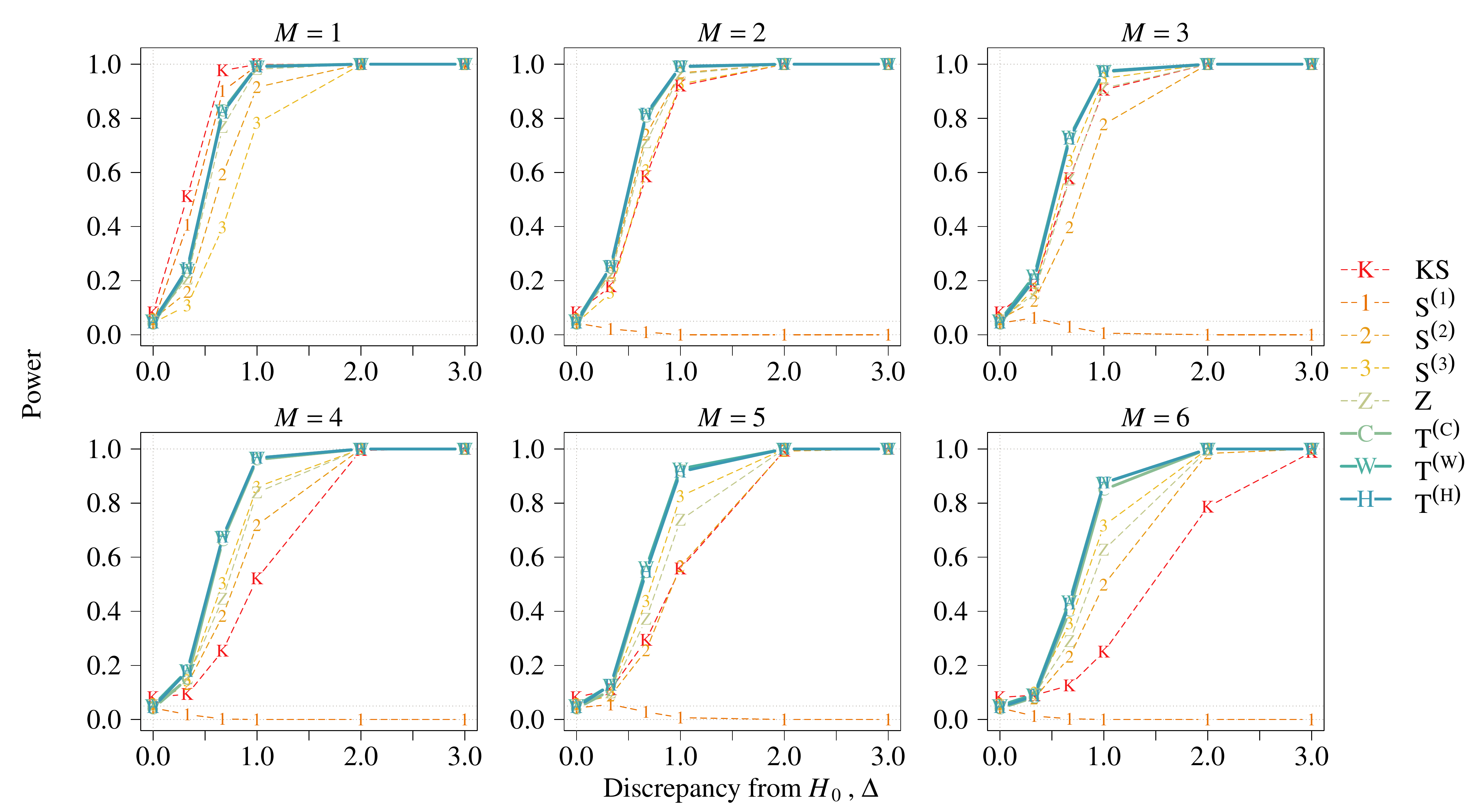}
\vspace{-0.3cm}
\caption{\footnotesize \label{fig: power nar0.5} Power under NAR model with $\varpi = 0.5$ and mean function \eqref{EQ mean function normal}.}
\end{figure}

\begin{figure}[H]
\centering
\includegraphics[width=\linewidth]{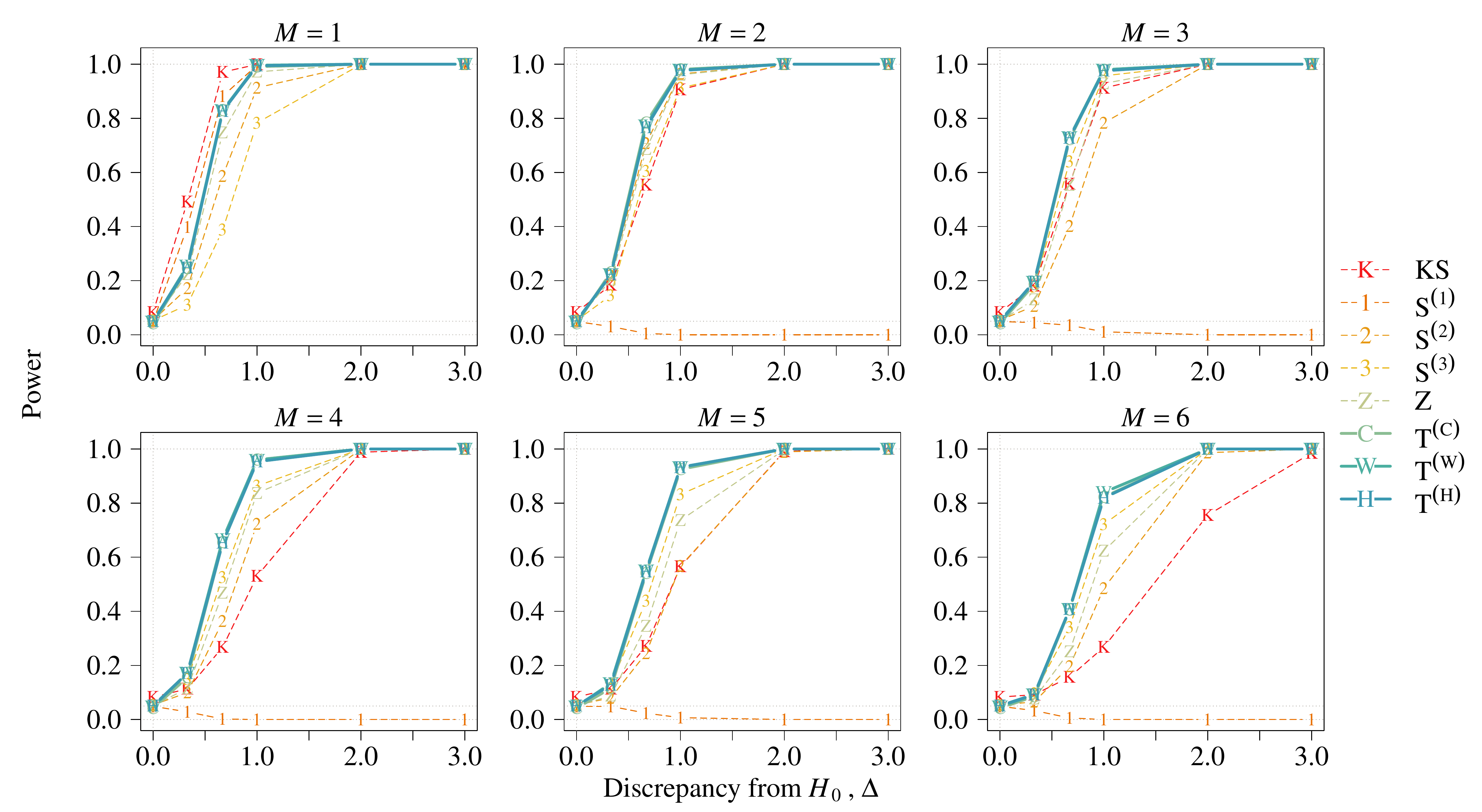}
\vspace{-0.3cm}
\caption{\footnotesize \label{fig: power nar-0.5} Power under NAR model with $\varpi = -0.5$ and mean function \eqref{EQ mean function normal}.}
\end{figure}

\begin{figure}[H]
\centering
\includegraphics[width=\linewidth]{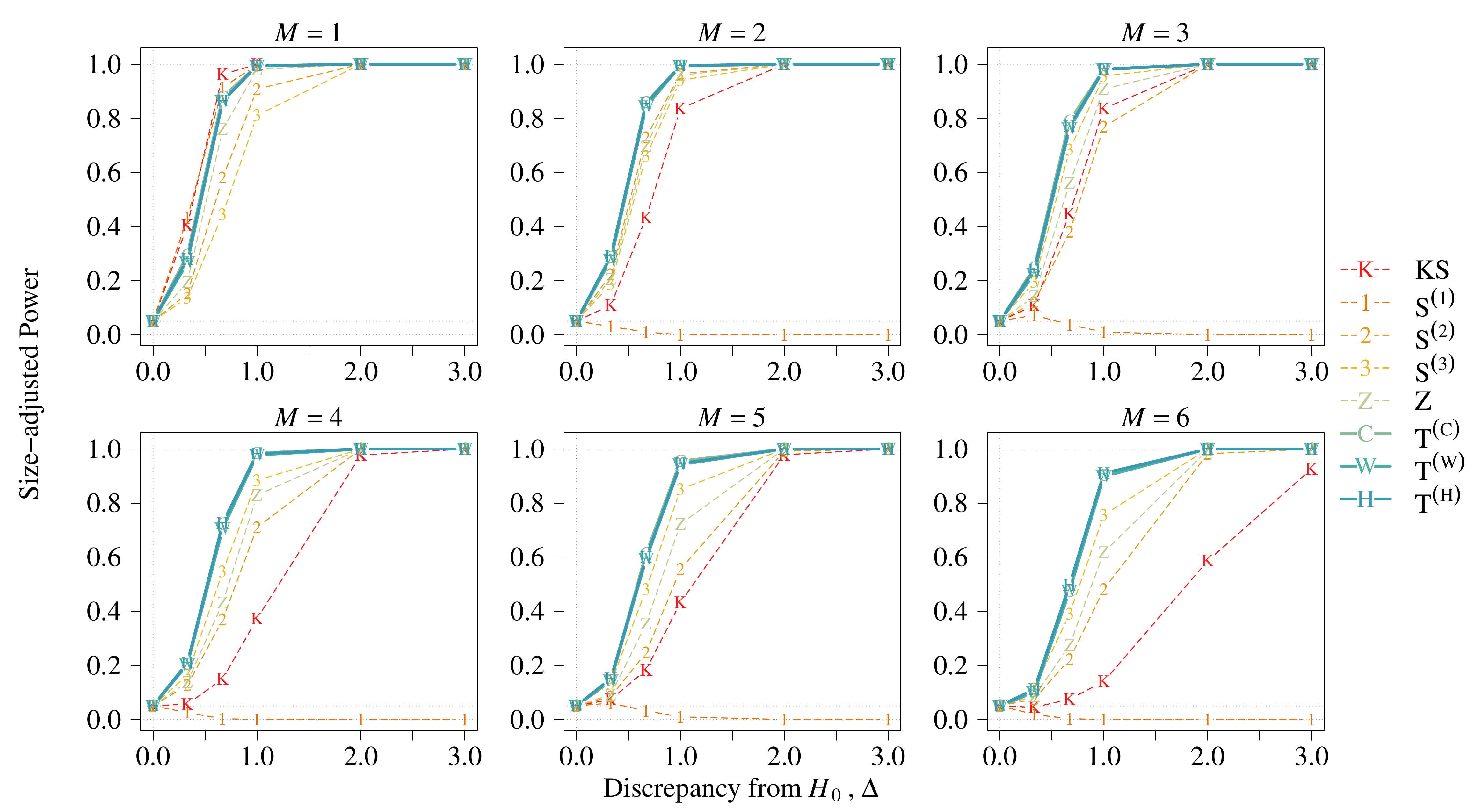}
\vspace{-0.3cm}
\caption{\footnotesize \label{fig: size adjusted power nar0.5} Size-adjusted power under NAR model with $\varpi = 0.5$ and mean function \eqref{EQ mean function normal}.}
\end{figure}

\begin{figure}[H]
\centering
\includegraphics[width=\linewidth]{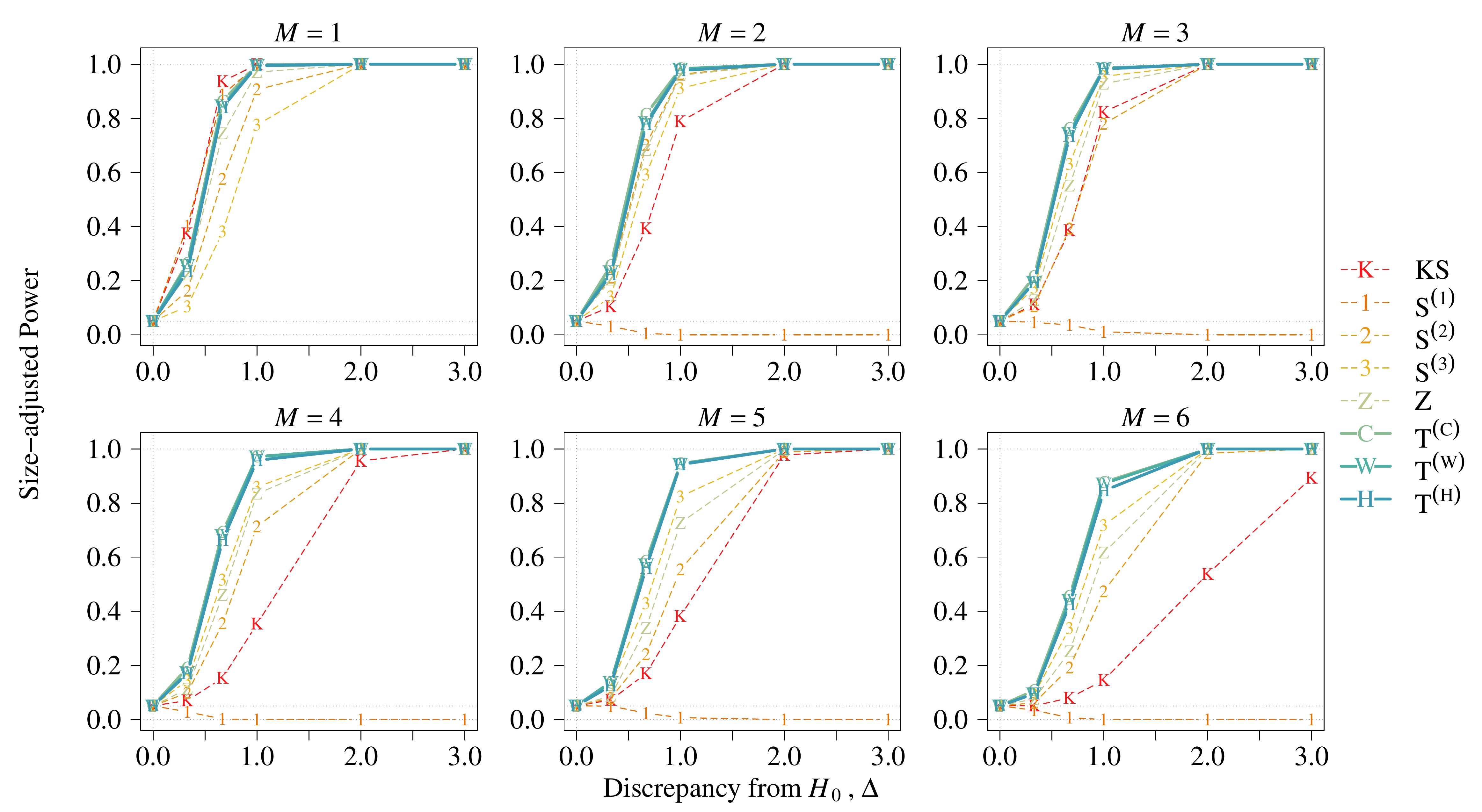}
\vspace{-0.3cm}
\caption{\footnotesize \label{fig: size adjusted power nar-0.5} Size-adjusted power under NAR model with $\varpi = -0.5$ and mean function \eqref{EQ mean function normal}.}
\end{figure}

\subsection{BAR Model}
\label{sec: BAR model}
Supplementing the result in the main text, we further report the power when $\vartheta = 0.8, -0.8$. 

\begin{figure}[H]
\centering
\includegraphics[width=0.95\linewidth]{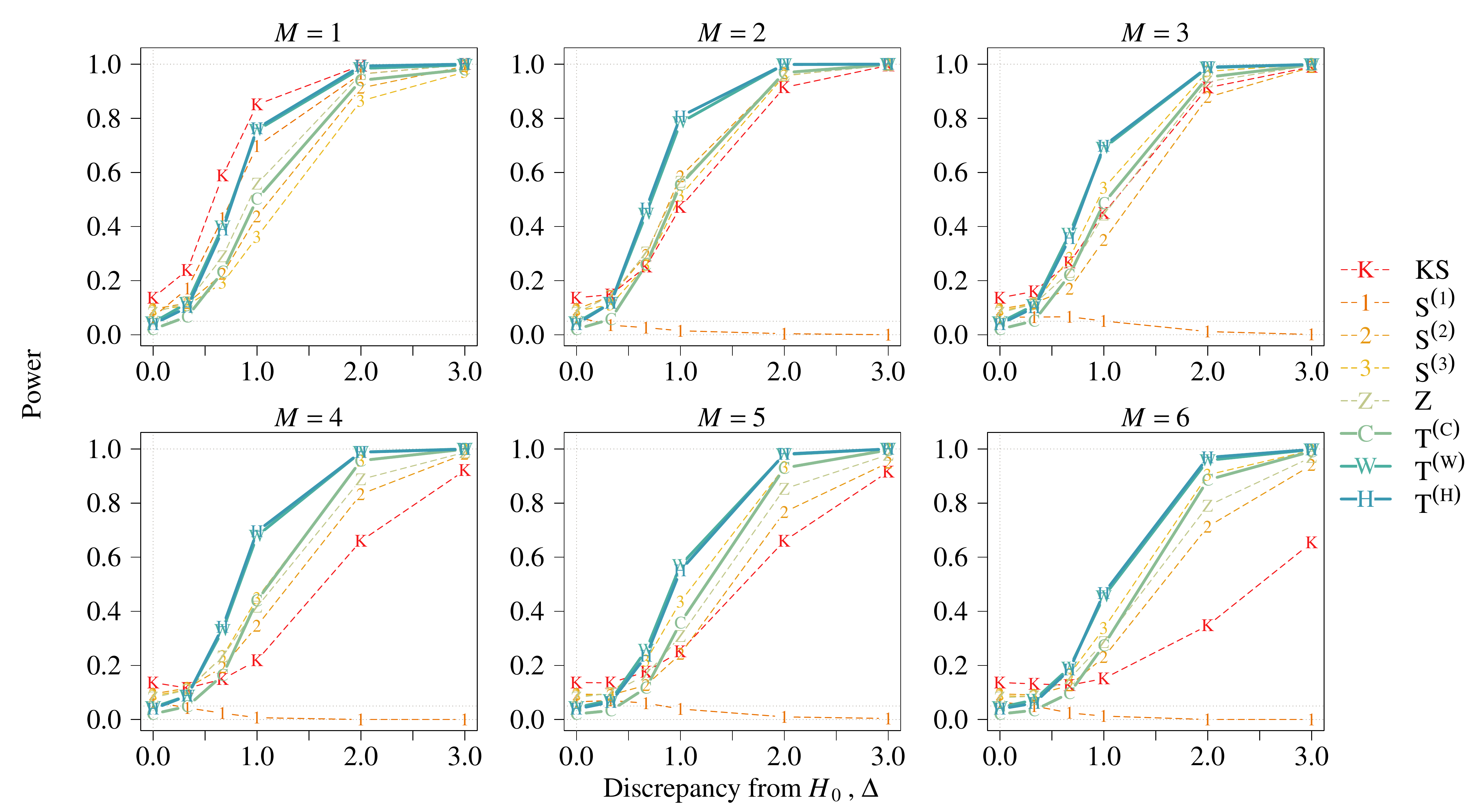}
\vspace{-0.3cm}
\caption{\footnotesize \label{fig: raw power ar0.5g0.5} Power under BAR model with $\varpi = \vartheta = 0.5$ and mean function \eqref{EQ mean function normal}.}
\end{figure}

\begin{figure}[H]
\centering
\includegraphics[width=\linewidth]{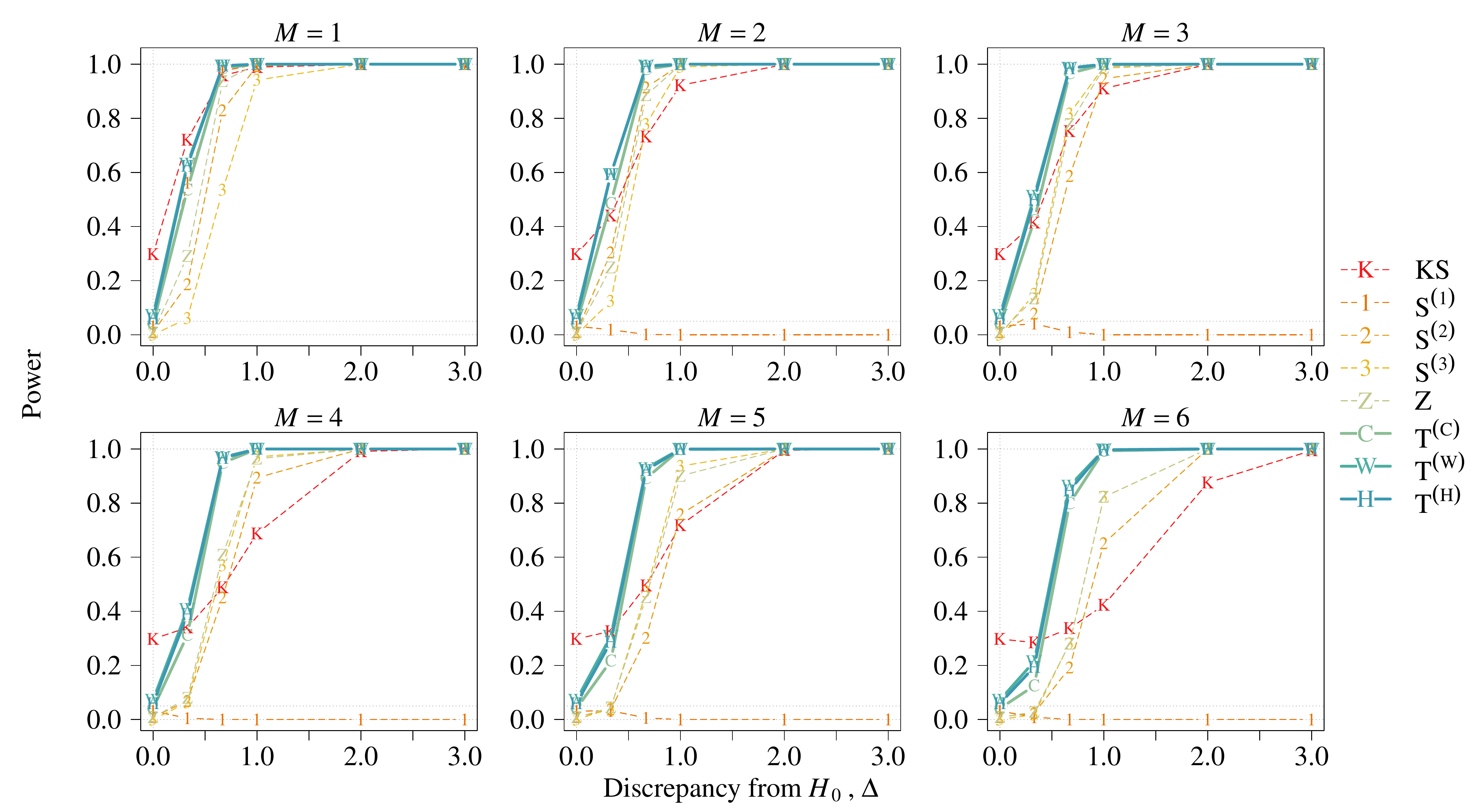}
\vspace{-0.3cm}
\caption{\footnotesize \label{fig: power bar ar-0.5 gamma0.5} Power under BAR model with $-\varpi = \vartheta = 0.5$ and mean function \eqref{EQ mean function normal}.}
\end{figure}

\begin{figure}[H]
\centering
\includegraphics[width=\linewidth]{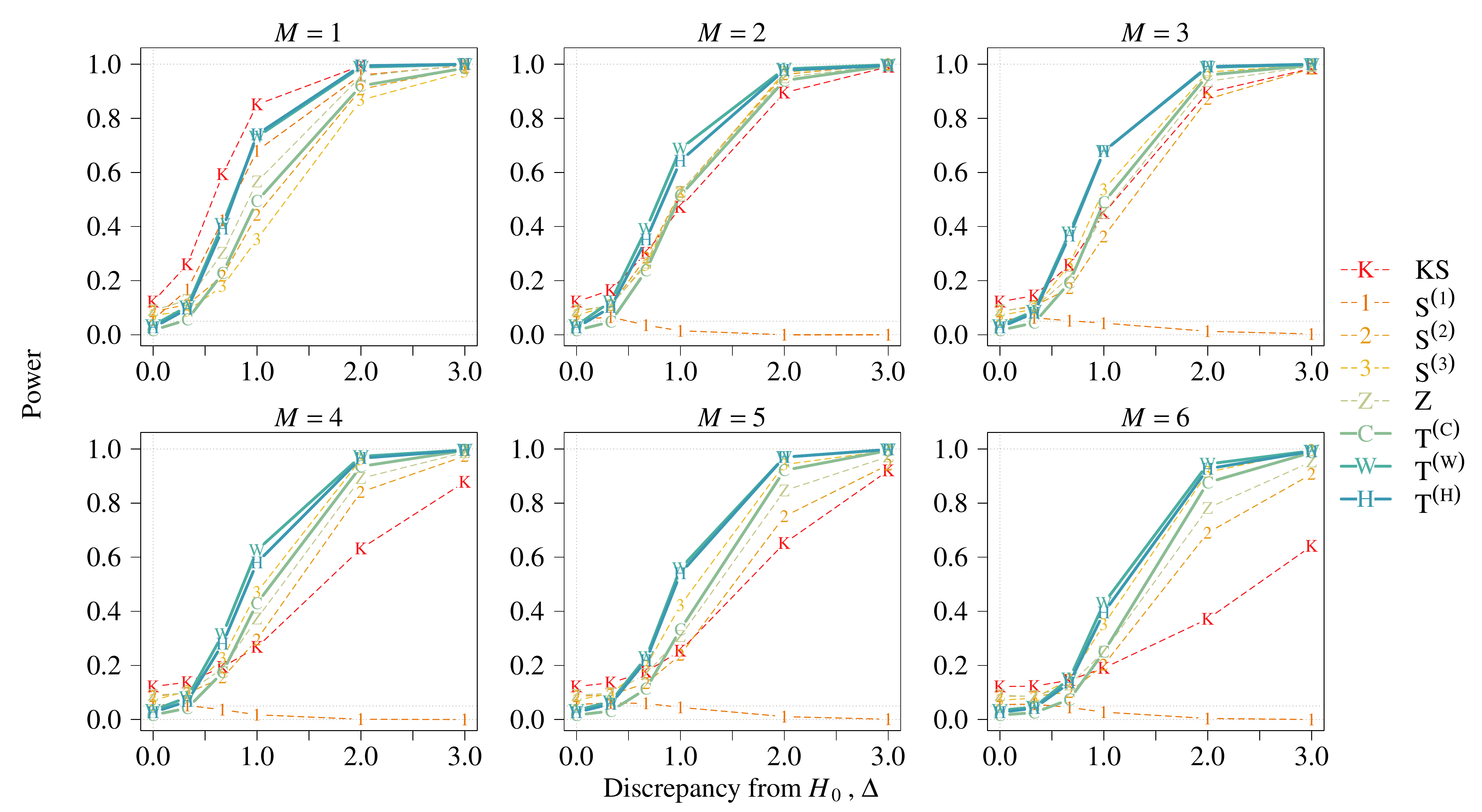}
\vspace{-0.3cm}
\caption{\footnotesize \label{fig: raw power ar0.5g-0.5} Power under BAR model with $\varpi = -\vartheta= 0.5$ and mean function \eqref{EQ mean function normal}.}
\end{figure}

\begin{figure}[H]
\centering
\includegraphics[width=\linewidth]{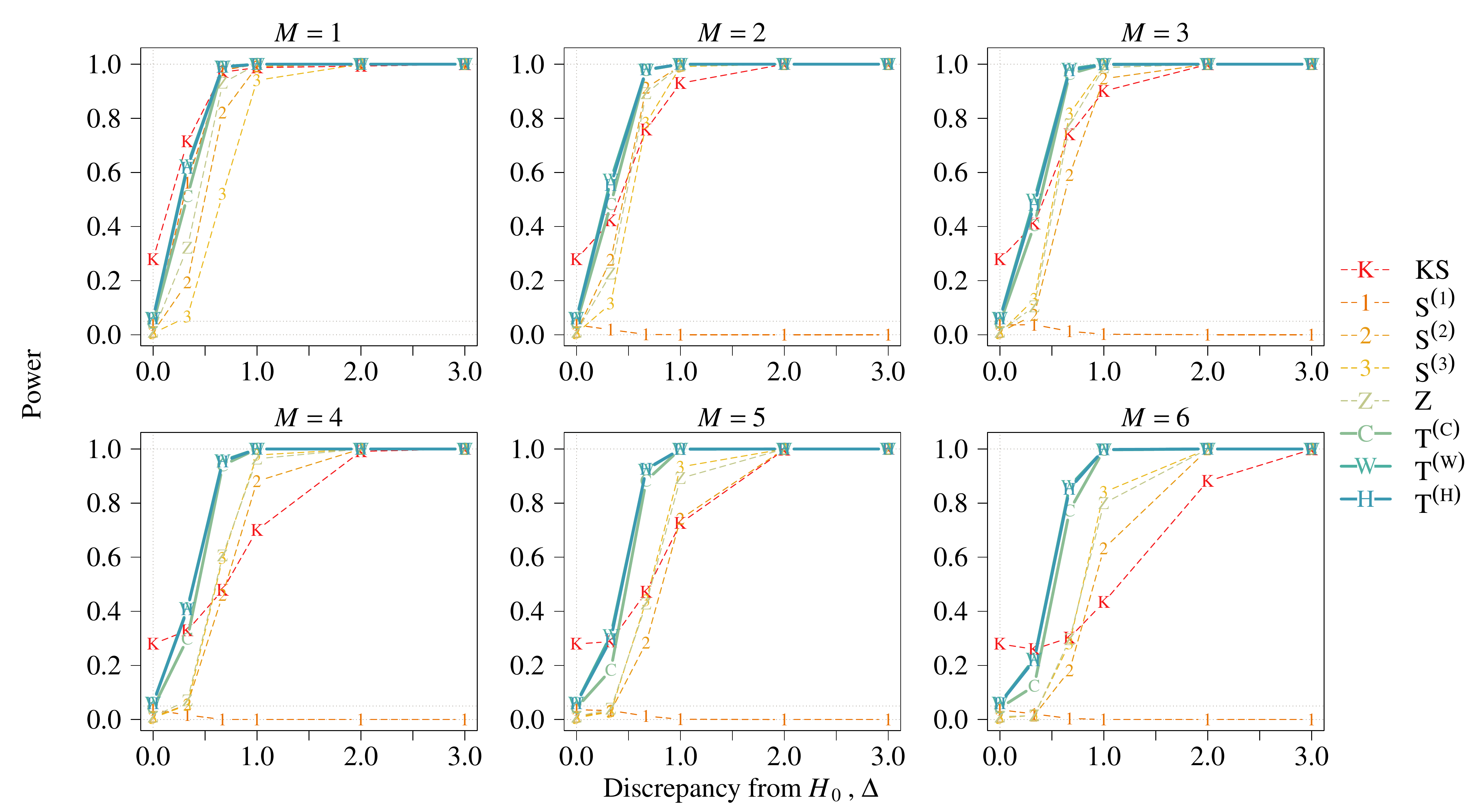}
\vspace{-0.3cm}
\caption{\footnotesize \label{fig: power bar ar-0.5 gamma-0.5} Power under BAR model with $\varpi = \vartheta = -0.5$ and mean function \eqref{EQ mean function normal}.}
\end{figure}

\begin{figure}[H]
\centering
\includegraphics[width=\linewidth]{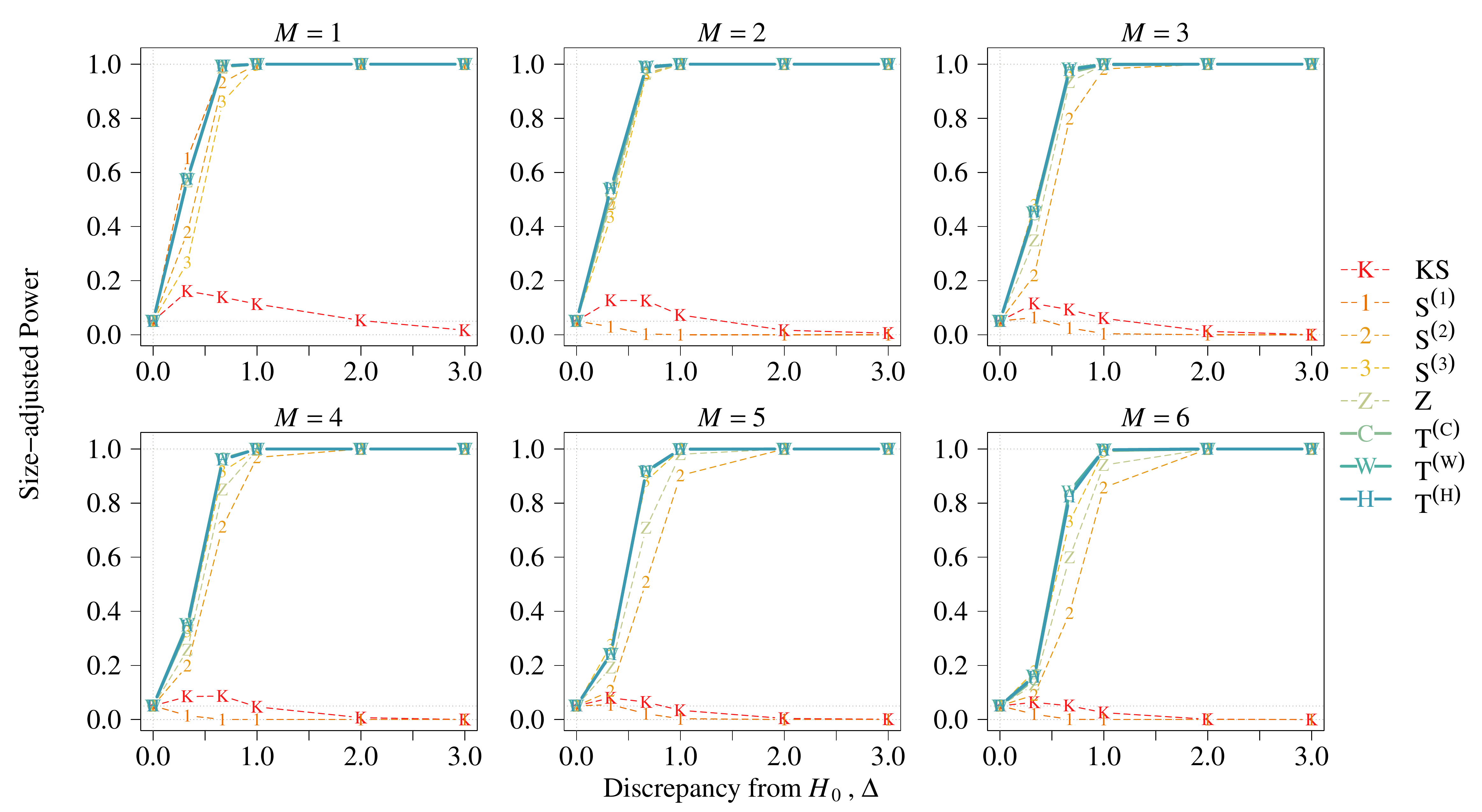}
\vspace{-0.3cm}
\caption{\footnotesize \label{fig: size adjusted power bar ar-0.5 gamma0.5} Size-adjusted power under BAR model with $-\varpi = \vartheta = 0.5$ and mean function \eqref{EQ mean function normal}.}
\end{figure}

\begin{figure}[H]
\centering
\includegraphics[width=\linewidth]{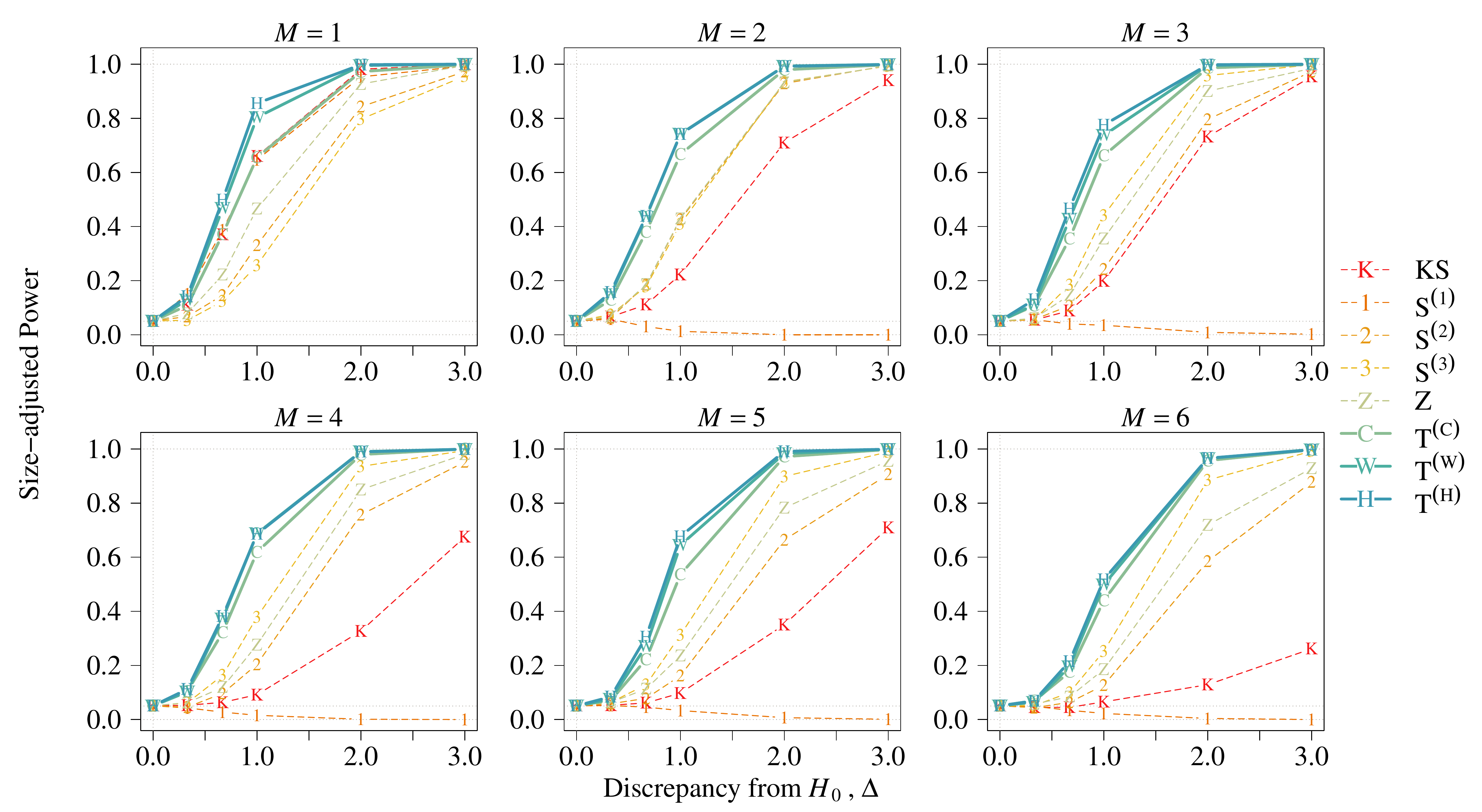}
\vspace{-0.3cm}
\caption{\footnotesize \label{fig: size adjusted power ar0.5g-0.5}  Size-adjusted power under BAR model with $\varpi = -\vartheta = 0.5$ and mean function \eqref{EQ mean function normal}.}
\end{figure}

\begin{figure}[H]
\centering
\includegraphics[width=\linewidth]{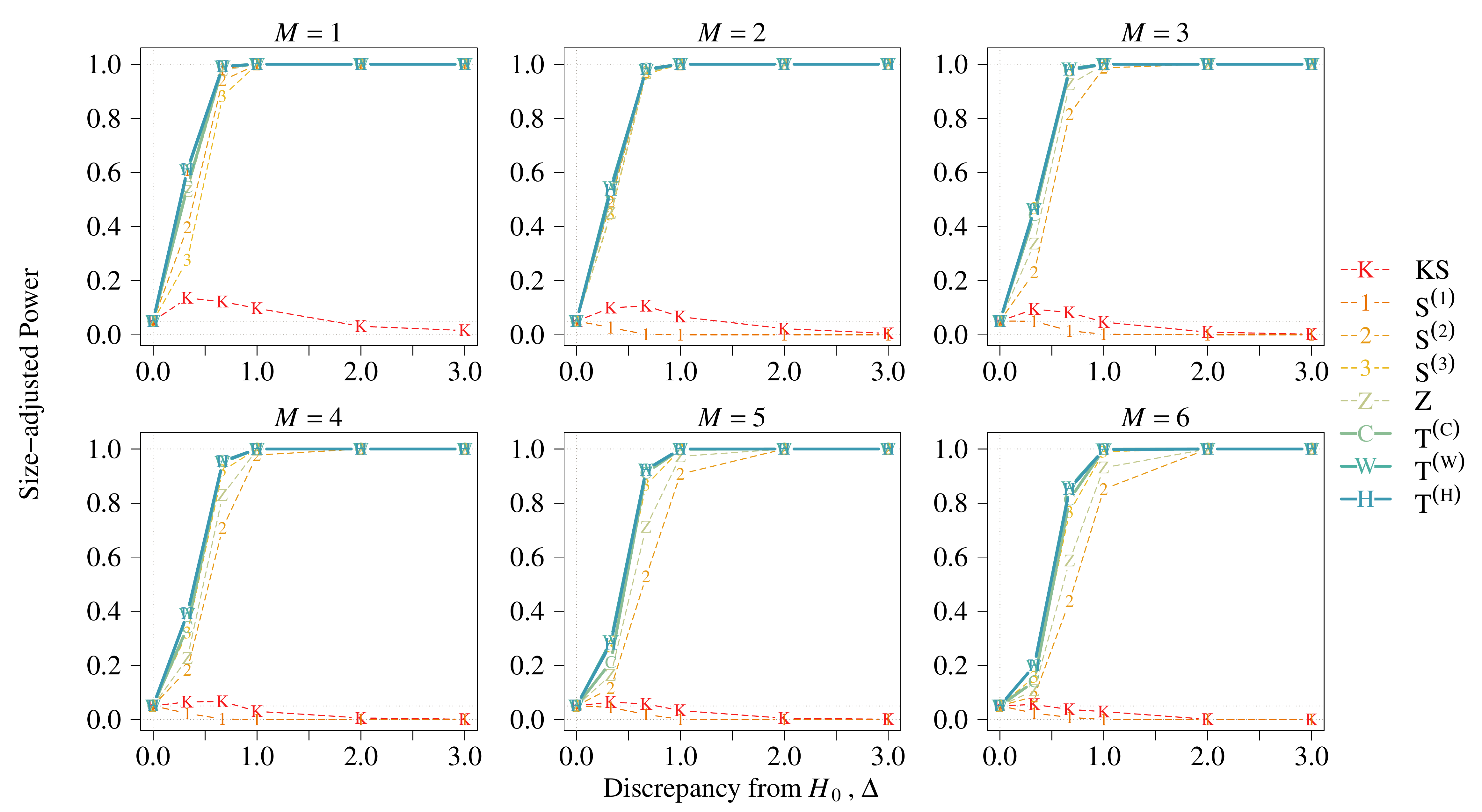}
\vspace{-0.3cm}
\caption{\footnotesize \label{fig: size adjusted power bar ar-0.5 gamma-0.5} Size-adjusted power under BAR model with $\varpi = \vartheta = -0.5$ and mean function \eqref{EQ mean function normal}.}
\end{figure}

\subsection{Non-Gaussian time series}
\label{sec: non-g time series}
Consider again the signal-plus-noise model, we further investigate the performance under the AR model and the BAR model if the innovations $\varepsilon$ follow the $t$-distribution with degree of freedom $5$ instead of the standard normal distribution. Tables \ref{tab: ar size non-G t dist} and \ref{tab: bar size non-G t dist} present the size under the AR model and the BAR model respectively. Figures \ref{fig: ar05 raw power non-G t dist}--\ref{fig: ar-05 adjusted power non-G t dist} and Figures \ref{fig: ar05g05 raw power non-G t dist}--\ref{fig: ar-05g05 adjusted powernon-G t dist} present the power under the AR model and the BAR model respectively. Similar result is also observed. Our proposed tests, $\T_n^{(\C)}$, $\T_n^{(\W)}$ and $\T_n^{(\H)}$, have higher size accuracy than other multiple-CP tests. Moreover, our tests have the largest power under multiple-CP case while lose the least power among the multiple-CP tests comparing to AMOC tests $\KS_n$ and $\shao_n^{(1)}$ in single CP case.

\begin{figure}[H]
\centering
\includegraphics[width=\linewidth]{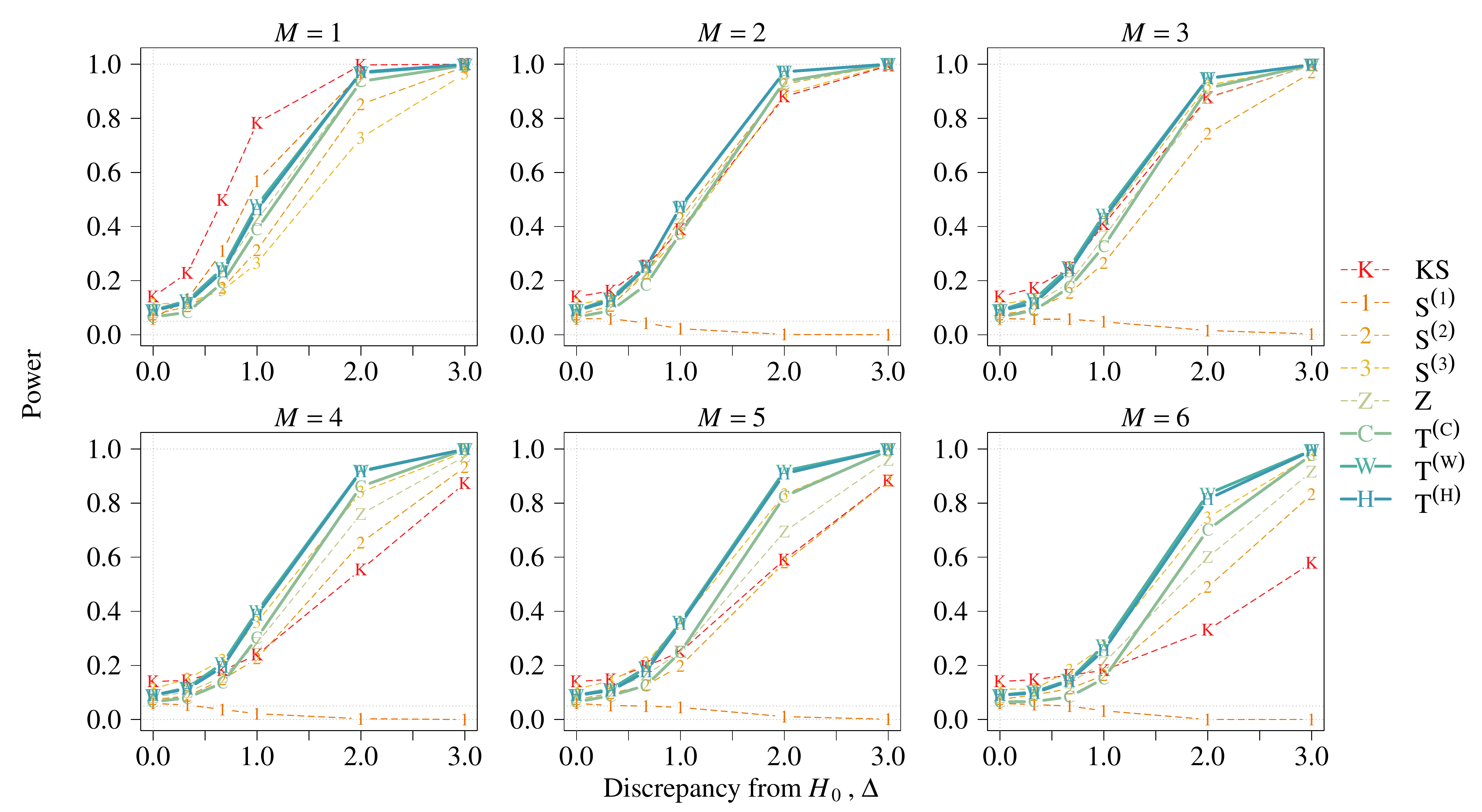}
\vspace{-0.3cm}
\caption{\footnotesize \label{fig: ar05 raw power non-G t dist} Power under AR model with $\varpi = 0.5$, mean function \eqref{EQ mean function normal} and innovations following $t$-distribution.}
\end{figure}

\begin{table}[t]
\setlength{\tabcolsep}{3pt}
\centering
\footnotesize
\caption{\footnotesize \label{tab: ar size non-G t dist} Null rejection rates (\%) at 5\% nominal size under AR model with innovations following $t$-distribution.}
\vspace{-0.3cm}
\begin{tabular}{ r  r rrr r rrr  r rrr r rrr}
\toprule
&  \multicolumn{8}{c}{$n=200$} & \multicolumn{8}{c}{$n=400$}\\
\cmidrule(r){2-9}\cmidrule(r){10-17}
$\varpi$& $\KS_n$ & $\shao_n^{(1)}$ & $\shao_n^{(2)}$ & $\shao_n^{(3)}$ & $\zhang_n$ & $\T_n^{(\C)}$ & $\T_n^{(\W)}$ & $\T_n^{(\H)}$ & $\KS_n$ &  $\shao_n^{(1)}$ & $\shao_n^{(2)}$ & $\shao_n^{(3)}$ & $\zhang_n$ & $\T_n^{(\C)}$ & $\T_n^{(\W)}$ & $\T_n^{(\H)}$  \\ 
\cmidrule(r){1-17}
0.8 & 53.4 & 8.4 & 24.2 & 41.3 & 29.5 & 15.9 & 25.9 & 23.8 & 54.7 & 6.8 & 14.1 & 23.1 & 17.0 & 8.8 & 12.8 & 12.5\\
0.5 & 14.1 & 6.0 & 7.4 & 11.3 & 8.3 & 6.5 & 9.0 & 9.0 & 11.1 & 5.7 & 6.5 & 8.1 & 7.0 & 4.8 & 6.4 & 6.3\\
0.3 & 10.6 & 5.8 & 5.5 & 6.3 & 6.2 & 5.6 & 7.7 & 7.3 & 8.5 & 5.2 & 4.9 & 4.9 & 4.8 & 4.8 & 6.3 & 6.2\\
0 & 7.4 & 5.2 & 3.5 & 2.5 & 3.1 & 4.1 & 4.7 & 4.6 & 5.0 & 4.2 & 4.5 & 4.6 & 3.8 & 5.2 & 4.7 & 4.2\\
$-0.3$ & 14.6 & 4.4 & 2.2 & 1.5 & 1.8 & 5.2 & 5.9 & 5.5 & 10.4 & 4.9 & 3.6 & 2.1 & 3.1 & 5.3 & 5.9 & 5.8\\
$-0.5$ & 27.6 & 3.6 & 1.5 & 0.6 & 1.2 & 5.5 & 6.1 & 6.0 & 19.8 & 4.7 & 2.9 & 1.5 & 2.3 & 5.4 & 6.3 & 6.4\\
$-0.8$ & 46.9 & 1.8 & 0.3 & 0.0 & 0.0 & 8.6 & 9.0 & 8.5 & 41.2 & 3.7 & 1.5 & 0.2 & 1.0 & 6.6 & 7.9 & 7.3\\

\bottomrule
\end{tabular} 
\end{table}

\begin{table}[t]
\setlength{\tabcolsep}{3pt}
\centering
\footnotesize
\caption{\footnotesize \label{tab: bar size non-G t dist} Null rejection rates (\%) at 5\% nominal size under BAR model with innovations following $t$-distribution.}
\vspace{-0.3cm}
\begin{tabular}{ rr  r rrr r rrr  r rrr r rrr}
\toprule
&  \multicolumn{8}{c}{$n=200$} & \multicolumn{8}{c}{$n=400$}\\
\cmidrule(r){3-10}\cmidrule(r){11-18}
$\vartheta$ &$\varpi$& $\KS_n$ & $\shao_n^{(1)}$ & $\shao_n^{(2)}$ & $\shao_n^{(3)}$ & $\zhang_n$ & $\T_n^{(\C)}$ & $\T_n^{(\W)}$ & $\T_n^{(\H)}$ & $\KS_n$ &  $\shao_n^{(1)}$ & $\shao_n^{(2)}$ & $\shao_n^{(3)}$ & $\zhang_n$ & $\T_n^{(\C)}$ & $\T_n^{(\W)}$ & $\T_n^{(\H)}$  \\ 
\cmidrule(r){1-18}
0.8 & 0.5 & 15.5 & 8.1 & 6.5 & 3.8 & 5.7 & 0.3 & 4.1 & 2.2 & 15.1 & 7.1 & 6.3 & 2.5 & 4.8 & 0.8 & 4.3 & 3.0\\
 & 0.3 & 17.8 & 6.7 & 4.5 & 2.3 & 3.4 & 0.7 & 3.0 & 2.7 & 15.6 & 7.0 & 4.2 & 1.8 & 3.3 & 0.9 & 4.9 & 4.2\\
 & 0 & 19.7 & 5.2 & 2.4 & 1.1 & 2.2 & 0.9 & 5.7 & 3.6 & 16.6 & 5.9 & 2.6 & 0.9 & 2.0 & 1.4 & 6.5 & 5.1\\
 & $-0.3$ & 29.4 & 5.2 & 1.9 & 0.0 & 1.5 & 1.0 & 6.1 & 4.8 & 22.6 & 5.7 & 1.2 & 0.4 & 0.9 & 0.8 & 6.9 & 5.8\\
 & $-0.5$ & 31.3 & 4.7 & 1.3 & 0.3 & 1.0 & 1.3 & 7.6 & 5.0 & 29.2 & 6.1 & 1.6 & 0.2 & 0.7 & 0.8 & 7.2 & 5.6\\

\cmidrule(r){1-18}

0.5 & 0.8 & 27.1 & 10.3 & 22.9 & 20.9 & 21.3 & 2.1 & 10.7 & 7.5 & 29.2 & 8.5 & 11.6 & 8.9 & 11.1 & 1.0 & 4.3 & 4.2\\
 & 0.5 & 14.3 & 6.0 & 8.0 & 6.7 & 7.2 & 2.1 & 4.6 & 3.8 & 11.3 & 5.4 & 6.8 & 5.4 & 6.1 & 3.0 & 5.4 & 4.7\\
 & 0.3 & 12.1 & 4.8 & 5.6 & 3.4 & 4.6 & 2.3 & 4.9 & 4.1 & 8.4 & 5.8 & 5.6 & 4.2 & 5.0 & 3.9 & 6.4 & 6.2\\
 & 0 & 14.6 & 3.6 & 3.2 & 1.7 & 2.9 & 2.8 & 5.4 & 5.1 & 10.0 & 6.6 & 4.1 & 2.8 & 3.7 & 4.0 & 6.1 & 5.7\\
 & $-0.3$ & 21.0 & 3.5 & 2.1 & 0.8 & 1.8 & 3.3 & 6.3 & 5.4 & 15.7 & 5.6 & 2.6 & 2.1 & 1.9 & 4.0 & 6.2 & 5.6\\
 & $-0.5$ & 30.6 & 3.9 & 1.1 & 0.3 & 0.8 & 3.2 & 7.2 & 5.9 & 25.6 & 5.0 & 1.4 & 1.1 & 1.4 & 3.6 & 7.5 & 7.0\\
 & $-0.8$ & 40.7 & 3.0 & 0.7 & 0.1 & 0.4 & 2.7 & 8.2 & 7.6 & 40.5 & 4.2 & 0.7 & 0.1 & 0.3 & 2.1 & 7.3 & 6.6\\
\cmidrule(r){1-18}

$-0.5$ & 0.8 & 28.6 & 10.4 & 25.6 & 22.6 & 24.4 & 2.2 & 10.5 & 8.3 & 29.0 & 9.4 & 11.2 & 6.4 & 10.7 & 1.3 & 4.9 & 3.5\\
 & 0.5 & 13.1 & 6.9 & 8.3 & 6.7 & 7.5 & 1.7 & 4.5 & 3.7 & 10.7 & 6.2 & 4.2 & 3.8 & 3.8 & 1.8 & 5.2 & 4.5\\
 & 0.3 & 11.1 & 5.4 & 4.7 & 2.8 & 4.6 & 1.9 & 4.2 & 3.7 & 8.9 & 5.0 & 2.9 & 2.9 & 2.9 & 2.2 & 5.1 & 4.5\\
 & 0 & 12.5 & 4.5 & 2.7 & 1.7 & 2.2 & 2.3 & 5.4 & 4.6 & 9.9 & 5.4 & 3.0 & 1.7 & 2.0 & 2.9 & 5.9 & 5.2\\
 & $-0.3$ & 21.1 & 4.1 & 1.8 & 0.5 & 1.5 & 2.3 & 6.8 & 6.8 & 16.1 & 4.8 & 2.4 & 1.2 & 1.8 & 3.0 & 4.9 & 4.2\\
 & $-0.5$ & 30.2 & 3.4 & 0.9 & 0.1 & 0.9 & 2.9 & 7.5 & 6.3 & 28.5 & 5.2 & 1.9 & 0.7 & 1.5 & 2.8 & 6.0 & 5.5\\
 & $-0.8$ & 42.6 & 3.1 & 0.7 & 0.0 & 0.6 & 2.7 & 7.4 & 6.5 & 41.9 & 3.8 & 0.3 & 0.1 & 0.4 & 2.0 & 5.6 & 4.4\\
\cmidrule(r){1-18}

$-0.8$ & 0.5 & 18.0 & 6.9 & 8.8 & 3.4 & 5.9 & 0.3 & 4.3 & 3.1 & 14.9 & 5.9 & 3.6 & 1.4 & 3.0 & 0.2 & 2.7 & 1.8\\
 & 0.3 & 17.4 & 5.6 & 4.7 & 1.7 & 3.7 & 0.4 & 4.3 & 3.2 & 14.9 & 6.0 & 2.7 & 1.1 & 2.1 & 0.7 & 4.6 & 3.5\\
 & 0 & 20.4 & 4.2 & 2.7 & 0.9 & 2.1 & 0.7 & 4.9 & 2.5 & 16.2 & 6.1 & 1.9 & 0.4 & 1.5 & 0.8 & 4.8 & 4.1\\
 & $-0.3$ & 25.9 & 5.2 & 1.3 & 0.1 & 1.3 & 0.7 & 6.5 & 5.5 & 23.7 & 6.0 & 2.0 & 0.7 & 1.3 & 0.8 & 5.5 & 4.4\\
 & $-0.5$ & 34.0 & 4.8 & 2.2 & 0.3 & 1.7 & 1.0 & 7.9 & 6.0 & 30.5 & 5.4 & 1.7 & 0.4 & 0.8 & 0.7 & 5.2 & 4.7\\

\bottomrule
\end{tabular} 
\end{table}
\clearpage

\begin{figure}[H]
\centering
\includegraphics[width=\linewidth]{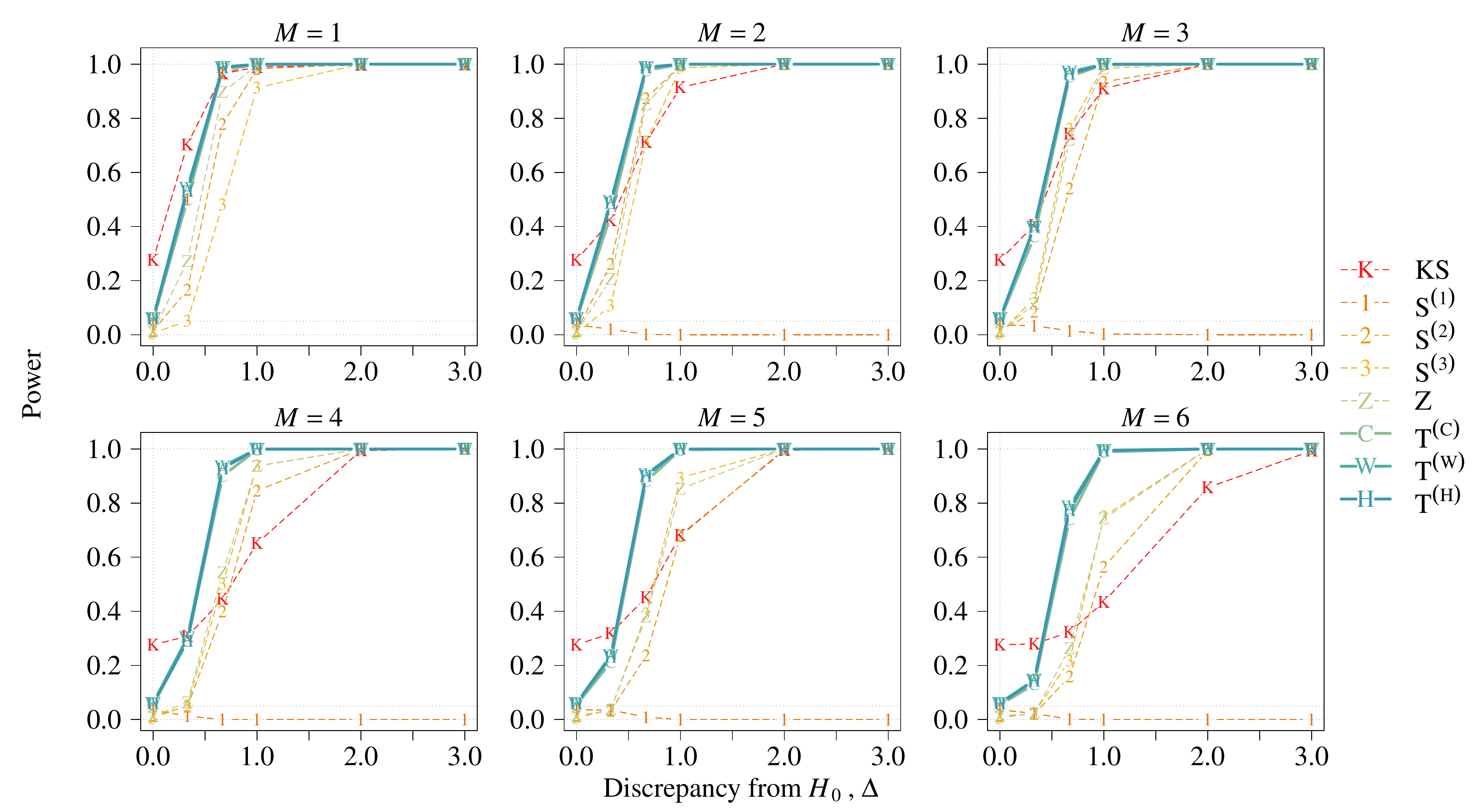}
\vspace{-0.3cm}
\caption{\footnotesize \label{fig: ar-05 raw power non-G t dist} Power under AR model with $\varpi = -0.5$, mean function \eqref{EQ mean function normal} and innovations following $t$-distribution.}
\end{figure}

\begin{figure}[H]
\centering
\includegraphics[width=\linewidth]{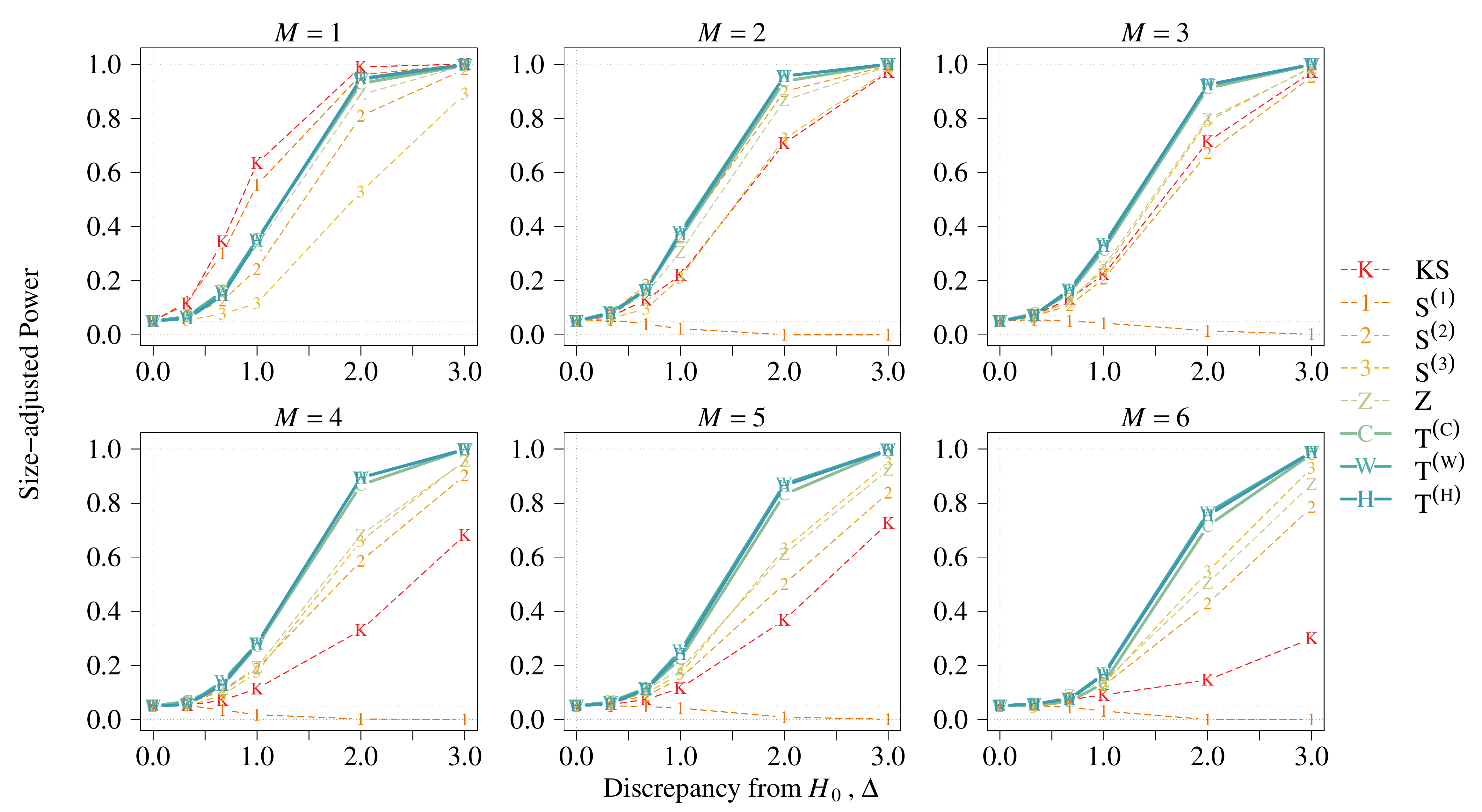}
\vspace{-0.3cm}
\caption{\footnotesize \label{fig: ar05 adjusted power non-G t dist} Size-adjusted power under AR model with $\varpi = 0.5$, mean function \eqref{EQ mean function normal} and innovations following $t$-distribution.}
\end{figure}

\begin{figure}[H]
\centering
\includegraphics[width=\linewidth]{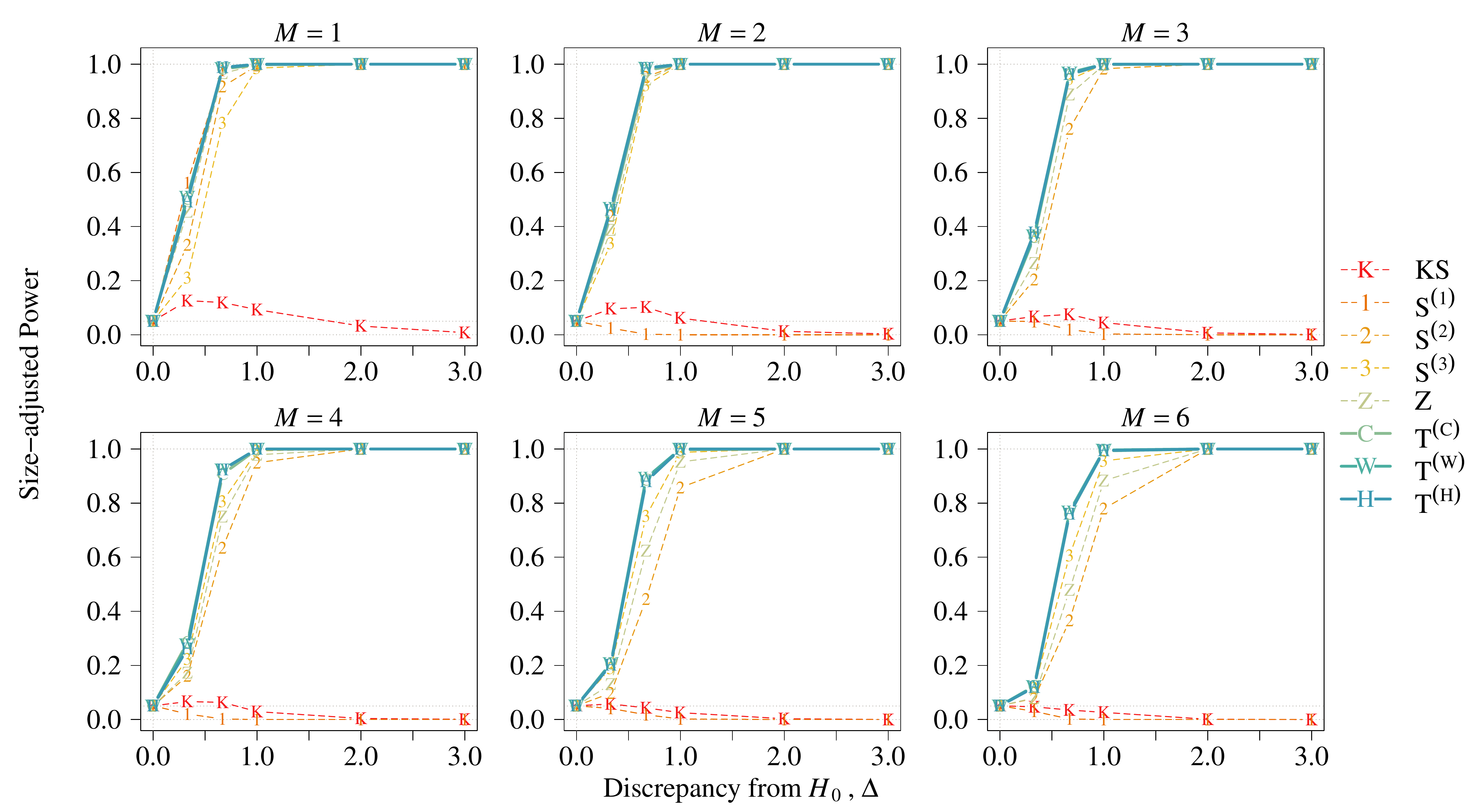}
\vspace{-0.3cm}
\caption{\footnotesize \label{fig: ar-05 adjusted power non-G t dist} Size-adjusted power under AR model with $\varpi= -0.5$, mean function \eqref{EQ mean function normal} and innovations following $t$-distribution.}
\end{figure}

\begin{figure}[H]
\centering
\includegraphics[width=\linewidth]{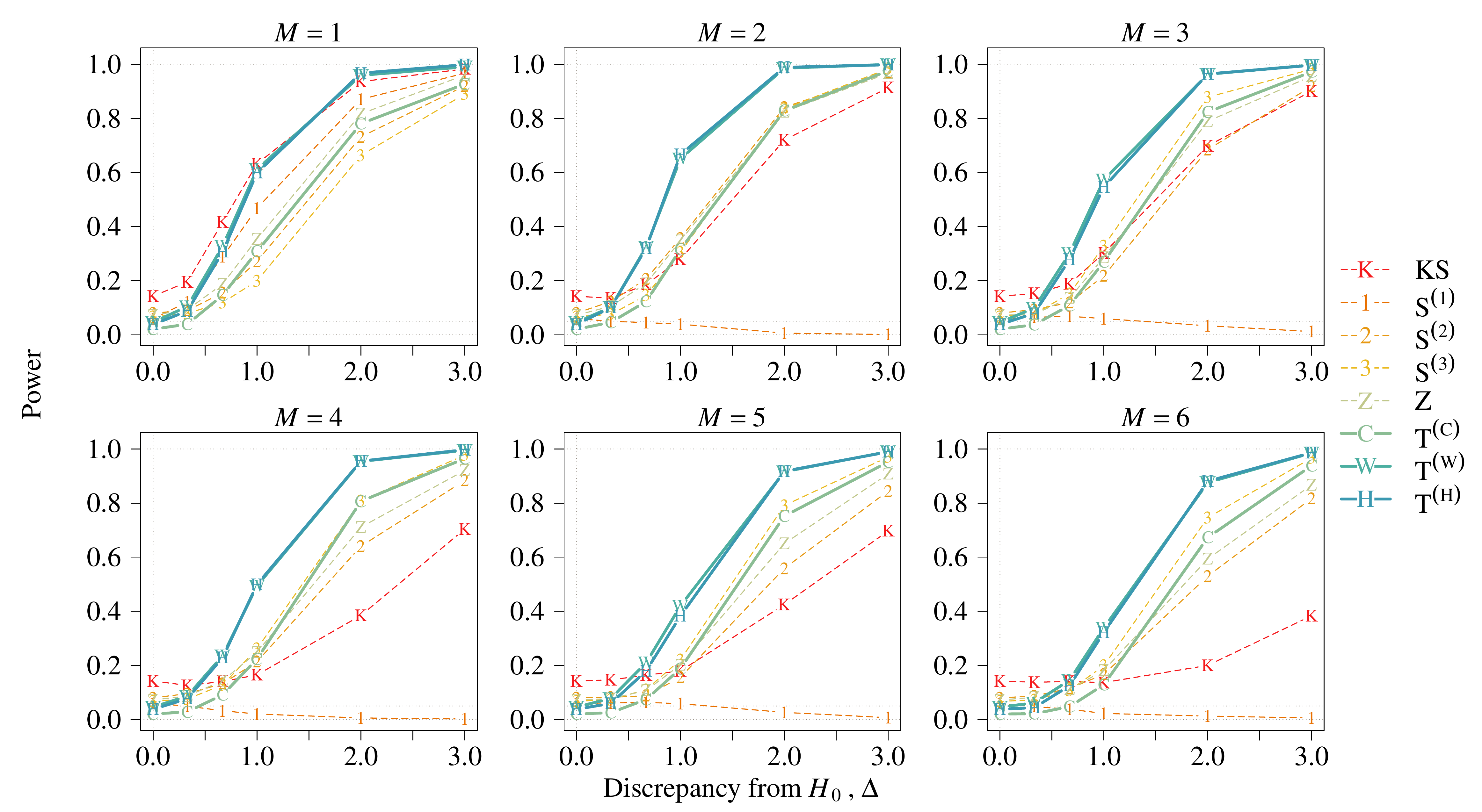}
\vspace{-0.3cm}
\caption{\footnotesize \label{fig: ar05g05 raw power non-G t dist} Power under BAR model with $\varpi = \vartheta = 0.5$, mean function \eqref{EQ mean function normal} and innovations following $t$-distribution.}
\end{figure}

\begin{figure}[H]
\centering
\includegraphics[width=\linewidth]{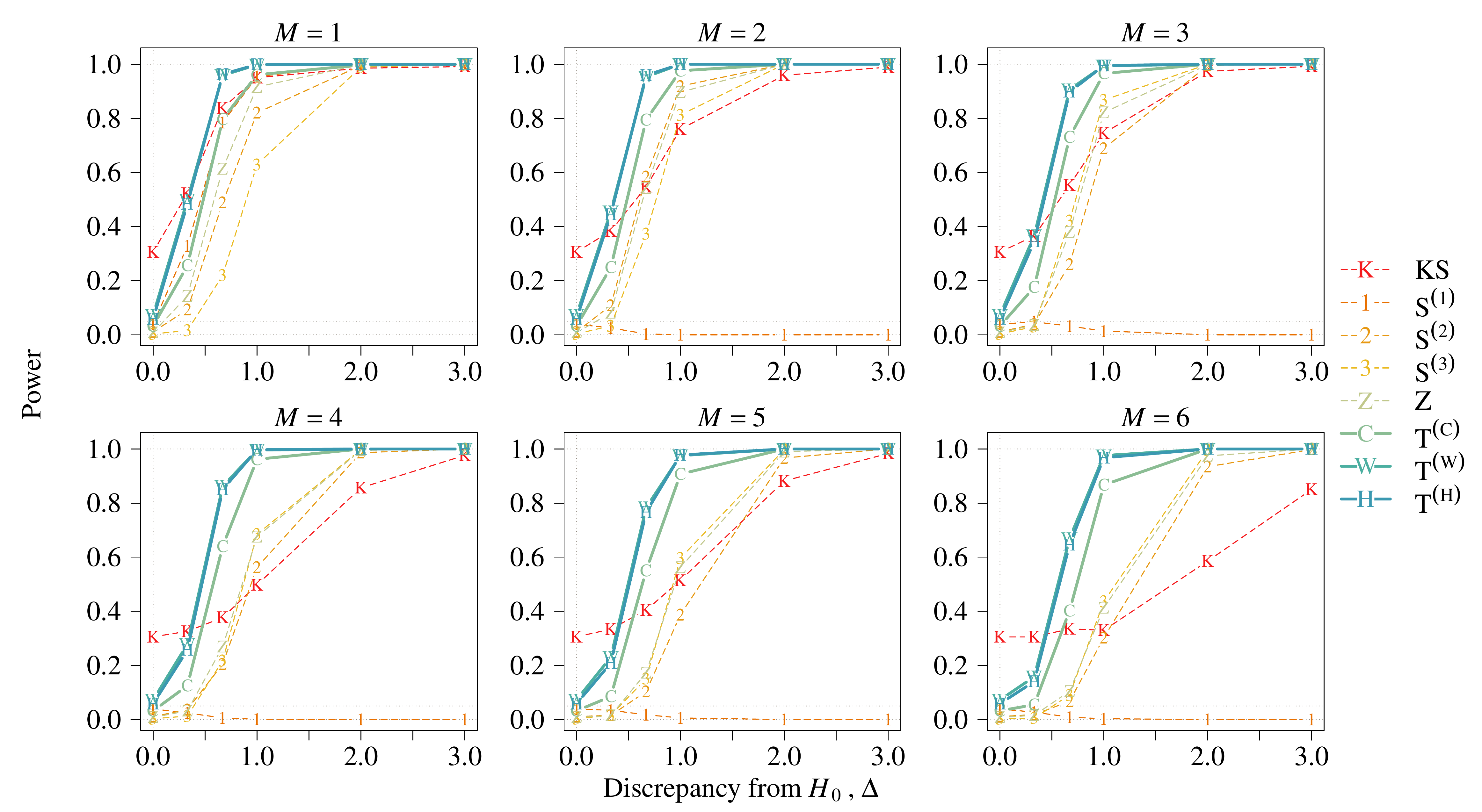}
\vspace{-0.3cm}
\caption{\footnotesize \label{fig: ar-05g05 raw power non-G t dist} Power under BAR model with $-\varpi = \vartheta = 0.5$, mean function \eqref{EQ mean function normal} and innovations following $t$-distribution.}
\end{figure}

\begin{figure}[H]
\centering
\includegraphics[width=\linewidth]{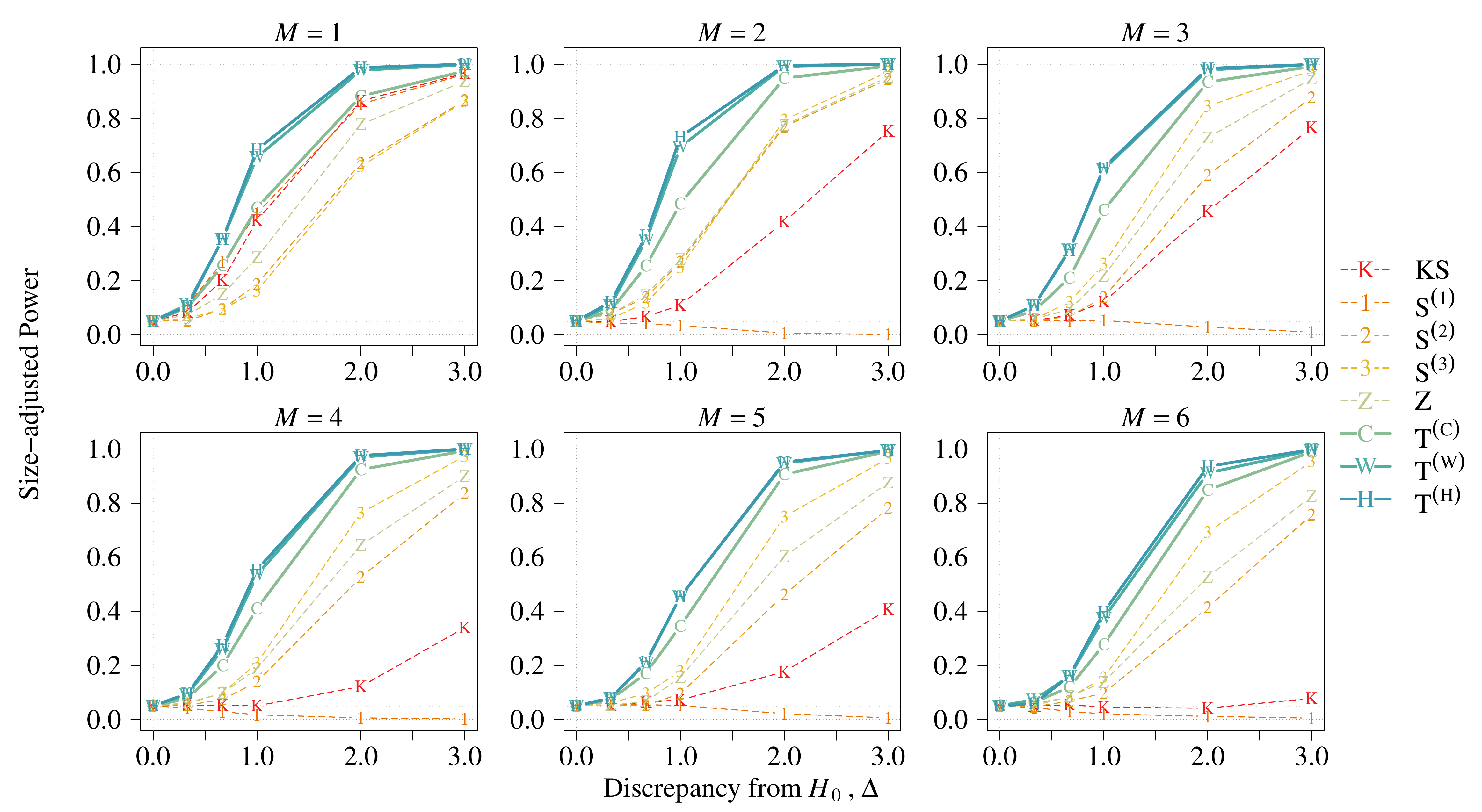}
\vspace{-0.3cm}
\caption{\footnotesize \label{fig: ar05g05 adjusted power non-G t dist} Size-adjusted power under BAR model with $\varpi = \vartheta = 0.5$, mean function \eqref{EQ mean function normal} and innovations following $t$-distribution.}
\end{figure}

\begin{figure}[H]
\centering
\includegraphics[width=\linewidth]{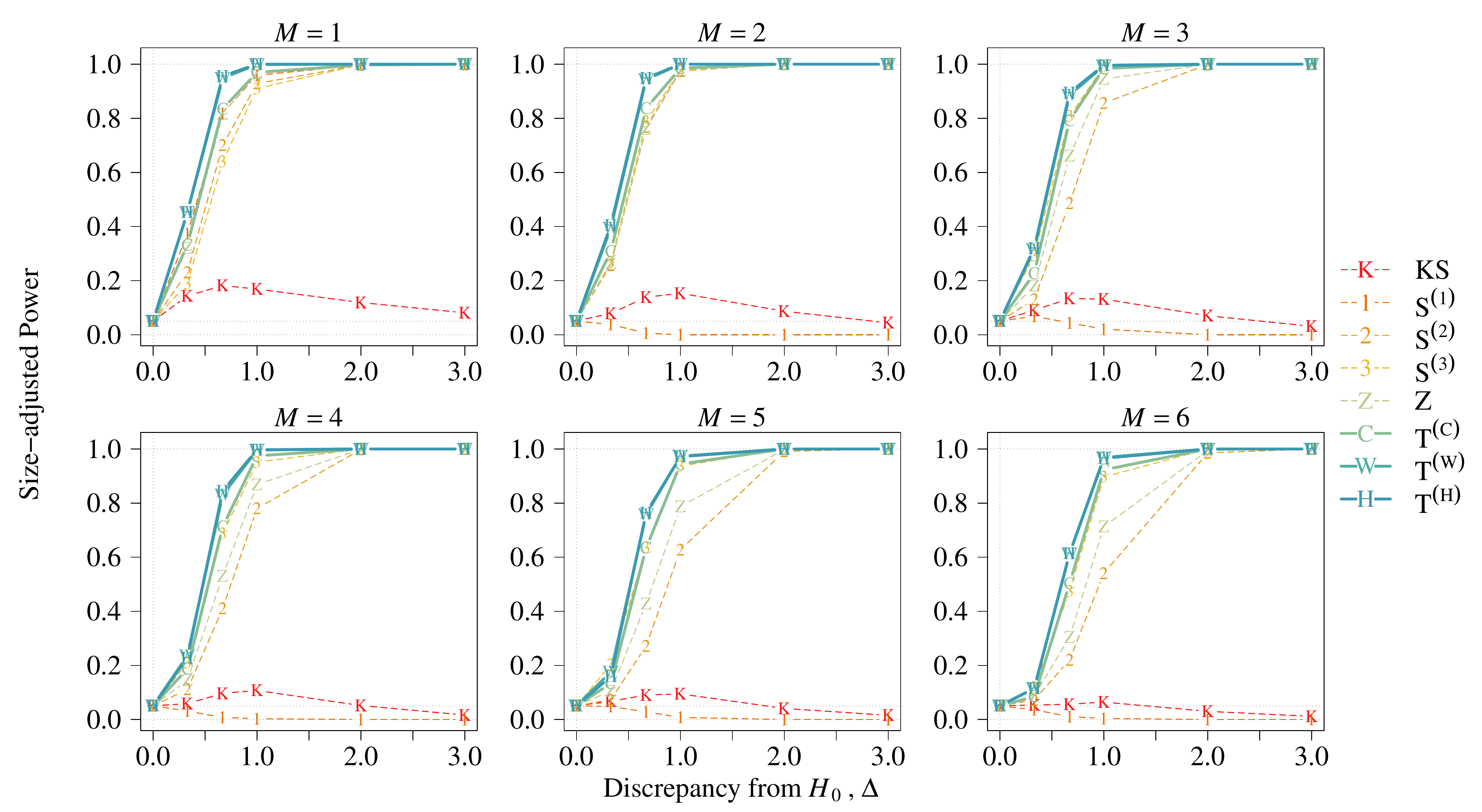}
\vspace{-0.3cm}
\caption{\footnotesize \label{fig: ar-05g05 adjusted powernon-G t dist} Size-adjusted power under BAR model with $-\varpi = \vartheta = 0.5$, mean function \eqref{EQ mean function normal} and innovations following $t$-distribution.}
\end{figure}

\subsection{Unequal change magnitudes}
\label{sec: Unequal change}
This section extends the  simulation results of Section \ref{sec:sim_allCPsEffect} in the main text. The CP models, Cases 1 and 2, are defined in Section \ref{sec:sim_allCPsEffect} in the main text. Figures \ref{fig: S3CP ar0.5}--\ref{fig: S3CP ar-0.5} and \ref{fig: S3CP ar-0.5g0.5}--\ref{fig: S3CP ar-0.5g-0.5} present the size-adjusted power under the AR and BAR models respectively. The result indicates that \cite{zhang2018unsupervised}'s proposed test $\zhang_n$ suffers larger reduction in power comparing to our proposed tests when the first and the last changes are reduced by half. Therefore, our proposed tests are more general and capable to capture all change structure, while $\zhang_n$ tends to capture the first and the last CPs.

\begin{figure}[H]
\centering
\includegraphics[width=0.7\linewidth]{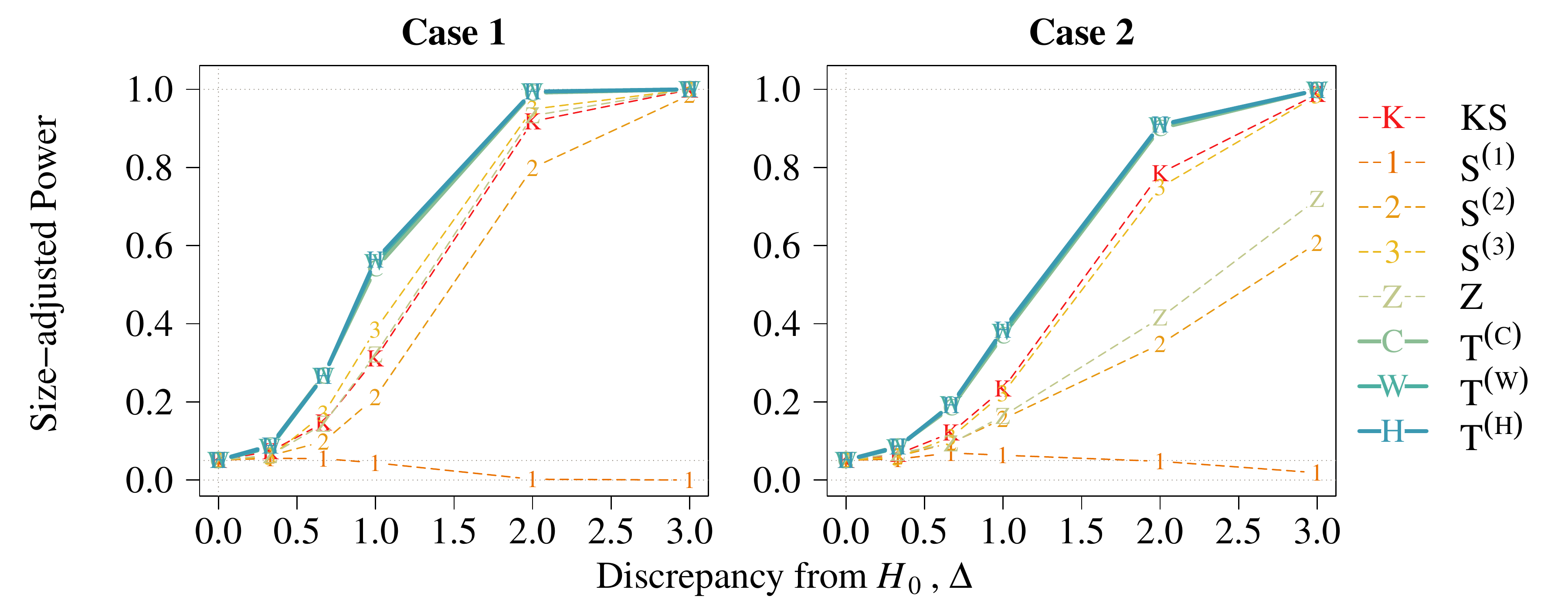}
\vspace{-0.3cm}
\caption{\footnotesize  \label{fig: S3CP ar0.5} Size-adjusted power under AR model with $\varpi =0.5$, $n=200$ and mean functions in Cases 1 and 2.}
\end{figure}

\begin{figure}[H]
\centering
\includegraphics[width=0.7\linewidth]{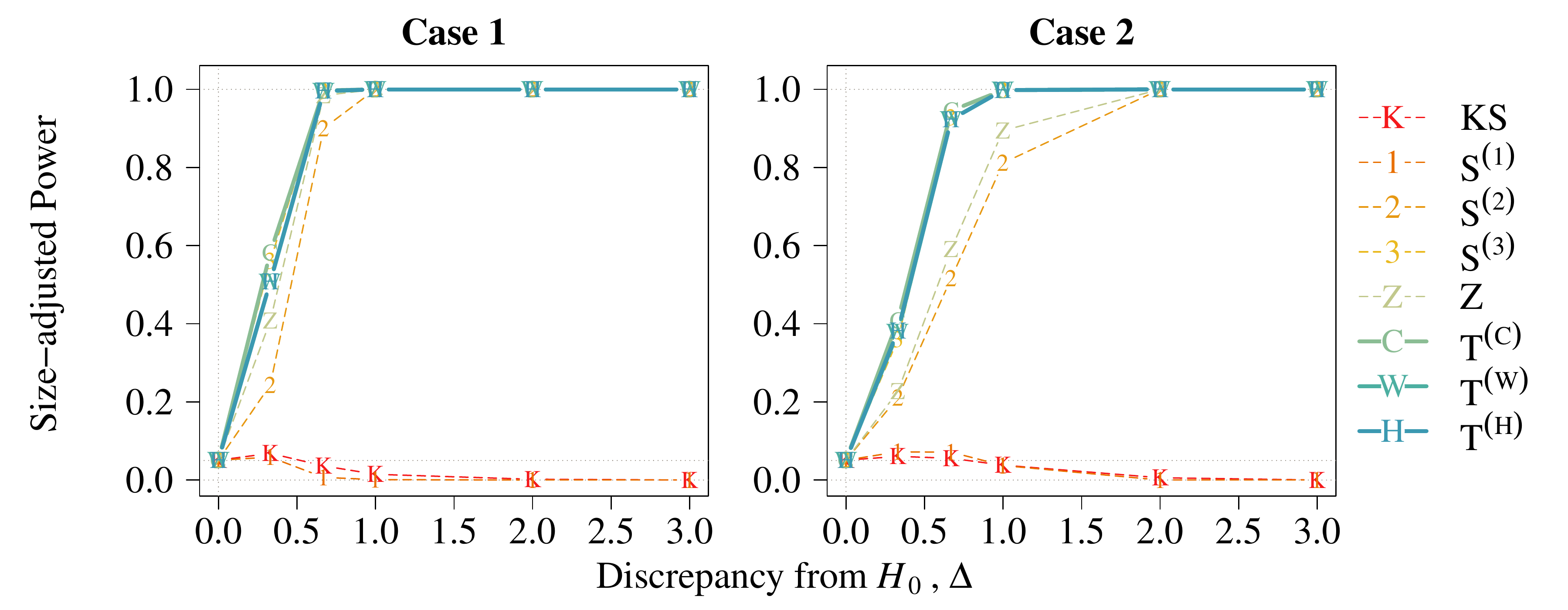}
\vspace{-0.3cm}
\caption{\footnotesize  \label{fig: S3CP ar-0.5} Size-adjusted power under AR model with $\varpi = -0.5$, $n=200$ and mean functions in Cases 1 and 2.}
\end{figure}

\begin{figure}[H]
\centering
\includegraphics[width=0.7\linewidth]{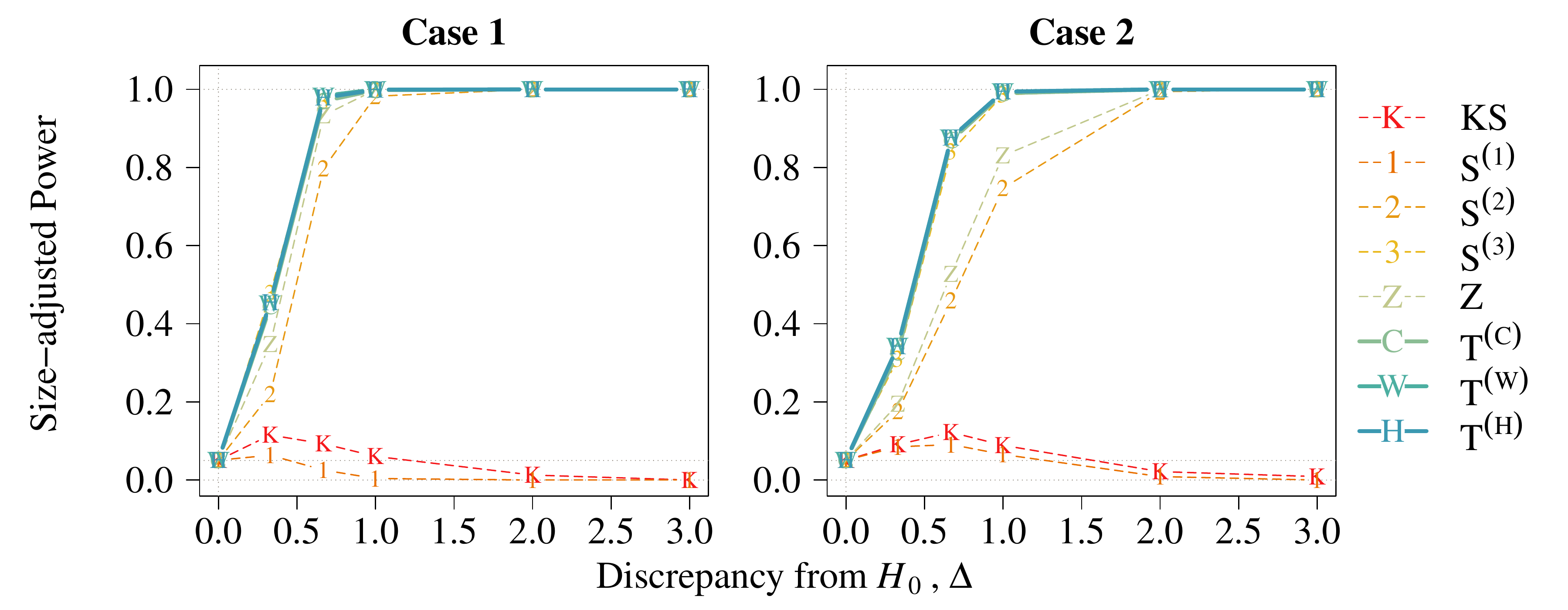}
\vspace{-0.3cm}
\caption{\footnotesize  \label{fig: S3CP ar-0.5g0.5} Size-adjusted power under BAR model with $-\varpi =\vartheta = 0.5$ , $n=200$ and mean functions in Cases 1 and 2.}
\end{figure}

\begin{figure}[H]
\centering
\includegraphics[width=0.7\linewidth]{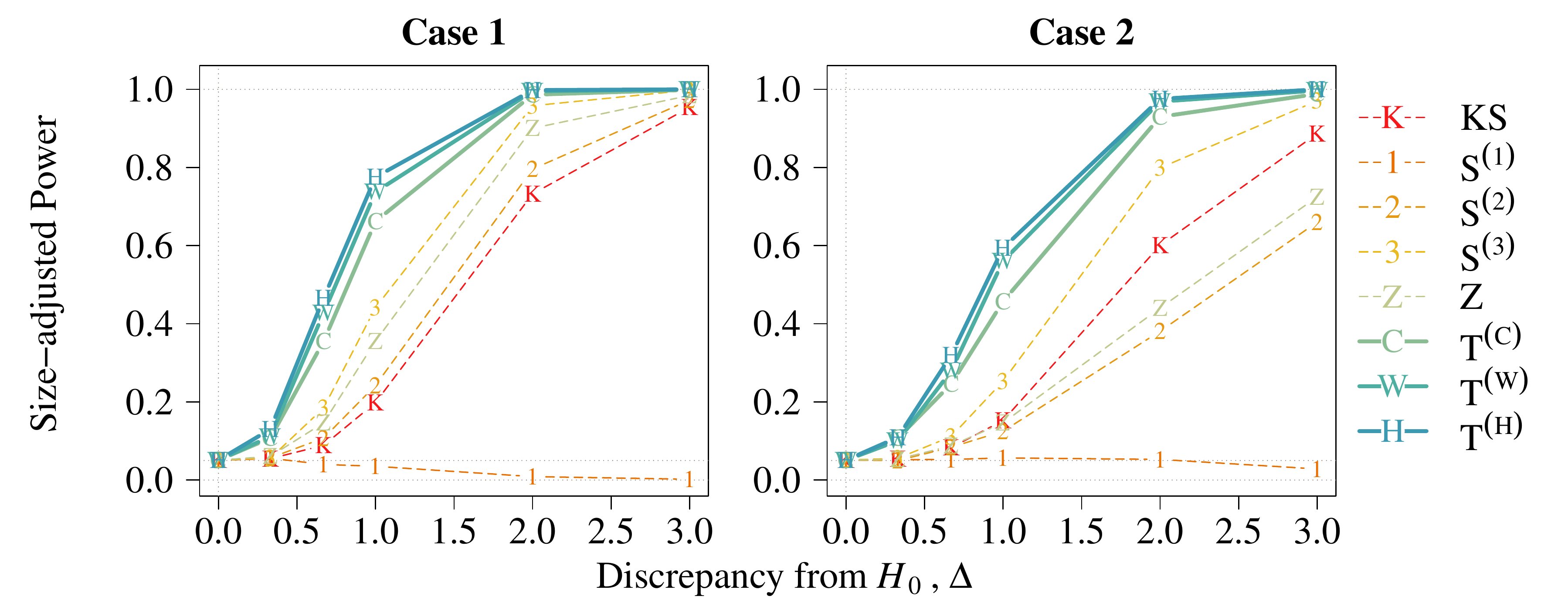}
\caption{\footnotesize  \label{fig: S3CP ar0.5g-0.5} Size-adjusted power under BAR model with $\varpi =-\vartheta = 0.5$ , $n=200$ and mean functions in Cases 1 and 2.}
\end{figure}

\begin{figure}[H]
\centering
\includegraphics[width=0.7\linewidth]{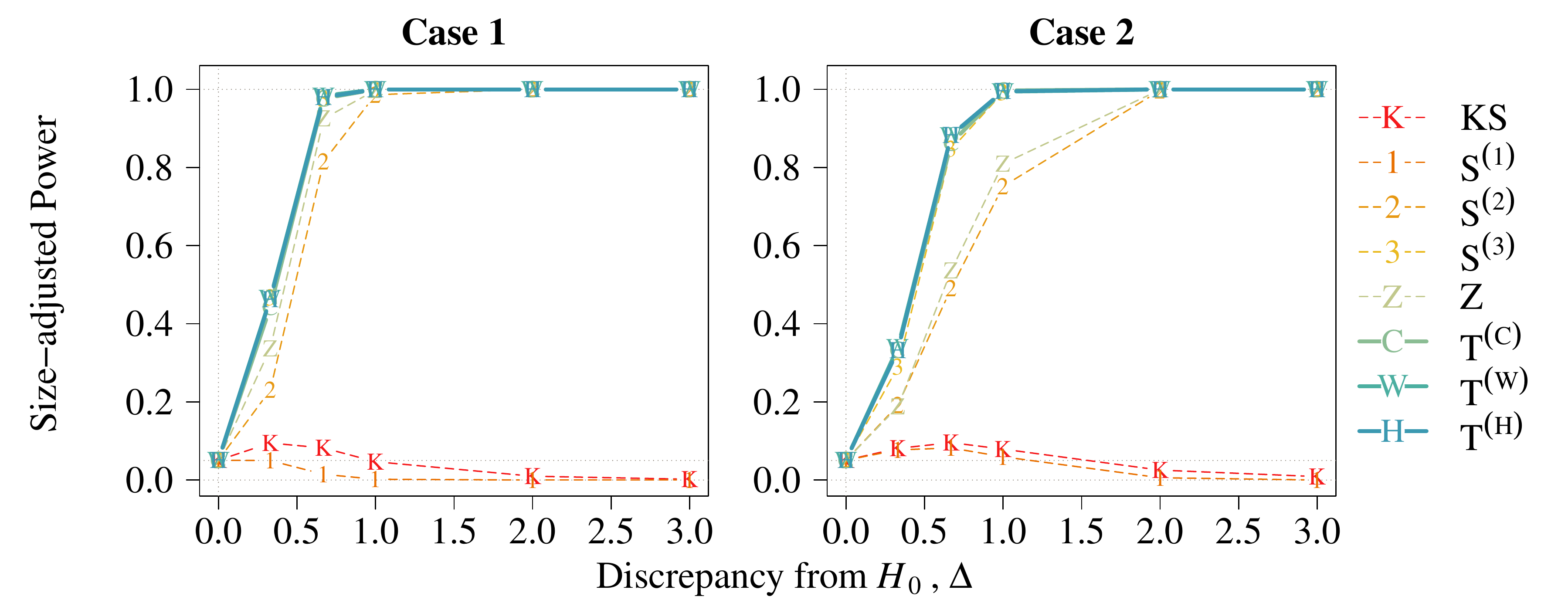}
\vspace{-0.3cm}
\caption{\footnotesize \label{fig: S3CP ar-0.5g-0.5} Size-adjusted power under BAR model with $\varpi =\vartheta = -0.5$ , $n=200$ and mean functions in Cases 1 and 2.}
\end{figure}

\subsection{Nonsymmetric windows}
\label{sec: Non-symmetric windows}
This section presents simulation result supplementing Section \ref{sec: justify symm windows} in the main text. The mean functions, Cases 1 and 3 are defined in Section \ref{sec:sim_allCPsEffect} and \ref{sec: justify symm windows} in the main text respectively. The size-adjusted power is presented in Figures \ref{fig: compare nonsym window, ar0.5}--\ref{fig: compare nonsym window, ar-0.5} under the AR model and in Figure \ref{fig: compare nonsym window, ar-0.5g0.5} under the BAR model when $-\varpi =\vartheta = 0.5$. From the result, we do not observe significant difference in power after nonsymmetric windows are used. Therefore, using nonsymmetric windows does not give significant improvement in power in general. 

\begin{figure}[H]
\centering
\includegraphics[width=0.7\linewidth]{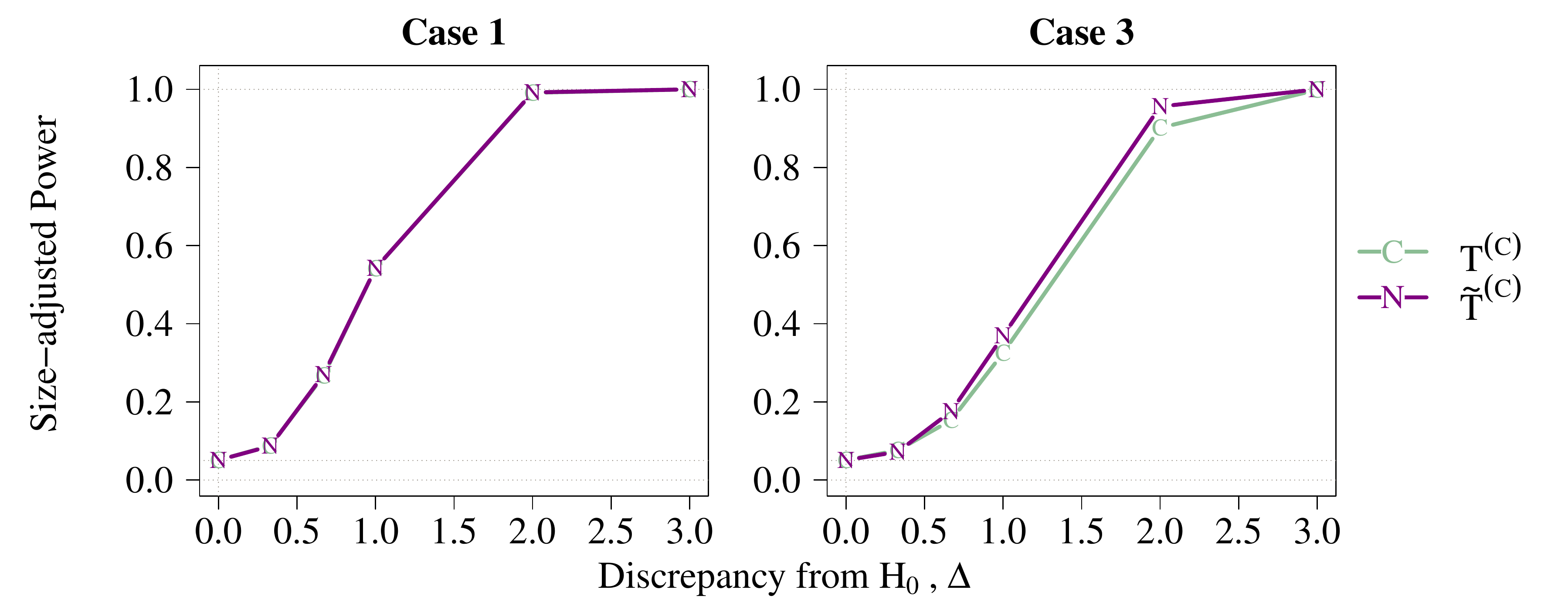}
\vspace{-0.3cm}
\caption{\footnotesize  \label{fig: compare nonsym window, ar0.5} Size-adjusted power under AR model with $\varpi = 0.5$ and mean functions in Cases 1 and 3.}
\end{figure}

\begin{figure}[H]
\centering
\includegraphics[width=0.7\linewidth]{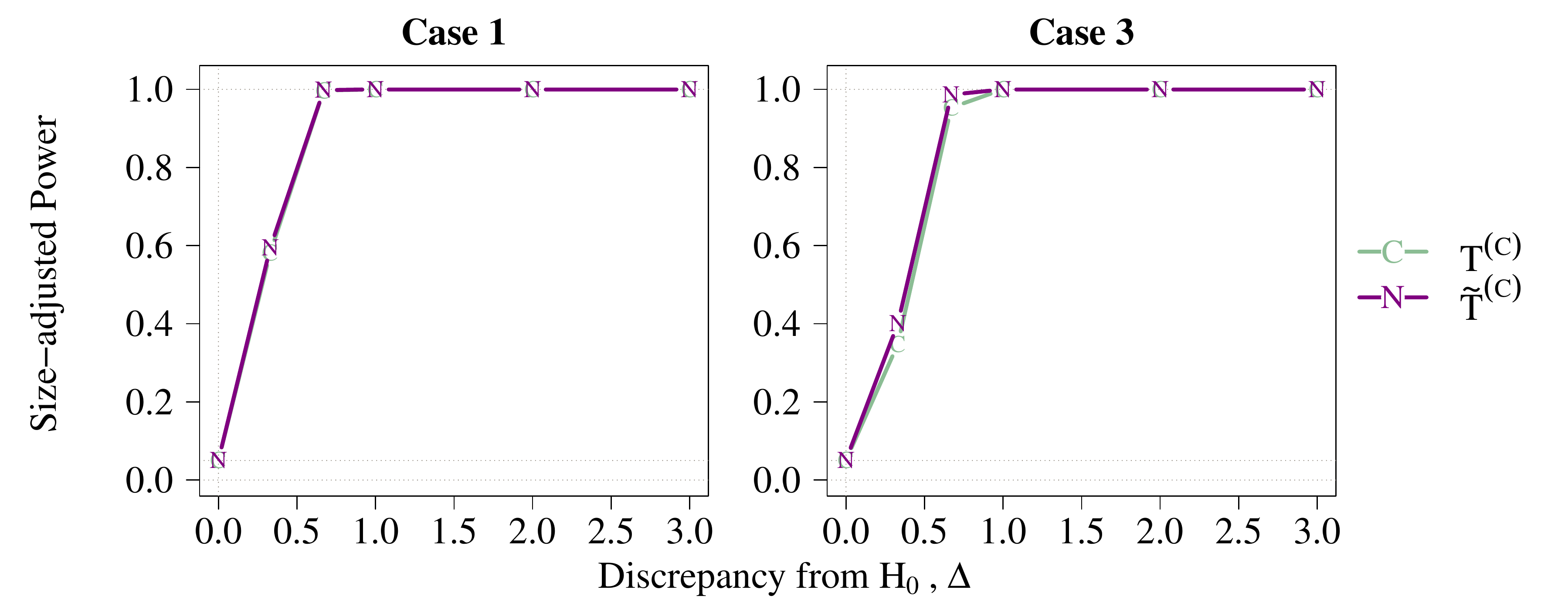}
\vspace{-0.3cm}
\caption{\footnotesize  \label{fig: compare nonsym window, ar-0.5} Size-adjusted power under AR model with $\varpi = -0.5$ and mean functions in Cases 1 and 3.}
\end{figure}

\begin{figure}[H]
\centering
\includegraphics[width=0.7\linewidth]{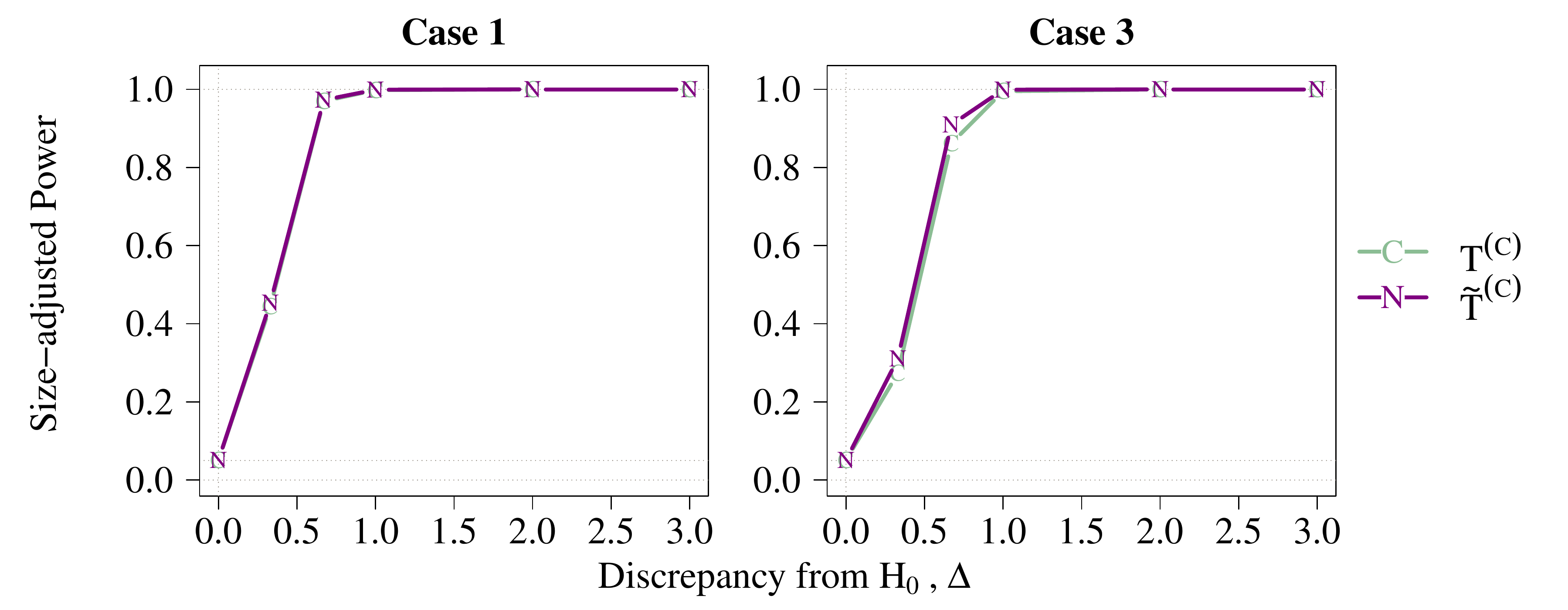}
\vspace{-0.3cm}
\caption{\footnotesize  \label{fig: compare nonsym window, ar-0.5g0.5} Size-adjusted power under BAR model with $-\varpi =\vartheta = 0.5$ and mean functions in Cases 1 and 3.}
\end{figure}

\subsection{Conclusion}
Our proposed tests in general control the size well. We also observed that non-self-normalized KS test has larger size distortion than self-normalized tests. The results agree to \cite{kiefer2002heteroskedasticity} and \cite{shao2015self}. Our proposed tests also have larger power comparing to the state-of-art self-normalized multiple-CP tests, $\shao_n^{(2)}$, $\shao_n^{(3)}$ and $\zhang_n$. Under the BAR model, our proposed tests even outperform the AMOC tests, $\shao_n^{(1)}$ and $\KS_n$, in the single CP case. Moreover, the oracle tests $\shao_n^{(2)}$ and $\shao_n^{(3)}$ have smaller power comparing to our proposed tests even when the number of CP is correctly specified. It demonstrates that knowing the exact number of CPs does not give a significant advantage in CP detection.

\section{Finite-sample adjusted critical value}\label{sec:finite_n_cVal}

The finite-$n$ adjusted critical values $c_{\alpha}(n, \rho)$ 
for $\T_n^{(\aleph)}$ at size $\alpha\in\{0.1, 0.05, 0.01\}$ are presented in 
Tables \ref{tab: 90 Critical Values}, \ref{tab: 95 Critical Values} and \ref{tab: 99 Critical Values}, respectively, 
where $n\in\{100, 200, \ldots, 1000, 2000, \ldots, 10000\}$ and $\rho\in\{0, \pm 0.1, \ldots, \pm 0.9\}$. 
Here the critical values are valid for $\aleph \in \{\C, \W,\H, \G \}$. 
The critical values $\mathit{c}_{\alpha}(n, \rho)$'s are computed using 
$200000$ replications under autoregressive AR(1) models for each 
sample sizes $n = 100, \ldots, 1000$ and AR parameter
$\rho \in \{0, \pm 0.1, \ldots, \pm0.9\}$, with standard normal innovations.

\begin{table}[t]
\setlength{\tabcolsep}{2pt}
\centering
\footnotesize
\caption{\label{tab: 90 Critical Values} 
The finite-$n$ adjusted critical values 
$c_{\alpha}(n, \rho)$ of $\T_n^{(\aleph)}$ under $\alpha = 10\%$ nominal size.}
\begin{tabular}{r rrrrrrrrr r rrrrrrrrr}
\toprule
$n\setminus \rho$ & $-0.9$ & $-0.8$ & $-0.7$ & $-0.6$ & $-0.5$ & $-0.4$ & $-0.3$ & $-0.2$ & $-0.1$ & 0 & 0.1 & 0.2 & 0.3 & 0.4 & 0.5 & 0.6 & 0.7 & 0.8 & 0.9\\[0.5ex]
\cmidrule(r){1-20}  & \\[-1.5ex]
100 & 6.8 & 8.4 & 9.6 & 10.7 & 11.5 & 12.4 & 13.1 & 13.9 & 14.7 & 15.5 & 16.5 & 17.6 & 18.9 & 20.6 & 22.9 & 26.2 & 30.9 & 38.1 & 48.0\\
200 & 8.5 & 10.5 & 11.8 & 12.8 & 13.5 & 14.1 & 14.7 & 15.1 & 15.6 & 16.1 & 16.6 & 17.2 & 17.9 & 18.8 & 20.0 & 21.7 & 24.7 & 30.0 & 41.2\\
300 & 9.8 & 11.9 & 13.1 & 13.9 & 14.5 & 15.0 & 15.4 & 15.8 & 16.1 & 16.4 & 16.8 & 17.1 & 17.6 & 18.2 & 18.9 & 20.1 & 22.0 & 25.9 & 35.6\\
400 & 10.7 & 12.8 & 13.9 & 14.6 & 15.2 & 15.6 & 15.9 & 16.1 & 16.4 & 16.6 & 16.9 & 17.1 & 17.4 & 17.8 & 18.4 & 19.2 & 20.6 & 23.5 & 31.6\\
500 & 11.4 & 13.5 & 14.5 & 15.1 & 15.6 & 15.9 & 16.1 & 16.3 & 16.5 & 16.7 & 16.9 & 17.1 & 17.3 & 17.6 & 18.0 & 18.7 & 19.8 & 22.0 & 28.8\\
600 & 12.0 & 14.0 & 14.9 & 15.5 & 15.9 & 16.1 & 16.4 & 16.5 & 16.7 & 16.8 & 17.0 & 17.1 & 17.3 & 17.5 & 17.9 & 18.4 & 19.2 & 21.1 & 26.8\\
700 & 12.5 & 14.4 & 15.3 & 15.8 & 16.1 & 16.4 & 16.5 & 16.7 & 16.8 & 16.9 & 17.0 & 17.1 & 17.3 & 17.5 & 17.7 & 18.1 & 18.9 & 20.4 & 25.3\\
800 & 13.0 & 14.8 & 15.6 & 16.0 & 16.3 & 16.5 & 16.6 & 16.8 & 16.9 & 17.0 & 17.1 & 17.2 & 17.3 & 17.4 & 17.6 & 18.0 & 18.6 & 19.9 & 24.2\\
900 & 13.3 & 15.0 & 15.8 & 16.2 & 16.5 & 16.6 & 16.8 & 16.9 & 17.0 & 17.0 & 17.1 & 17.2 & 17.3 & 17.4 & 17.6 & 17.9 & 18.4 & 19.5 & 23.3\\
1000 & 13.6 & 15.3 & 16.0 & 16.3 & 16.6 & 16.7 & 16.8 & 16.9 & 17.0 & 17.1 & 17.1 & 17.2 & 17.3 & 17.4 & 17.6 & 17.8 & 18.3 & 19.3 & 22.6\\
2000 & 16.7 & 16.9 & 17.0 & 17.1 & 17.2 & 17.2 & 17.3 & 17.2 & 17.3 & 17.3 & 17.3 & 17.3 & 17.3 & 17.3 & 17.3 & 17.3 & 17.4 & 17.7 & 18.6\\
3000 & 16.9 & 17.2 & 17.3 & 17.3 & 17.4 & 17.4 & 17.4 & 17.4 & 17.4 & 17.4 & 17.4 & 17.4 & 17.4 & 17.4 & 17.4 & 17.4 & 17.4 & 17.5 & 18.0\\
4000 & 17.1 & 17.3 & 17.4 & 17.5 & 17.5 & 17.5 & 17.5 & 17.5 & 17.5 & 17.5 & 17.4 & 17.4 & 17.4 & 17.4 & 17.4 & 17.5 & 17.5 & 17.5 & 17.8\\
5000 & 17.2 & 17.4 & 17.5 & 17.5 & 17.6 & 17.6 & 17.5 & 17.5 & 17.5 & 17.5 & 17.5 & 17.5 & 17.5 & 17.5 & 17.5 & 17.5 & 17.4 & 17.5 & 17.6\\
6000 & 17.3 & 17.5 & 17.6 & 17.6 & 17.6 & 17.6 & 17.6 & 17.6 & 17.6 & 17.5 & 17.5 & 17.5 & 17.5 & 17.5 & 17.5 & 17.5 & 17.5 & 17.5 & 17.6\\
7000 & 17.4 & 17.6 & 17.6 & 17.6 & 17.6 & 17.6 & 17.6 & 17.6 & 17.6 & 17.6 & 17.5 & 17.5 & 17.5 & 17.5 & 17.5 & 17.5 & 17.5 & 17.5 & 17.6\\
8000 & 17.5 & 17.6 & 17.6 & 17.7 & 17.6 & 17.6 & 17.6 & 17.6 & 17.6 & 17.6 & 17.6 & 17.6 & 17.6 & 17.5 & 17.5 & 17.5 & 17.5 & 17.5 & 17.5\\
9000 & 17.5 & 17.6 & 17.6 & 17.7 & 17.7 & 17.6 & 17.6 & 17.6 & 17.6 & 17.6 & 17.6 & 17.6 & 17.6 & 17.5 & 17.5 & 17.5 & 17.5 & 17.5 & 17.5\\
10000 & 17.6 & 17.7 & 17.7 & 17.8 & 17.7 & 17.7 & 17.7 & 17.7 & 17.6 & 17.6 & 17.6 & 17.6 & 17.5 & 17.5 & 17.5 & 17.4 & 17.4 & 17.4 & 17.5\\

\bottomrule
\end{tabular} 
\end{table}

\begin{table}[t]
\setlength{\tabcolsep}{2pt}
\centering
\footnotesize
\caption{\footnotesize \label{tab: 95 Critical Values} 
The finite-$n$ adjusted critical values 
$c_{\alpha}(n, \rho)$ of $\T_n^{(\aleph)}$ under $\alpha = 5\%$ nominal size.}
\vspace{-0.3cm}
\begin{tabular}{r rrrrrrrrr r rrrrrrrrr}
\toprule
$n\setminus \rho$ & $-0.9$ & $-0.8$ & $-0.7$ & $-0.6$ & $-0.5$ & $-0.4$ & $-0.3$ & $-0.2$ & $-0.1$ & 0 & 0.1 & 0.2 & 0.3 & 0.4 & 0.5 & 0.6 & 0.7 & 0.8 & 0.9\\
\cmidrule(r){1-20}

100 & 7.6 & 9.4 & 10.8 & 11.9 & 12.9 & 13.9 & 14.7 & 15.6 & 16.5 & 17.5 & 18.5 & 19.8 & 21.3 & 23.3 & 25.9 & 29.5 & 34.8 & 42.5 & 52.9\\
200 & 9.5 & 11.8 & 13.2 & 14.2 & 15.1 & 15.8 & 16.3 & 16.9 & 17.4 & 18.0 & 18.6 & 19.2 & 20.0 & 21.0 & 22.4 & 24.4 & 27.7 & 33.7 & 45.7\\
300 & 10.9 & 13.2 & 14.6 & 15.5 & 16.2 & 16.7 & 17.2 & 17.6 & 18.0 & 18.3 & 18.7 & 19.2 & 19.7 & 20.3 & 21.2 & 22.5 & 24.7 & 29.1 & 39.8\\
400 & 11.9 & 14.2 & 15.5 & 16.3 & 16.9 & 17.3 & 17.7 & 18.0 & 18.3 & 18.5 & 18.8 & 19.1 & 19.5 & 19.9 & 20.6 & 21.5 & 23.2 & 26.4 & 35.5\\
500 & 12.7 & 15.0 & 16.1 & 16.8 & 17.3 & 17.7 & 18.0 & 18.2 & 18.5 & 18.7 & 18.9 & 19.1 & 19.4 & 19.7 & 20.2 & 20.9 & 22.2 & 24.8 & 32.4\\
600 & 13.4 & 15.6 & 16.6 & 17.3 & 17.7 & 18.0 & 18.2 & 18.4 & 18.6 & 18.8 & 18.9 & 19.1 & 19.3 & 19.6 & 20.0 & 20.6 & 21.5 & 23.7 & 30.1\\
700 & 13.9 & 16.0 & 17.0 & 17.5 & 17.9 & 18.2 & 18.4 & 18.6 & 18.7 & 18.9 & 19.0 & 19.1 & 19.3 & 19.5 & 19.8 & 20.3 & 21.1 & 22.9 & 28.5\\
800 & 14.4 & 16.4 & 17.3 & 17.8 & 18.1 & 18.4 & 18.5 & 18.7 & 18.8 & 18.9 & 19.0 & 19.1 & 19.3 & 19.5 & 19.7 & 20.1 & 20.8 & 22.3 & 27.2\\
900 & 14.8 & 16.7 & 17.6 & 18.0 & 18.3 & 18.5 & 18.7 & 18.8 & 18.9 & 19.0 & 19.1 & 19.2 & 19.3 & 19.4 & 19.6 & 20.0 & 20.6 & 21.9 & 26.2\\
1000 & 15.1 & 17.0 & 17.7 & 18.2 & 18.4 & 18.6 & 18.7 & 18.8 & 18.9 & 19.0 & 19.1 & 19.1 & 19.2 & 19.4 & 19.6 & 19.8 & 20.4 & 21.5 & 25.3\\
2000 & 18.8 & 18.8 & 19.0 & 19.1 & 19.1 & 19.1 & 19.2 & 19.2 & 19.3 & 19.3 & 19.3 & 19.2 & 19.3 & 19.3 & 19.3 & 19.3 & 19.4 & 19.7 & 20.9\\
3000 & 18.9 & 19.1 & 19.2 & 19.3 & 19.3 & 19.3 & 19.4 & 19.3 & 19.4 & 19.4 & 19.4 & 19.4 & 19.3 & 19.4 & 19.3 & 19.4 & 19.4 & 19.5 & 20.1\\
4000 & 19.0 & 19.2 & 19.3 & 19.4 & 19.4 & 19.5 & 19.4 & 19.4 & 19.5 & 19.4 & 19.4 & 19.4 & 19.4 & 19.4 & 19.4 & 19.4 & 19.4 & 19.5 & 19.8\\
5000 & 19.2 & 19.3 & 19.5 & 19.5 & 19.5 & 19.5 & 19.5 & 19.5 & 19.5 & 19.5 & 19.5 & 19.5 & 19.5 & 19.5 & 19.4 & 19.4 & 19.5 & 19.5 & 19.7\\
6000 & 19.3 & 19.4 & 19.5 & 19.5 & 19.6 & 19.6 & 19.5 & 19.6 & 19.5 & 19.5 & 19.5 & 19.5 & 19.5 & 19.5 & 19.4 & 19.4 & 19.5 & 19.5 & 19.6\\
7000 & 19.4 & 19.5 & 19.5 & 19.6 & 19.6 & 19.6 & 19.6 & 19.6 & 19.6 & 19.5 & 19.5 & 19.5 & 19.5 & 19.5 & 19.5 & 19.4 & 19.5 & 19.5 & 19.6\\
8000 & 19.5 & 19.6 & 19.6 & 19.6 & 19.6 & 19.6 & 19.6 & 19.6 & 19.6 & 19.5 & 19.5 & 19.5 & 19.5 & 19.5 & 19.5 & 19.5 & 19.5 & 19.5 & 19.5\\
9000 & 19.5 & 19.6 & 19.6 & 19.6 & 19.6 & 19.6 & 19.6 & 19.6 & 19.6 & 19.5 & 19.5 & 19.5 & 19.5 & 19.5 & 19.5 & 19.4 & 19.5 & 19.4 & 19.5\\
10000 & 19.5 & 19.7 & 19.7 & 19.7 & 19.7 & 19.6 & 19.6 & 19.6 & 19.6 & 19.6 & 19.5 & 19.5 & 19.5 & 19.6 & 19.4 & 19.4 & 19.4 & 19.4 & 19.5\\

\bottomrule
\end{tabular} 
\end{table}

\begin{table}[t]
\setlength{\tabcolsep}{2pt}
\centering
\footnotesize
\caption{\label{tab: 99 Critical Values} 
The finite-$n$ adjusted critical values 
$c_{\alpha}(n, \rho)$ of $\T_n^{(\aleph)}$ under $\alpha = 1\%$ nominal size.}
\begin{tabular}{r rrrrrrrrr r rrrrrrrrr}
\toprule
$n\setminus \rho$ & $-0.9$ & $-0.8$ & $-0.7$ & $-0.6$ & $-0.5$ & $-0.4$ & $-0.3$ & $-0.2$ & $-0.1$ & 0 & 0.1 & 0.2 & 0.3 & 0.4 & 0.5 & 0.6 & 0.7 & 0.8 & 0.9\\
\cmidrule(r){1-20}  
100 & 9.5 & 11.6 & 13.3 & 14.7 & 15.9 & 17.0 & 18.1 & 19.2 & 20.3 & 21.5 & 22.9 & 24.5 & 26.4 & 28.8 & 32.1 & 36.6 & 43.0 & 51.9 & 63.8\\
200 & 11.7 & 14.4 & 16.1 & 17.4 & 18.4 & 19.2 & 20.0 & 20.7 & 21.4 & 22.1 & 22.8 & 23.6 & 24.6 & 25.9 & 27.7 & 30.2 & 34.3 & 41.4 & 55.5\\
300 & 13.3 & 16.1 & 17.8 & 18.8 & 19.7 & 20.4 & 20.9 & 21.4 & 21.9 & 22.4 & 22.9 & 23.4 & 24.1 & 24.9 & 26.0 & 27.7 & 30.5 & 36.0 & 49.1\\
400 & 14.6 & 17.4 & 18.9 & 19.9 & 20.6 & 21.1 & 21.6 & 22.0 & 22.3 & 22.7 & 23.0 & 23.4 & 23.9 & 24.4 & 25.3 & 26.5 & 28.5 & 32.7 & 43.9\\
500 & 15.5 & 18.3 & 19.7 & 20.5 & 21.2 & 21.6 & 22.0 & 22.3 & 22.6 & 22.8 & 23.1 & 23.4 & 23.8 & 24.3 & 24.9 & 25.8 & 27.4 & 30.7 & 40.1\\
600 & 16.3 & 18.9 & 20.2 & 21.0 & 21.5 & 21.9 & 22.2 & 22.4 & 22.7 & 22.9 & 23.1 & 23.3 & 23.6 & 24.0 & 24.5 & 25.3 & 26.5 & 29.2 & 37.4\\
700 & 16.9 & 19.4 & 20.6 & 21.3 & 21.7 & 22.0 & 22.3 & 22.5 & 22.7 & 22.9 & 23.1 & 23.3 & 23.5 & 23.8 & 24.2 & 24.8 & 25.9 & 28.1 & 35.2\\
800 & 17.5 & 19.9 & 21.0 & 21.6 & 22.0 & 22.3 & 22.5 & 22.7 & 22.9 & 23.0 & 23.2 & 23.3 & 23.5 & 23.7 & 24.1 & 24.6 & 25.5 & 27.5 & 33.6\\
900 & 18.0 & 20.2 & 21.3 & 21.8 & 22.2 & 22.5 & 22.7 & 22.8 & 22.9 & 23.1 & 23.2 & 23.3 & 23.5 & 23.7 & 24.0 & 24.4 & 25.3 & 26.9 & 32.4\\
1000 & 18.4 & 20.5 & 21.5 & 22.0 & 22.3 & 22.6 & 22.7 & 22.9 & 23.0 & 23.1 & 23.2 & 23.3 & 23.5 & 23.6 & 23.9 & 24.3 & 25.0 & 26.5 & 31.4\\
2000 & 23.2 & 22.9 & 23.0 & 23.1 & 23.3 & 23.3 & 23.3 & 23.2 & 23.4 & 23.4 & 23.4 & 23.4 & 23.4 & 23.3 & 23.4 & 23.6 & 23.6 & 24.1 & 25.9\\
3000 & 23.1 & 23.1 & 23.3 & 23.4 & 23.4 & 23.4 & 23.4 & 23.4 & 23.5 & 23.5 & 23.5 & 23.5 & 23.5 & 23.5 & 23.5 & 23.6 & 23.6 & 23.8 & 24.7\\
4000 & 23.2 & 23.4 & 23.5 & 23.6 & 23.6 & 23.5 & 23.6 & 23.5 & 23.6 & 23.6 & 23.5 & 23.6 & 23.7 & 23.6 & 23.5 & 23.6 & 23.6 & 23.7 & 24.3\\
5000 & 23.2 & 23.4 & 23.5 & 23.6 & 23.6 & 23.7 & 23.6 & 23.6 & 23.6 & 23.6 & 23.5 & 23.6 & 23.6 & 23.6 & 23.5 & 23.5 & 23.6 & 23.7 & 24.0\\
6000 & 23.3 & 23.5 & 23.6 & 23.7 & 23.7 & 23.8 & 23.6 & 23.6 & 23.8 & 23.6 & 23.6 & 23.6 & 23.6 & 23.7 & 23.6 & 23.6 & 23.6 & 23.6 & 23.9\\
7000 & 23.5 & 23.6 & 23.6 & 23.8 & 23.7 & 23.7 & 23.7 & 23.8 & 23.7 & 23.6 & 23.7 & 23.7 & 23.7 & 23.6 & 23.6 & 23.6 & 23.6 & 23.6 & 23.8\\
8000 & 23.7 & 23.7 & 23.7 & 23.9 & 23.8 & 23.7 & 23.8 & 23.8 & 23.8 & 23.6 & 23.7 & 23.8 & 23.8 & 23.7 & 23.6 & 23.6 & 23.7 & 23.7 & 23.8\\
9000 & 23.7 & 23.8 & 23.7 & 23.9 & 23.8 & 23.7 & 23.8 & 23.8 & 23.8 & 23.6 & 23.7 & 23.8 & 23.8 & 23.7 & 23.6 & 23.6 & 23.6 & 23.6 & 23.8\\
10000 & 23.7 & 23.9 & 23.8 & 24.0 & 23.9 & 23.9 & 23.8 & 23.9 & 23.8 & 23.7 & 23.7 & 23.8 & 23.7 & 23.6 & 23.7 & 23.6 & 23.6 & 23.6 & 23.8\\

\bottomrule
\end{tabular} 
\end{table}

\section{Extension to point estimation}\label{sec: pt estimation}
Our proposed method can be extended to CP estimation. The first-step estimator for the set of CPs are defined in \eqref{EQ: cp set pt estimate}.
If $H_0$ is rejected, the proposed score function used in the calculation can be reused and integrated to the screening and ranking algorithm \citep{niu2012screening} to obtain CP estimates;see Algorithm \ref{algo:SARA}, where $ h = \lfloor \epsilon n \rfloor$ and the criterion function $\mathcal{C}$ can be maximum likelihood \citep{yao1987approximating}, least-squares estimation \citep{lavielle2000least}, and minimum description length \citep{davis2006structural}, etc.

\begin{algorithm}[t]
\caption{SaRa for CPs location estimation}\label{algo:SARA}
\SetAlgoVlined
\DontPrintSemicolon
\SetNlSty{texttt}{[}{]}
\small
\textbf{Input}: {\;
(i) $X_1, \ldots, X_n$ -- the data set; \;
(ii) $\widehat{\T}_n^{(\C)}(\cdot)$ -- the proposed score function;\; 
(ii) $\mathcal{C}(\cdot)$ -- criterion function to perform CP times selection.\;
}
\Begin{
Compute the initial set of estimated CP locations $\widehat{\Pi}\equiv\{\widehat{k}_1,\ldots, \widehat{k}_{\widehat{m}} \}$ by \eqref{EQ: cp set pt estimate}.\;
Rearrange $\widehat{k}_1,\ldots, \widehat{k}_{\widehat{m}}$ to 
$\widehat{k}_{(1)}, \ldots , \widehat{k}_{(\widehat{m})}$ 
such that $\T_n^{(\C)}(\widehat{k}_{(1)}) \geq \cdots \geq \T_n^{(\C)}(\widehat{k}_{(\widehat{m})})$.\;
Select $\widehat{m}' = \argmin_{1\leq j\leq \widehat{m}} \mathcal{C}(\widehat{k}_{(1)},\ldots,\widehat{k}_{(j)} ).$\;
Set $\widehat{\Pi}' = \{\widehat{k}_{(1)},\ldots,\widehat{k}_{(\widehat{m}')} \}$.\;	
}
\textbf{return} $\widehat{\Pi}'$ -- the selected set of estimated CP locations.
\end{algorithm}

We can also incorporate our CP test and CP location estimator \eqref{EQ: cp set pt estimate} into the popular binary segmentation algorithm \citep{vostrikova1981detecting} 
to produce a potentially better set of CP locations estimator; see Algorithm \ref{algo:BS}. The stopping criteria, that is based on the $p$-values, is also considered by \cite{shaoadaptiveHDsnmct2021}. They integrated their proposed self-normalized test and the wild binary segmentation proposed by \cite{fryzlewicz2014wild}.

\begin{algorithm}[t]
\caption{Binary segmentation}\label{algo:BS}
\SetAlgoVlined
\DontPrintSemicolon
\SetNlSty{texttt}{[}{]}
\small
\textbf{Input}: {\;
(i) $X_1, \ldots, X_n$ -- the data set; \;
(ii) $\widehat{\T}_{n'}^{(\C)}(\cdot \mid \mathcal{X})$ -- the proposed score function when the a dataset $\mathcal{X}$ of size $n'$ is used;\; 
(iv) $\widehat{p}_{n'}(\mathcal{X})$ -- the $p$-value function when the a dataset $\mathcal{X}$ of size $n'$ is used;\;
(iii) $p_0$ -- the $p$-value threshold to decide whether the CPs exist in the sample; \;
(iii) $\lfloor \epsilon n \rfloor$ -- minimum window size.  \;
}
\SetKwFunction{FMain}{getCP}
\SetKwProg{Fn}{Function}{:}{}
\Fn{\FMain{s, e}}{ 
	\If{$\widehat{p}(X_s, \ldots, X_e) < p_0$ and $e-s+1\geq \lfloor \epsilon n \rfloor$}{
		\KwRet the set of the most obvious CP location $\{ \argmax_{s\leq k \leq e} \widehat{\T}_{e-s+1}^{(\C)}(k\mid X_{s:e}) \}. $\;
	}\Else{\KwRet an empty set $\{\}$}
}
\Begin{
Compute $\widehat{\Pi}'' = \widetilde{\Pi} = \texttt{getCP}(1,n)$. \;
\While{$\widetilde{\Pi}\neq \emptyset$}{
	$\widetilde{\Pi}_0 = \emptyset$\;
	\ForEach{$\widetilde{k}\in\widetilde{\Pi}$}{
		Set $\widetilde{k}^- = \max\{\widehat{\Pi}''\cap [1, \widetilde{k}]\}$ and 
		$\widetilde{k}^+ = \min\{\widehat{\Pi}''\cap [\widetilde{k}+1,n]\}$. \;
		Update $\widetilde{\Pi}_0 \gets \widetilde{\Pi}_0\cup \texttt{getCP}(\widetilde{k}^-,\widetilde{k})\cup \texttt{getCP}(\widetilde{k}+1,\widetilde{k}^+)$.\;
	}
	Update $\widetilde{\Pi} \gets \widetilde{\Pi}_0$ and $\widehat{\Pi}'' \gets \widehat{\Pi}''\cup \widetilde{\Pi}_0$.
}	
}
\textbf{return} $\widehat{\Pi}''$ -- the selected set of estimated CP locations.
\end{algorithm}

\section{Recursively computation}
\label{sec: recursive computational}
Simplify $\V_n^{(\Y)}(k\mid k-d, k+1+d) = \V_n^{(\Y)}(k\mid d)$. The self-normalizer (SN) using symmetric window, given any input change detecting process $\{\Y_n(k) \}_{1 \leq k \leq n}$ are defined as,
\begin{equation*}
\V_n^{(\Y)}(k\mid d) = \frac{1}{4(d+1)}\left\{\sum_{j = k-d}^k \L_n^{(\Y)}(k \mid k-d, k)^2 + \sum_{j = k+1}^{k+1+d} \L_n^{(\Y)}(k \mid k+1,k+1+d)^2 \right\}.
\end{equation*}
For a given $k$, we want to recursively calculate the SN for $d > 0$. Define $S_{\Y\Y}(k) = \sum_{i = 1}^k \Y_n(i)^2$, $S_{\Y}(k) = \sum_{i = 1}^k \Y_n(i)$ and $S_{\I\Y}(k) = \sum_{i = 1}^k i\Y_n(i)$. The SN can be decomposed as
\begin{align*}
\V_n^{(\Y)}(k\mid d) &= \frac{n}{4(d+1)^2}\bigg[(d+1)^2\bigg\{S_{\Y\Y}(k+1+d) - S_{\Y\Y}(k-d-1)\bigg\}\\
&\qquad + \frac{d(2d+1)}{6}\bigg\{\Y_n(k-d-1)^2+ \Y_n(k)^2\bigg\}\\
&\qquad+ \frac{d(2d+3)}{6}\bigg\{\Y_n(k)^2 + \Y_n(k+1+d)^2\bigg\}\\
&\qquad +2\bigg\{S_{\I\Y}(k+1+d) - S_{\I\Y}(k)\bigg\}\bigg\{\Y_n(k) - \Y_n(k+1+d)\bigg\}\\
&\qquad +2\bigg\{S_{\I\Y}(k) - S_{\I\Y}(k-d-1)\bigg\}\bigg\{\Y_n(k-d-1) - \Y_n(k)\bigg\}\\
\end{align*}
\begin{align*}
&\qquad +2\bigg\{S_{\Y}(k+1+d) - S_{\Y}(k)\bigg\}\bigg\{k\Y_n(k+1+d) - (k+1+d)\Y_n(k)\bigg\}\\
&\qquad +2\bigg\{S_{\Y}(k) - S_{\Y}(k-d-1)\bigg\}\bigg\{(k-d-1)\Y_n(k) - k\Y_n(k-d-1)\bigg\}\\
&\qquad + \frac{d(d+2)}{3}\Y_n(k)\bigg\{\Y_n(k-d-1) + \Y_n(k+1+d)\bigg\}\bigg].
\end{align*}
Moreover, simplifying $\L_n^{(\Y)}(k \mid k-d, k+1+d) = \L_n^{(\Y)}(k \mid d)$, the squared localized process is defined as
\begin{equation*}
\L_n^{(\Y)}(k \mid d) = \frac{n}{2(d+1)}\left[ \Y_n(k) - \Y_n(k-d-1) - \frac{1}{2}\left\{\Y_n(k+d+1) - \Y_n(k-d-1) \right\}\right]^2.
\end{equation*}
Both $\L_n^{(\Y)}(k \mid d)^2$ and $\V_n^{(\Y)}(k\mid d)$ requires $O(1)$ time to update as $d$ increases and $k$ being fixed. Therefore, the score at time $k$, $\T_n^{(\Y)}(k) = \sup_{\lfloor \epsilon n \rfloor \leq d \leq n} \L_n^{(\Y)}(k \mid d)^2/\V_n^{(\Y)}(k\mid d)$ can be recursively computed and thus requires $O(n)$ computational complexity. Consequently, the aggregation process takes $O(n^2)$ computational complexity.

	\bibliographystyle{rss} 
	{\small
	\setstretch{1.0}
	\bibliography{Bibliography-MM-MC}
	}
\end{document}